\documentclass[aps,pra,groupedaddress,superscriptaddress,twocolumn,notitlepage,bibnotes]{revtex4-2}

\usepackage[latin1]{inputenc}
\usepackage{amsthm}
\usepackage{amssymb}
\usepackage{amsmath}
\usepackage{bbold}
\usepackage{bbm}
\usepackage{nonfloat}
\usepackage[pdftex]{hyperref}
\usepackage{braket}
\usepackage{dsfont}
\usepackage{mathdots}
\usepackage{mathtools}
\usepackage{enumerate}
\usepackage[shortlabels]{enumitem}
\usepackage{csquotes}
\usepackage{stmaryrd}
\usepackage[cal=boondox]{mathalfa}
\usepackage{graphicx}
\usepackage{stackengine}
\usepackage{scalerel}
\usepackage{xr}
\usepackage{array}
\usepackage{makecell}
\newcolumntype{x}[1]{>{\centering\arraybackslash}p{#1}}
\usepackage{tikz}
\usepackage{pgfplots}
\usetikzlibrary{shapes.geometric, shapes.misc, positioning, arrows, arrows.meta, decorations.pathreplacing, decorations.pathmorphing, patterns, angles, quotes, calc}
\usepackage{booktabs}
\usepackage{xfrac}
\usepackage{siunitx}
\usepackage{centernot}
\usepackage{comment}
\usepackage{chngcntr}

\newtheorem{thm}{Theorem}
\newtheorem*{thm*}{Theorem}

\newtheorem*{prop*}{Proposition}
\newtheorem{lemma}[thm]{Lemma}
\newtheorem*{lemma*}{Lemma}
\newtheorem{lemma_app}{Lemma}

\newtheorem*{cor*}{Corollary}
\newtheorem{cj}[thm]{Conjecture}
\newtheorem*{cj*}{Conjecture}

\newtheorem*{Def*}{Definition}

\newtheorem{definition_app}{Definition}
\newtheorem{remark}{Remark}

\makeatletter
\def\thmhead@plain#1#2#3{%
  \thmname{#1}\thmnumber{\@ifnotempty{#1}{ }\@upn{#2}}%
  \thmnote{ {\the\thm@notefont#3}}}
\let\thmhead\thmhead@plain
\makeatother

\theoremstyle{definition}

\newtheorem{ex}[thm]{Example}

\newtheorem{manualthminner}{Theorem}
\newenvironment{manualthm}[1]{%
  \renewcommand\themanualthminner{#1}%
  \manualthminner \it
}{\endmanualthminner}

\newcommand{\bb}{\begin{equation}\begin{aligned}\hspace{0pt}}
\newcommand{\bbb}{\begin{equation*}\begin{aligned}}
\newcommand{\ee}{\end{aligned}\end{equation}}
\newcommand{\eee}{\end{aligned}\end{equation*}}
\newcommand*{\coloneqq}{\mathrel{\vcenter{\baselineskip0.5ex \lineskiplimit0pt \hbox{\scriptsize.}\hbox{\scriptsize.}}} =}

\newcommand{\ketbra}[1]{\ket{#1}\!\!\bra{#1}}
\newcommand{\ketbraa}[2]{\ket{#1}\!\!\bra{#2}}

\newcommand{\ketbrasub}[1]{\ket{#1}\!\bra{#1}}
\newcommand{\ketbraasub}[2]{\ket{#1}\!\bra{#2}}

\newcommand{\R}{\mathds{R}}
\newcommand{\N}{\mathds{N}}

\DeclareMathOperator{\Tr}{Tr}

\DeclareMathAlphabet{\pazocal}{OMS}{zplm}{m}{n}

\DeclareMathOperator{\Id}{Id}

\newcommand{\HH}{\pazocal{H}}

\newcommand{\lsmatrix}{\left(\begin{smallmatrix}}
\newcommand{\rsmatrix}{\end{smallmatrix}\right)}

\stackMath

\stackMath

\makeatletter
\newcommand*\rel@kern[1]{\kern#1\dimexpr\macc@kerna}
\newcommand*\widebar[1]{%
  \begingroup
  \def\mathaccent##1##2{%
    \rel@kern{0.8}%
    \overline{\rel@kern{-0.8}\macc@nucleus\rel@kern{0.2}}%
    \rel@kern{-0.2}%
  }%
  \macc@depth\@ne
  \let\math@bgroup\@empty \let\math@egroup\macc@set@skewchar
  \mathsurround\z@ \frozen@everymath{\mathgroup\macc@group\relax}%
  \macc@set@skewchar\relax
  \let\mathaccentV\macc@nested@a
  \macc@nested@a\relax111{#1}%
  \endgroup
}

\counterwithin*{equation}{part}
\counterwithin*{thm}{part}
\counterwithin*{figure}{part}

\tikzset{meter/.append style={draw, inner sep=10, rectangle, font=\vphantom{A}, minimum width=30, line width=.8, path picture={\draw[black] ([shift={(.1,.3)}]path picture bounding box.south west) to[bend left=50] ([shift={(-.1,.3)}]path picture bounding box.south east);\draw[black,-latex] ([shift={(0,.1)}]path picture bounding box.south) -- ([shift={(.3,-.1)}]path picture bounding box.north);}}}
\tikzset{roundnode/.append style={circle, draw=black, fill=gray!20, thick, minimum size=10mm}}
\tikzset{squarenode/.style={rectangle, draw=black, fill=none, thick, minimum size=10mm}}

\definecolor{Blues5seq1}{RGB}{239,243,255}
\definecolor{Blues5seq2}{RGB}{189,215,231}
\definecolor{Blues5seq3}{RGB}{107,174,214}
\definecolor{Blues5seq4}{RGB}{49,130,189}
\definecolor{Blues5seq5}{RGB}{8,81,156}

\definecolor{Greens5seq1}{RGB}{237,248,233}
\definecolor{Greens5seq2}{RGB}{186,228,179}
\definecolor{Greens5seq3}{RGB}{116,196,118}
\definecolor{Greens5seq4}{RGB}{49,163,84}
\definecolor{Greens5seq5}{RGB}{0,109,44}

\definecolor{Reds5seq1}{RGB}{254,229,217}
\definecolor{Reds5seq2}{RGB}{252,174,145}
\definecolor{Reds5seq3}{RGB}{251,106,74}
\definecolor{Reds5seq4}{RGB}{222,45,38}
\definecolor{Reds5seq5}{RGB}{165,15,21}

\usepackage{pgfplots}
\usepackage[pass]{geometry}

\newtheorem{definition}{Definition}

\usepackage{chngcntr}

\pgfplotsset{width=10cm,compat=1.9}
\usetikzlibrary{decorations.markings}
 
\begin{document}
\title{{Quantum optical communication in the presence of strong attenuation noise}}

\author{Francesco Anna Mele}
\email{francesco.mele@sns.it}
\affiliation{NEST, Scuola Normale Superiore and Istituto Nanoscienze, Consiglio Nazionale delle Ricerche, Piazza dei Cavalieri 7, IT-56126 Pisa, Italy}

\author{Ludovico Lami}
\email{ludovico.lami@gmail.com}
\affiliation{Institut f\"{u}r Theoretische Physik und IQST, Universit\"{a}t Ulm, Albert-Einstein-Allee 11, D-89069 Ulm, Germany}

\author{Vittorio Giovannetti}
\email{vittorio.giovannetti@sns.it}
\affiliation{NEST, Scuola Normale Superiore and Istituto Nanoscienze, Consiglio Nazionale delle Ricerche, Piazza dei Cavalieri 7, IT-56126 Pisa, Italy}

\begin{abstract}

Is quantum communication possible over an optical fibre with transmissivity $\lambda\leq 1/2$? The answer is well known to be negative if the environment with which the incoming signal interacts is initialised in a thermal state. However, in [PRL 125:110504, 2020] the quantum capacity was found to be always bounded away from zero for all $\lambda>0$, a phenomenon dubbed {`die-hard quantum communication' (D-HQCOM)}, provided that the initial environment state can be chosen appropriately --- depending on $\lambda$. Here we show an even stronger version of {D-HQCOM} in the context of entanglement-assisted classical communication: entanglement assistance and control of the environment enable communication with performance at least equal to that of the ideal case of absence of noise, even if $\lambda>0$ is arbitrarily small. These two phenomena of {D-HQCOM} have technological potential provided that we are able to control the environment. How can we achieve this? Our second main result answers this question. Here we provide a fully consistent protocol to activate the phenomena of {D-HQCOM} without directly accessing the environment state. This is done by sending over the channel `trigger signals', i.e.~signals which do not encode information, prior to the actual communication, with the goal of modifying the environment in an advantageous way. This is possible thanks to the memory effects which arise when the sender feeds signals separated by a sufficiently short temporal interval. Our results may offer a concrete scheme to communicate across arbitrarily long optical fibres, without using quantum repeaters. As a by-product of our analysis, we derive a simple Kraus representation of the thermal attenuator exploiting the associated Lindblad master equation.
\end{abstract}


\maketitle

\section{Introduction} 

Transmitting qubits over long distances is crucial for building a global \emph{quantum internet}, which will have an extraordinary impact on science and technology~\cite{quantum_internet_Wehner,quantum_internet_kimble}. Examples of applications of a global quantum internet are the possibility for any two parties on Earth to have unconditionally secure communication, shared entanglement, and clock synchronisation~\cite{Clock_qinternet}; the exploitation of distributed quantum computing~\cite{Distributed_QC}; the improvement of telescope observations~\cite{Telescope_qinternet}; and the possibility to access remote quantum computers in a private way~\cite{secure_access_qinternet}.   

Optical fibres have not been able to transmit qubits over long distances, even with the help of quantum repeaters~\cite{repeaters,Munro2015,quantum_repeaters_linearopt}.
It is generally maintained that this impossibility comes not only from technological limitations, but also from more profound theoretical reasons, most notably having to do with the facts that (a)~the quantum capacity of the thermal attenuator vanishes for transmissivities lower than $1/2$~\cite{Caruso_weak}; (b)~the two-way quantum capacity of the thermal attenuator vanishes for sufficiently low transmissivities~\cite{PLOB}.  

These theoretical conclusions rest on the approximation that optical fibres are memoryless~\cite{memory-review} --- i.e.\ that the noise affecting the transmission acts identically and independently on each signal. However, in~\cite{Banaszek-Memory,Ball-Memory} it has been experimentally observed that optical fibres do not always behave in a way that complies with this assumption. Hence, in principle, non-memoryless optical fibres may allow the transmission of qubits over long distances. In this paper, we prove that this is indeed the case.

In particular, we show that it is in principle possible to transmit qubits across arbitrarily long optical fibres without using repeaters. We do this by considering a more realistic model of optical fibres, where we relax the memoryless approximation and we take into account memory effects.

This paper builds on the framework of quantum communication theory. 
The goal of quantum communication theory is to transfer information reliably and efficiently from a sender Alice to a receiver Bob across a fixed noisy communication line~\cite{NC,MARK,HOLEVO-CHANNELS-2}. Since any signal that Alice feeds into the communication line can be corrupted by the noise, Alice and Bob need to apply suitable protocols in order to communicate reliably. The typical protocol is composed of three phases:
\begin{itemize}
    \item Encoding part: Alice encodes the (classical or quantum) information she wishes to send to Bob into a quantum state of a large number of systems, each of which is then sent across the communication line. 
    \item Noise part: The noise affects the signal travelling from Alice to Bob.
    \item Decoding part: {Bob applies a decoding procedure on the received (corrupted) signals in order to recover the information Alice wished to send.} 
\end{itemize}
In technological implementations, these protocols often employ continuous variable quantum systems~\cite{BUCCO,Ralph1999,Braunstein-review}, e.g.\ photonic signals propagating across optical fibres.

The transmissivity of an optical fibre is defined as the fraction of input energy that reaches the output of the fibre. Typically, the transmissivity $\lambda$ of an optical fibre decreases exponentially with its length $L$:
\begin{equation}
    \lambda=10^{-\gamma \frac{L}{\SI{10}{\km}}}\,,
    \label{lambda_decay}
\end{equation}
where typically $\gamma\simeq 0.2$. Nowadays the absolute record is $\gamma\simeq 0.14$~\cite{Tamura2018, Li2020}.

A memoryless optical fibre {with} 
transmissivity $\lambda$ is usually schematised as a thermal attenuator $\Phi_{\lambda,\tau_\nu}$.  A thermal attenuator $\Phi_{\lambda,\tau_\nu}$ is a bosonic quantum channel that acts by mixing the input state with an environment initialised in a thermal state $\tau_\nu$ through a beam splitter {(BS)} of transmissivity $\lambda\in[0,1]$, where $\nu\ge0$ denotes the mean photon number of the environment. 

The quantum capacity is a figure of merit that measures the ability of a given quantum channel in transmitting reliably qubits. It turns out that the quantum capacity of the thermal attenuator $Q(\Phi_{\lambda,\tau_\nu})$ vanishes for all $\lambda\le 1/2$ and $\nu\ge0$~\cite{Caruso_weak}. Going back to~\eqref{lambda_decay}, this is seen to mean that it is impossible to transmit qubits across memoryless optical fibres which are longer than \SI{15}{\km} or at most \SI{21.5}{\km}.

If Alice and Bob have the possibility to communicate classically, then the relevant figure of merit for transmitting qubits reliably is the so-called two-way assisted quantum capacity $Q_2(\cdot)$~\cite{PLOB}. It has been shown that $Q_2(\Phi_{\lambda,\tau_\nu})$ drops to zero for $\lambda$ sufficiently small if $\nu>0$, otherwise if $\nu=0$ it tends to zero for $\lambda$ tending to zero~\cite{PLOB}. This means that, even if there is assistance by unlimited two-way classical communication, qubits can not be transmitted by sufficiently long optical fibres which are memoryless. However, we show that this is not the case if one exploit memory effects in optical fibres.

The memoryless approximation holds when the temporal interval between uses of the channel is sufficiently long, so that the channel environment can return to its initial thermal state $\tau_\nu$ between subsequent uses~\cite{memory-review,Dynamical-Model}. In this case, the quantum states of each input signal are affected by the same quantum channel $\Phi_{\lambda,\tau_\nu}$. On the contrary, if the time intervals are sufficiently short, the environment has not enough time to reset into the initial thermal state $\tau_\nu$. Therefore, in this case, any signal interacts with an environment state $\sigma$ which depends on the signals that have been sent previously, i.e.\ there are memory effects. This leads to the definition of \emph{general attenuator} $\Phi_{\lambda,\sigma}$. The latter acts as a thermal attenuator, but the environment is initialised in an arbitrary state $\sigma$ (not necessarily thermal).  

In~\cite{die-hard} a curious phenomenon, dubbed `die-hard quantum communication' {(D-HQCOM)}, was uncovered: for all $\lambda>0$ it is possible to find a suitable environment state $\sigma=\sigma(\lambda)$ such that the quantum capacity of the general attenuator $\Phi_{\lambda,\sigma(\lambda)}$ is larger than a positive constant $c>0$, in formula $Q(\Phi_{\lambda,\sigma(\lambda)})>c$. In this paper, we 
endow the phenomenon of {D-HQCOM} with a technological relevance, and we do so by exploiting memory effects in optical fibres. To this end, we 
design a 
`noise attenuation protocol' that effectively turns an optical fibre with transmissivity $\lambda$ and arbitrary (typically thermal) environment state into a general attenuator $\Phi_{\lambda,\sigma(\lambda)}$, where --- crucially --- $\sigma(\lambda)$ activates {D-HQCOM}~\cite{die-hard}. 
This leads to our main result: within our model and under our assumptions, it is possible for \emph{optical fibres to transmit qubits at a constant rate over arbitrarily long distances}, i.e.\ for arbitrarily low transmissivities.


The basic idea of the noise attenuation protocol is that right before sending the actual information-carrying signals Alice will flash a trigger signal into the fibre. This serves to alter the effective environment state of the channel, turning it from $\tau_\nu$ into some $\sigma(\lambda)$. The net effect is that the information-carrying signal is affected not by the channel $\Phi_{\lambda, \tau_\nu}$, but rather by $\Phi_{\lambda,\sigma(\lambda)}$.

Our second main result concerns a generalisation of the phenomenon of {D-HQCOM} to the context of entanglement-assisted communication. We show that for arbitrarily low non-zero values of $\lambda$, there exists a suitable state $\sigma$ such that the entanglement-assistance allows $\Phi_{\lambda,\sigma}$ to reliably transfer: 
\begin{itemize}
    \item qubits with performance of the same order of that achievable by the identity channel (i.e.~the noiseless channel) in unassisted communication;
    \item bits with even better performance than that achievable by the identity channel in unassisted communication. In other words, \emph{entanglement assistance and the ability to control the environment state allows one to completely neutralise the effect of the noise in classical communication}. 
\end{itemize}
Once again, these results are not purely theoretical: on the contrary, we show that the noise attenuation protocol allows one to obtain the environment states $\sigma$ which activates these enhanced entanglement-assisted performances.

The present paper motivates the theoretical and experimental study of quantum communication 
{over} non-memoryless communication channels, 
{laying the stress on the fact} that memory effects can constitute a resource that improves communication performance. Recently, it has been shown that memory effects can also improve the performance of a quantum relay~\cite{pirandola_relay}.

The present paper also serves as the companion to the paper~\cite{Die-Hard-2-PRL}, by providing  proofs of the results stated in~\cite{Die-Hard-2-PRL} and additional developments.

The structure of the paper is as follows. In Sec.~\ref{sec_notation} we review definitions and preliminary results relevant for the rest of the paper. In Sec.~\ref{sec_Cea} we discuss the performance of entanglement-assisted communication when the environment state can be chosen by Alice and Bob. {In Sec.~\ref{sec_Cea} we also derive a simple Kraus representation of the thermal attenuator (see Theorem~\ref{TeoKraus}), which can allow one to obtain simple expressions for the action of the thermal attenuator on generic operators. This is done by solving the Lindblad master equation associated with the thermal attenuator (see Sec.~\ref{sec_Cea} for details).} 
In Sec.~\ref{sec_control}, first we introduce a realistic model of quantum communication that takes into account memory effects, and secondly we present the noise attenuation protocol and discuss its implications.

\section{Notation and preliminaries}\label{sec_notation}
A positive semi-definite trace class operators acting on a Hilbert space $\HH$ is called a density operator, or a quantum state, if it has unit trace. The set of density operators on $\HH$ will be denoted by $\mathfrak{S}(\HH)$. 
The trace norm of a bounded linear operator $\Theta$ is defined by $$\|\Theta\|_1\coloneqq \Tr\sqrt{\Theta^\dagger\Theta}\,.$$
Given $\rho_1,\rho_2\in\mathfrak{S}(\HH)$, the quantity $\frac{1}{2}\|\rho_1-\rho_2\|_1$ is called the trace distance between $\rho_1$ and $\rho_2$. It has an operational meaning that captures the concept of distinguishability of $\rho_1$ and $\rho_2$~\cite{HELSTROM,Holevo1976}.
Another way to measure how close two states $\rho_1$ and $\rho_2$ are is provided by the fidelity. This is defined by
\begin{equation}
    {F\left(\rho_1,\rho_2\right)\coloneqq \|\sqrt{\rho_1}\sqrt{\rho_2}\|_1\,.}
\end{equation}
For any $\rho_1$ and $\rho_2$, the following inequality holds~\cite{Fuchs1999}:
\begin{equation}\label{fidel_trace_ineq}
	1-F(\rho_1,\rho_2)\le \frac{1}{2}\|\rho_1-\rho_2\|_1 \le \sqrt{1-F^2(\rho_1,\rho_2)}\,.
\end{equation}
The quantum carrier we deal with is a single-mode of electromagnetic radiation with definite frequency and polarization, which is associated with the Hilbert space $\HH_S\coloneqq L^2(\mathbb{R})$. The quantum mechanical theory describing the quantum carrier is the same as the quantum harmonic oscillator.  A {BS} of transmissivity $\lambda\in[0,1]$ acting on two single-mode systems $S_1$ and $S_2$ is represented by the operator
\begin{equation}
	U_{\lambda}^{(S_1 S_2)}\coloneqq\exp\left[\arccos\sqrt{\lambda}\left(a_1^\dagger a_2-a_1a_2^\dagger\right)\right]\,,
\end{equation}
where $a_1$ and $a_2$ are the \emph{annihilation} operators on $S_1$ and $S_2$, respectively.
Let us consider a single-mode system $\HH_E\coloneqq L^2(\mathbb{R})$, which we call \emph{environment}, whose annihilation operator is denoted by $b$. Fixed $\lambda\in[0,1]$ and $\sigma\in\mathfrak{S}(\HH_E)$, a \emph{general attenuator} $\Phi_{\lambda,\sigma}:\mathfrak{S}(\HH_S)\mapsto\mathfrak{S}(\HH_S)$ is a quantum channel defined by 
	\begin{equation}\label{def_genatt}
	\Phi_{\lambda,\sigma}(\rho)\coloneqq\Tr_E\left[U_\lambda^{(SE)}  \rho\otimes\sigma  \left(U_\lambda^{(SE)}\right)^\dagger\right]\,.
	\end{equation}
If $\lambda=0$ the channel is completely noisy, since $\Phi_{0,\sigma}(\rho)\equiv \sigma$. If $\lambda=1$ the channel is noiseless, since $\Phi_{1,\sigma}(\rho)\equiv \Id(\rho)\equiv\rho$. The \emph{thermal state} $\tau_\nu$ with mean photon number $\nu\ge0$ is defined as
\begin{equation}\label{tau}
	\tau_{\nu}\coloneqq\frac{1}{\nu+1}\sum_{n=0}^\infty \left( \frac{\nu}{\nu+1}\right)^n \ketbra{n}\,,
\end{equation}
where $\ket{n}\coloneqq (n!)^{-1/2} (a^\dag)^n \ket{0}$ is the $n$-th \emph{Fock state} and $\ket{0}$ is the \emph{vacuum}. Thermal states are important since they maximise the entropy among all states with a fixed mean photon number, as established by Lemma~\ref{maxthermstate}~\cite{max_entropy_therm}. 
\begin{lemma}\label{maxthermstate}
	For all $\nu>0$ it holds that
	\bb
		&\max\left\{S(\rho)\text{ : }\rho\in\mathfrak{S}(\HH_S^{\otimes n}) \text{, }\Tr\left[\rho\sum_{i=1}^n a_i^\dagger a_i\right]\le \nu\right\} \\&\quad= S\left(\tau_{\nu/n}^{\otimes n}\right)=ng\left(\frac{\nu}{n}\right)\,.
	\ee
	where 
	\begin{equation}\label{bosonicent}
	g(\nu)\coloneqq (\nu+1)\log_2(\nu+1) - \nu\log_2 \nu 
	\end{equation}
	is a monotonically increasing function called the bosonic entropy.
\end{lemma}
The \emph{thermal attenuator} $\Phi_{\lambda,\tau_\nu}$, i.e.\ the general attenuator with a thermal environment state, is a well-studied example of general attenuator. The thermal attenuator with $\nu=0$, i.e.~$\Phi_{\lambda,\ketbrasub{0}}$, is called \emph{pure loss channel}. 

Every quantum channel $\Phi:\mathfrak{S}(\HH_S)\mapsto\mathfrak{S}(\HH_S)$ can be written in \emph{Stinespring representation}~\cite{Stinespring}, i.e.\ there exist an `environmental' Hilbert space $\HH_{E'}$, a state $\ket{0}_{E'}\in\HH_{E'}$, and a unitary operator $U^{(SE')}$ on $\HH_S\otimes\HH_{E'}$ such that for all $\rho\in\mathfrak{S}(\HH_S)$ it holds that
\begin{equation}
    \Phi(\rho)=\Tr_{E'}\left[U^{(SE')}\rho\otimes\ketbra{0}_{E'}\left( U^{(SE')} \right)^\dagger\right]\,.
\end{equation}
Hence, the output of a quantum channel that acts on a system $S$ can be always regarded as the state of $S$ after the interaction between $S$ and a fixed pure state of the environment $E'$. We can consider the state of $E'$ after the interaction. This defines the so-called \emph{complementary channel} of $\Phi$:
\begin{equation}
    \tilde{\Phi}(\rho)\coloneqq\Tr_{S}\left[U^{(SE')}\rho\otimes\ketbra{0}_{E'}\left( U^{(SE')} \right)^\dagger\right]\,.
\end{equation}
The Stinespring representation and the associated complementary channel are uniquely determined up to an isometry on the environmental system.

Every quantum channel $\Phi$ can also be written in \emph{physical representation}. The latter is a generalisation of the Stinespring representation where the environment is not necessarily in a pure state. In this case, the environment state after its interaction with the system leads to the definition of the so-called \emph{weak complementary channel} of $\Phi$~\cite{Caruso_weak}. For example, consider the case of the general attenuator $\Phi_{\lambda,\sigma}$.~\eqref{def_genatt} provides a physical representation of $\Phi_{\lambda,\sigma}$. Hence, the map $\tilde{\Phi}_{\lambda,\sigma}^{\text{wc}}:\mathfrak{S}(\HH_S)\mapsto\mathfrak{S}(\HH_E)$, defined as 
\begin{equation}
	\tilde{\Phi}_{\lambda,\sigma}^{\text{wc}}(\rho)\coloneqq\Tr_{S} \left[U_\lambda^{(SE)}  \rho\otimes\sigma  \left(U_\lambda^{(SE)}\right)^\dagger\right]\,,
\end{equation}
is a weak complementary channel of $\Phi_{\lambda,\sigma}$. The state $\tilde{\Phi}_{\lambda,\sigma}^{\text{wc}}(\rho)$ is the environment state after the interaction between the system initialised in $\rho$ and the environment initialised in $\sigma$ trough the {BS} unitary $U_\lambda^{(SE)}$.

The \emph{classical capacity} $C(\Phi)$ (resp.\ \emph{quantum capacity} $Q(\Phi)$) of a quantum channel $\Phi$ is the maximum number of bits (resp.\ qubits) that can be reliably transferred through $\Phi$ per use of $\Phi$~\cite{Bennett-EA, Bennett2002}. In addition, if Alice and Bob can exploit an unlimited number of pre-shared entangled states in the design of their communication protocols, then the relevant figure of merit is the \emph{entanglement-assisted} classical (resp.\ quantum) capacity $C_{\text{ea}}(\Phi)$ (resp.\ $Q_{\text{ea}}(\Phi)$)~\cite{Bennett-EA,Bennett2002}. Since, in practise, a protocol can not consume infinite energy, the mean total input photon number per use of $\Phi$  has to be upper bounded by a parameter $N$. This leads to the notion of \emph{energy-constrained} (EC) capacities $C(\Phi,N)$, $Q\left(\Phi,N\right)$, $C_{\text{ea}}\left(\Phi,N\right)$, and $Q_{\text{ea}}\left(\Phi,N\right)$~\cite{Mark-energy-constrained, Holevo2004,Holevo-energy-constrained,Holevo2013}.
\emph{Super-dense coding}~\cite{dense-coding} and \emph{quantum teleportation}~\cite{teleportation} protocols guarantee that $C_{\text{ea}}\left(\Phi,N\right)=2Q_{\text{ea}}\left(\Phi,N\right)$. 
The \emph{EC entanglement-assisted classical capacity} can be written as~\cite{entanglement-assisted,Bennett2002, Holevo2013}:
\begin{equation}\label{Cea}
C_{\text{ea}}\left(\Phi,N\right)=\max_{\rho\in\mathfrak{S}(\HH_S):\\ \Tr\left[\rho\,a^\dagger a \right]\le N}\left[S(\rho)+I_{\text{coh}}\left(\Phi,\rho\right)\right]
\end{equation}	
where $S(\rho)\coloneqq -\Tr[\rho\log_2\rho]$ is the \emph{von Neumann entropy}, 
\begin{equation}
    I_{\text{coh}}\left(\Phi,\rho\right)\coloneqq S\left(\Phi(\rho)\right)-S\left(\Phi\otimes \Id_P(\ketbra{\psi})\right)
\end{equation}
is the \emph{coherent information} with $\ket{\psi}\in\HH_S\otimes\HH_{P}$ being a \emph{purification} of $\rho$~\cite{NC,MARK,HOLEVO-CHANNELS-2}, $\HH_P$ being the purifying Hilbert space, and $\Id_P$ being the identity superoperator on $\HH_P$.
It turns out that the coherent information can be rewritten in terms of the complementary channel as
\begin{equation}
    I_{\text{coh}}\left(\Phi,\rho\right)\coloneqq S\left(\Phi(\rho)\right)-S\left(\tilde{\Phi}(\rho)\right)\,.
\end{equation}

Moreover, the \emph{EC quantum capacity} is given by~\cite{Lloyd-S-D, L-Shor-D, L-S-Devetak, HOLEVO-CHANNELS,Mark-energy-constrained}
\begin{equation}\label{lsdth2}
Q\left(\Phi,N\right)=\lim\limits_{n\rightarrow\infty}\frac{1}{n}Q_1\left(\Phi^{\otimes n},nN \right)\ge Q_1\left(\Phi,N\right)\,,
\end{equation}
where 
\begin{equation*}
Q_1\left(\Phi^{\otimes n},N\right)\coloneqq\max_{\rho\in\mathfrak{S}(\HH_S^{\otimes n}): \Tr\left[\rho \sum_{i=1}^n (a_i)^\dagger a_i \right]\le N}I_{\text{coh}}\left(\Phi^{\otimes n},\rho\right)\text{.}
\end{equation*}
The noiseless channel is the identity superoperator $\Id:\mathfrak{S}(\HH_S)\mapsto\mathfrak{S}(\HH_S)$, defined as $\Id(\rho)\coloneqq\rho$. The capacities of the noiseless channel are
\begin{equation}\label{capId}
C\left(\Id,N\right)=Q\left(\Id,N\right)=C_{\text{ea}}\left(\Id,N\right)/2=g(N)\,,
\end{equation}
as it can be shown by exploiting Lemma~\ref{maxthermstate}. Let us review some known cases where the capacities have been determined exactly. For the pure loss channel, it holds that~\cite{LossyECEAC1,LossyECEAC2}
\bb\label{pureloss}
C_{\text{ea}}\left(\Phi_{\lambda,\ketbrasub{0}},N\right)&=I\left(\Phi_{\lambda,\ketbrasub{0}},\tau_N\right) \\&=g(N)+g(\lambda N)-g((1-\lambda)N)\,,
\ee
and~\cite{holwer, Caruso2006, Wolf2007, Mark2012, Mark-energy-constrained, Noh2019}
\begin{equation*}\label{purelossq}
Q\left(\Phi_{\lambda,\ketbrasub{0}},N\right)=\begin{cases}
g(\lambda N)-g((1-\lambda)N) & \text{if $\lambda\ge 1/2$,} \\
0 & \text{otherwise.}
\end{cases}
\end{equation*}
The EC classical capacity of the thermal attenuator is~\cite{Giova_classical_cap} 
\bb
C(\Phi_{\lambda,\tau_\nu},N)=g\left(\lambda N+(1-\lambda)\nu\right)-g\left((1-\lambda)\nu\right)\,,
\ee
while its EC entanglement-assisted classical capacity
is~\cite{holwer}
\bb\label{func}
C_{\text{ea}}\left(\Phi_{\lambda,\tau_{\nu}},N\right)  &=
 g(N)+g(N')-g\left(\frac{D+N'-N-1}{2}\right) \\&-g\left(\frac{D-N'+N-1}{2}\right)\,,
\ee
where 
\bb
N'&\coloneqq \lambda N+(1-\lambda)\nu\,,\\
D&\coloneqq\sqrt{\left(N+N'+1\right)^2-4\lambda N(N+1)}\,.
\ee
A closed formula of $Q\left(\Phi_{\lambda,\tau_{\nu}},N\right)$ has not yet been discovered for the case in which $\lambda>1/2$ and $\nu>0$, although sharp bounds are known~\cite{PLOB, Rosati2018, Sharma2018, Noh2019,holwer, Noh2020,fanizza2021estimating}. If $\lambda\le 1/2$, it holds that $Q\left(\Phi_{\lambda,\tau_{\nu}},N\right)=0$ for all $\nu,N\ge0$.
Our work is motivated by the following result~\cite{die-hard}.
\begin{thm}\label{diehard_th}
For all $\lambda\in(0,1]$ there exists $\sigma(\lambda)$ such that
\begin{equation}
    Q\left(\Phi_{\lambda,\sigma(\lambda)}\right)\ge Q\left(\Phi_{\lambda,\sigma(\lambda)},1/2\right) >\eta\,,
\end{equation}
where $\eta>0$ is a universal constant.  More specifically, for $\varepsilon\ge0$ sufficiently small and for all $\lambda\in(0,1/2-\varepsilon)$ it holds that \begin{equation}
    Q\left(\Phi_{\lambda,\ketbrasub{n_\lambda} }\right)\ge Q\left(\Phi_{\lambda,\ketbrasub{n_\lambda} },1/2\right) >c(\varepsilon)\,,
\end{equation}
where $c(\varepsilon)\ge0$ is a constant with respect to $\lambda$ and $n_\lambda\in\N$ satisfies $1/\lambda-1\le n_\lambda\le 1/\lambda$. Moreover, it holds that $c(0)=0$, and $c(\bar{\varepsilon})\ge 5.133\times10^{-6}$ for an appropriate $\bar{\varepsilon}$ such that $0<\bar{\varepsilon}\ll 1/6$ (see the Supplemental Material of~\cite{die-hard}). 
\end{thm} 
Hence, suitable environmental Fock states are beneficial for the performance of quantum communication. Let us state a useful lemma which follows straightforwardly from~\cite[Theorem 1]{die-hard}.
\begin{lemma}\label{lemma_diehard}
Let $\alpha\in\mathbb{C}$. Let $\sigma\in\mathfrak{S}(\HH_E)$ be of the form $\sigma=D(\alpha)\sigma_0{D(\alpha)}^\dagger$, where:
\begin{equation}
    D(z)\coloneqq\exp\left[z\, b^\dagger-z^\ast\, b\right]
\end{equation}
is the displacement operator on $E$, and $\sigma_0$ satisfies 
\begin{equation}
    V \sigma_0 V^\dagger=\sigma_0\,,
\end{equation}
with $V \coloneqq (-1)^{b^\dagger b}$ being the parity operator. Then 
\begin{equation}
    Q\left(\Phi_{1/2,\sigma}\right)=0.
\end{equation}
Furthermore, it holds that
\begin{equation}\label{weak_formula}
\tilde{\Phi}_{\lambda,\sigma}^{\text{wc}}=\mathcal{V}\circ\mathcal{D}_{-2\sqrt{\lambda}\alpha}\circ\Phi_{1-\lambda,\sigma}\,,
\end{equation}
where $\mathcal{V}$ and $\mathcal{D}_{\alpha}$ are two quantum channels defined by: 
\begin{equation}\label{tV}
\mathcal{V}(\cdot)\coloneqq V (\cdot)  V^\dagger\,,
\end{equation}
\begin{equation}\label{tD}
\mathcal{D}_\alpha(\cdot)\coloneqq D(\alpha) (\cdot) D(\alpha)^\dagger\,.
\end{equation}
\end{lemma}

The {diamond norm} of a superoperator $\Delta:\mathfrak{S}(\HH)\mapsto\mathfrak{S}(\HH)$ is
\begin{equation}\label{dnorm}
	\left\|\Delta\right\|_{\diamond}\coloneqq\sup_{\rho\in\mathfrak{S}(\HH\otimes\HH_C)} \left\|(\Delta\otimes I_C)\rho\right\|_1\text{ ,}
\end{equation}
where the $\sup$ is taken also over the ancilla systems $\HH_C$.

It turns out that the capacities are continuous with respect to diamond norm in \emph{finite} dimension~\cite{LeungSmith}. However, in the infinite-dimensional scenario the topology induced by the diamond norm is often too strong to make the capacities continuous.
For instance, although it can be shown that $\left\| \Phi_{\lambda,\ketbrasub{0}}-\Phi_{\lambda',\ketbrasub{0}}\right\|_{\diamond}=2$ for all $\lambda\ne\lambda'$~\cite[Proposition 1]{VV-diamond}, the functions $Q\left(\Phi_{\lambda,\ketbrasub{0}},N\right)$ and $C_{\text{ea}}\left(\Phi_{\lambda,\ketbrasub{0}},N\right)$ are both continuous in $\lambda$.
To remedy these undesirable features of the diamond norm, it is customary to define the energy-constrained (EC) diamond norm. For a given single-mode system $\HH_S$ with annihilation operator $a$, some $N>0$, and a superoperator $\Delta:\mathfrak{S}(\HH_S)\mapsto\mathfrak{S}(\HH_S)$, one defines~\cite{PLOB, Shirokov2018, VV-diamond}
\begin{equation}\label{ednorm}
	\left\|\Delta\right\|_{\diamond N}\coloneqq\sup_{\rho\in\mathfrak{S}(\HH_S\otimes\HH_C)\text{ : }\Tr\left[\rho\, a^\dagger a\right]\le N} \left\|(\Delta\otimes I_C)\rho\right\|_1\text{ ,}
\end{equation}
where the $\sup$ is taken also over the ancilla systems $\HH_C$.

Analogously to the diamond norm, the supremum in~\eqref{ednorm} can be restricted to pure states of $\mathfrak{S}(\HH_S\otimes\HH_{\bar{C}})$ where the ancilla $\bar{C}$ is a copy of $S$. The paper~\cite{VV-diamond} shows that the energy-constrained diamond norm provides useful continuity bounds for the energy-constrained capacities of infinite-dimensional bosonic quantum channels.
\begin{lemma}{\cite[Theorem 9]{VV-diamond}}\label{continuityBound}
Let $\Phi_1,\Phi_2:\mathfrak{S}(\HH_S)\mapsto\mathfrak{S}(\HH_S)$ be quantum channels. Suppose there exist $\alpha$, $N_0$ $\in\mathbb{R}$ such that for all $\rho\in\mathfrak{S}(\HH_S)$ it holds that
\begin{equation}
    \Tr\left[a^\dagger a\, \Phi_i(\rho)\right]\le \alpha \Tr [a^\dagger a\,\rho] + N_0\qquad \forall\ i=1,2\,.
\end{equation}
Quantum channels satisfying such a constraint are said to be `energy limited'.

Suppose that there exists $\varepsilon\in(0,1)$ such that for all $N>0$ it holds that
\begin{equation}
    \frac{1}{2}\left\| \Phi_1-\Phi_2\right\|_{\diamond N}\le \varepsilon\,.
\end{equation}
Then for all $N>0$ it holds that
\bb \label{countboundq}
&\left|Q\left(\Phi_1,N\right)-Q\left(\Phi_2,N\right)\right| \\&\le 56\sqrt{\varepsilon}\, g\left(4\frac{\alpha N + N_0}{\sqrt{\varepsilon}}\right)+6g\left(4\sqrt{\varepsilon}\right)\,.
\ee
Similar statements can be proved for the energy-constrained entanglement-assisted capacities.

In particular, $Q(\cdot\,,N)$ and $C_{\text{ea}}\left(\cdot\,,N\right)$ are continuous with respect to the EC diamond norm over the subset of energy-limited quantum channels.
\end{lemma}

\section{Noise neutralisation thanks to environment control and entanglement assistance}\label{sec_Cea}
In this section we study the quantity $C_{\text{ea}}\left(\Phi_{\lambda,\ketbrasub{n}},N\right)$, for $n\in\N$, $N>0$, and $\lambda\in(0,1]$. Hence, the environment states under consideration are Fock states. To calculate $C_{\text{ea}}\left(\Phi_{\lambda,\ketbrasub{n}}, N\right)$, we need to maximise
\begin{equation}
    I\left(\Phi_{\lambda,\ketbrasub{n}},\rho\right)=S(\rho)+I_{\text{coh}}\left(\Phi_{\lambda,\ketbrasub{n}},\rho\right)
\end{equation} on the subset of states $\rho$ satisfying the energy constraint $\Tr\left[\rho\,a^\dagger a\right]\le N$. Since the term $S(\rho)$ appears in the definition of $I\left(\Phi_{\lambda,\ketbrasub{n}},\rho\right)$, it is worth choosing as an ansatz for $\rho$ the state that maximises the entropy among all states satisfying the energy constraint, i.e.~the thermal state $\tau_N$. Incidentally, it is known that $\tau_N$ is the maximizer of $I\left(\Phi_{\lambda,\ketbrasub{n}},\rho\right)$ if $n=0$ (see~\eqref{pureloss}), i.e.~the case of the pure loss channel.

Since $S(\tau_N)=g(N)$, we have
\bb \label{bound_Cea}
	C_{\text{ea}}\left(\Phi_{\lambda,\ketbrasub{n}},N\right) &\ge I\left(\Phi_{\lambda,\ketbrasub{n}},\tau_N\right) \\
	&=g(N)+I_{\text{coh}}\left(\Phi_{\lambda,\ketbrasub{n}},\tau_N\right)\, .
\ee
Since a weak complementary channel associated with a pure environment state is a complementary channel, it holds that
\bb
&I_{\text{coh}}\left(\Phi_{\lambda,\ketbrasub{n}},\tau_N\right) \\ &\qquad =S\left(\Phi_{\lambda,\ketbrasub{n}}(\tau_N)\right)-S\left(\tilde{\Phi}^{\text{wc}}_{\lambda,\ketbrasub{n}}(\tau_N)\right)\,.
\ee
By applying~\eqref{weak_formula} and the invariance of the 
von Neumman entropy under unitary transformation, we obtain 
\begin{equation}
    S\left(\tilde{\Phi}^{\text{wc}}_{\lambda,\ketbrasub{n}}(\tau_N)\right)= S\left(\Phi_{1-\lambda,\ketbrasub{n}}(\tau_N)\right)\,.
\end{equation}
As a consequence,
\bb\label{Icoh}
&I_{\text{coh}}\left(\Phi_{\lambda,\ketbrasub{n}},\tau_N\right) \\&=S\left(\Phi_{\lambda,\ketbrasub{n}}(\tau_N)\right) -S\left(\Phi_{1-\lambda,\ketbrasub{n}}(\tau_N)\right)\,.
\ee 
In order to compute the coherent information in~\eqref{Icoh}, let us calculate the eigenvalues of $\Phi_{\lambda,\ketbrasub{n}}(\tau_N)$.

We could try to do this by exploiting the `explicit formula of {BS}' 
of Lemma~\ref{teobeamexp} in {the} Appendix, derived by calculating the action of $U_\lambda^{(SE)}$ on the tensor product of two Fock states (or, equivalently, the formulae shown by
Sabapathy and Winter~\cite[Section III.B]{KK-VV}). However, adopting this approach we would obtain a very complicated expression of the eigenvalues of $\Phi_{\lambda,\ketbrasub{n}}(\tau_N)$ which involves a series that we are not able to sum. 
More explicitly, leveraging Lemma~\ref{teobeamexp} we would obtain that
\begin{equation} \label{explicit_formula}
U_{\lambda}^{(SE)}\ket{i}_S\ket{j}_E=\sum_{m=0}^{i+j}c_m^{(i,j)}(\lambda)\ket{i+j-m}_S\ket{m}_E \,,
\end{equation}
for all $i,j\in\N$, where
\bb\label{coefficients}
c_m^{(i,j)}(\lambda)&=   \frac{1}{\sqrt{i!j!}}\sum_{k=\max(0,m-j)}^{\min(i,m)}(-1)^k\binom{i}{k}\binom{j}{m-k}\\&\quad\times\lambda^{\frac{i+m-2k}{2}}(1-\lambda)^{\frac{j+2k-m}{2}}\sqrt{m!(i+j-m)!}\,,
\ee
for all $m\in\{0,1, \ldots,i+j\}$.
By using~\eqref{explicit_formula}, one obtains that
\bb \label{stato}
\Phi_{\lambda,\ketbrasub{n}}(\tau_N)&=\frac{1}{N+1}\sum_{k=0}^{\infty}\left(\frac{N}{N+1}\right)^k\Phi_{\lambda,\ketbrasub{n}}(\ketbra{k}) \\
&=\frac{1}{N+1}\sum_{k=0}^{\infty}\left(\frac{N}{N+1}\right)^k \\
&\quad\times\sum_{m=0}^{n+k}|c_{m}^{(k,n)}(\lambda)|^2\ketbra{n+k-m} \\
&=\sum_{l=0}^{\infty}P_l(N,n,\lambda)\ketbra{l}\,,
\ee
where 
\bb \label{explfockatt}
&P_l(N,n,\lambda) \\&\coloneqq \frac{1}{N+1}\sum_{k=\max(l-n,0)}^{\infty}{\left(\frac{N}{N+1}\right)}^k|c_{n+k-l}^{(k,n)}(\lambda)|^2\,.
\ee
\eqref{stato} implies that the eigenvalues of $\Phi_{\lambda,\tau_N}(\ketbra{n})$ are $\{P_l(N,n,\lambda)\}_{l\in\N}$.

Since the calculated expression of $P_l(N,n,\lambda)$ in~\eqref{explfockatt} is very complicated, we introduce the `master equation trick', which will allow us to obtain a much simpler expression. 

Before explaining it, let us observe that for any $\lambda$, $\sigma$, and $\rho$ it holds that
\begin{equation} \label{invert}
\Phi_{\lambda,\sigma}(\rho)=\Phi_{1-\lambda,\,\rho}(\sigma)\,,
\end{equation}
as established by Lemma~\ref{lemma_invert} in {the} Appendix.
Consequently, we obtain
\begin{equation}\label{scambio}
\Phi_{\lambda,\ketbrasub{n}}(\tau_N)=\Phi_{1-\lambda,\tau_N}(\ketbra{n})\,.
\end{equation}
Therefore, the problem has now been reformulated as that of calculating the action of the thermal attenuator $\Phi_{\lambda,\tau_N} $ on a Fock state.

Now, we apply the `master equation trick':  we retrieve the action of the thermal attenuator on a Fock state as a solution to the \emph{Gorini--Kossakowski--Sudarshan--Lindblad master equation}.
Our trick allows us to sum the complicated series which appears in the expression of the eigenvalues of $\Phi_{\lambda,\ketbrasub{n}}(\tau_N)$ in~\eqref{explfockatt}. Without this simplification, the analytical calculations in Sec.~\ref{subsec_Cea} would have been much more challenging. By defining $$\rho_N(t)\coloneqq\Phi_{\exp(-t),\tau_N}(\rho)$$ with $t\in[0,\infty)$, it holds that
\begin{equation}\label{Lindblad}
\begin{cases}
    \frac{d}{dt}\rho_N(t)=N\left[a^\dagger\rho_N(t) a-\frac{1}{2}\{\rho_N(t),a a^\dagger \}\right]\\\qquad\quad+(N+1)\left[a\rho_N(t) a^\dagger-\frac{1}{2}\{\rho_N(t),a^\dagger a\}\right]\\
    \rho_N(0)=\rho
\end{cases}
\end{equation}
where $\{A,B\}\coloneqq AB+BA$ denotes the anti-commutator between the operators $A$ and $B$. For the sake of completeness, we provide a proof of~\eqref{Lindblad} in Lemma~\ref{LemmaMaster} in {the} Appendix. The master equation in~\eqref{Lindblad} can be solved explicitly. The solution can be found in~\cite[Eq.~(3.18)]{fujii2012}. It reads:
\bb\label{master_sol}
\rho_N(t)=&\frac{e^{\frac{\mu-\nu}{2}t}}{F(t)}\sum_{k=0}^{\infty}\frac{G(t)^k}{k!}  \left(a^\dagger\right)^k e^{-a^\dagger a \ln F(t)}  \\&\times\left(\sum_{m=0}^\infty \frac{E(t)^m}{m!}a^m\rho{a^\dagger}^m\right) e^{- a^\dagger a \ln F(t)} a^k \, ,
\ee
where:
\bb\label{expressions_master}
\mu&\coloneqq N+1\,;\\
\nu&\coloneqq N\,;\\
E(t)&\coloneqq\frac{\frac{2\mu}{\mu-\nu}\sinh(\frac{\mu-\nu}{2}t)}{\cosh(\frac{\mu-\nu}{2}t)+\frac{\mu+\nu}{\mu-\nu}\sinh(\frac{\mu-\nu}{2}t)}\,;\\
G(t)&\coloneqq\frac{\frac{2\nu}{\mu-\nu}\sinh(\frac{\mu-\nu}{2}t)}{\cosh(\frac{\mu-\nu}{2}t)+\frac{\mu+\nu}{\mu-\nu}\sinh(\frac{\mu-\nu}{2}t)}\,;\\
F(t)&\coloneqq\cosh\left(\frac{\mu-\nu}{2}t\right)+\frac{\mu+\nu}{\mu-\nu}\sinh\left(\frac{\mu-\nu}{2}t\right)\,.
\ee
Exploiting the following formulae
\bb
a^m\ket{n}&=
\begin{cases}
\sqrt{\frac{n!}{(n-m)!}}\ket{n-m}, & \text{if $n\ge m$,} \\
0, & \text{otherwise}
\end{cases}\\
 \left(a^\dagger\right)^k\ket{n-m}&=
\sqrt{\frac{(n-m+k)!}{(n-m)!}}\ket{n-m+k}
\ee
and using the initial condition $\rho=\ketbra{n}$, one obtains that:
\bb
&\Phi_{\exp(-t),\tau_N}(\ketbra{n})=\rho_N(t)\\&=\frac{e^{t/2}}{F(t)}\sum_{k=0}^{\infty}\sum_{m=0}^n\frac{n!G(t)^k}{(n-m)!k!}\\&\quad\times\frac{E(t)^m}{m!F(t)^{2(n-m)}} (a^\dagger)^k   \ketbra{n-m} a^k  \\&=\frac{e^{t/2}}{F(t)}\sum_{k=0}^{\infty}\sum_{m=0}^nG(t)^k\frac{E(t)^m}{F(t)^{2(n-m)}}\\&\quad\times\binom{n}{m}\binom{n-m+k}{k} \ketbra{n-m+k}\,.  
\ee
In addition, by setting
\begin{equation}
f(t)\coloneqq\frac{2N+1+\coth(t/2)}{2}\,,
\end{equation}
the expression in~\eqref{expressions_master} can be rewritten as
\bb
E(t)&=(N+1)/f(t)\,,\\
G(t)&=N/f(t)\,,\\
F(t)&=\cosh(t/2)+(2N+1)\sinh(t/2)\,.
\ee
Hence, in terms of $\lambda=\exp(-t)$, it holds that
\bb\label{fF_ambda}
f(-\ln\lambda)&=\frac{N+1-N\lambda}{1-\lambda}\,,\\
F(-\ln\lambda)&=\frac{N+1-N\lambda}{\sqrt{\lambda}}\,.
\ee
Finally, we obtain
\bb
&\Phi_{\lambda,\tau_N}(\ketbra{n}) \\&=\sum_{k=0}^{\infty}\sum_{m=0}^n\frac{(1-\lambda)^{m+k}\lambda^{n-m}N^k(N+1)^m\binom{n}{m}\binom{n-m+k}{k} }{\left(N+1-N\lambda\right)^{k-m+2n+1}} \\&\quad\times\ketbra{n-m+k}= \\&=\sum_{l=0}^\infty\frac{1}{\left(N+1-N\lambda\right)^{l+n+1}}\sum_{m=\max\{0,n-l\}}^n(1-\lambda)^{2m+l-n} \\&\quad\times\lambda^{n-m}N^{l+m-n}(N+1)^m\binom{n}{m}\binom{l}{n-m}  \ketbra{l}\,.
\ee
In conclusion, by using 
\begin{equation}\label{defP}
    \Phi_{\lambda,\ketbrasub{n}}(\tau_N)=\sum_{l=0}^{\infty}P_l(N,n,\lambda)\ketbra{l}\,,
\end{equation}
~\eqref{scambio} implies
\bb\label{simpler}
    &P_l(N,n,\lambda)=\bra{l}\Phi_{\lambda,\ketbrasub{n}}(\tau_N)\ket{l}=\bra{l}\Phi_{1-\lambda,\tau_N}(\ketbra{n})\ket{l} \\&
    =\frac{1}{\left(1+N\lambda\right)^{l+n+1}}\sum_{m=\max\{0,n-l\}}^n\lambda^{2m+l-n}(1-\lambda)^{n-m} \\&\quad\times N^{l+m-n}(N+1)^m\binom{n}{m}\binom{l}{n-m}\,.
\ee
Therefore, the `master equation trick' provides a significantly simpler expression for the eigenvalues $\{P_l(N,n,\lambda)\}_{l\in\N}$ of $\Phi_{\lambda,\ketbrasub{n}}(\tau_N)$ than that in~\eqref{explfockatt}. The expression in~\eqref{simpler} was obtained from a standard calculation exploiting the `explicit formula of {BS}' of 
Lemma~\ref{teobeamexp}. Comparing the expressions in~\eqref{explfockatt} and~\eqref{simpler}, we find that 
\bb
&\sum_{k=\max(l-n,0)}^{\infty}{\left(\frac{N}{N+1}\right)}^k|c_{n+k-l}^{(k,n)}(\lambda)|^2 \\&=\frac{N+1}{\left(1+N\lambda\right)^{l+n+1}}\sum_{m=\max\{0,n-l\}}^n\lambda^{2m+l-n}(1-\lambda)^{n-m} \\&\quad\times N^{l+m-n}(N+1)^m\binom{n}{m}\binom{l}{n-m}\,.
\ee
Hence, the `master equation trick' allows us to sum this complicated series that involves the coefficients $c^{(i,j)}_m$ expressed in~\eqref{coefficients}.

Furthermore, the simple expression of $P_l(N,n,\lambda)$ in~\eqref{simpler} can be useful in lower bounding the capacities of general attenuators whose environment state is diagonal in Fock basis. Indeed, the capacities of a general attenuator $\Phi_{\lambda,\sigma}$ can be lower bounded in terms of $I_{\text{coh}}\left(\Phi_{\lambda,\sigma},\tau_N\right)$, and the latter in terms of the eigenvalues of $\Phi_{\lambda,\sigma}(\tau_N)$, by using the subadditivity inequality of the von Neumann entropy to deal with the term $S\left(\tilde{\Phi}_{\lambda,\sigma}(\tau_N)\right)$. Therefore, if $\sigma$ is diagonal in Fock basis, by noting that in this case~\eqref{defP} allows one to easily calculate the eigenvalues of $\Phi_{\lambda,\sigma}(\tau_N)$, one can obtain a lower bound on the capacities of $\Phi_{\lambda,\sigma}$. This reasoning is developed in Lemma~\ref{lowcap} in the Appendix.

{Moreover, the `master equation trick' provides a Kraus representation for the thermal attenuator, we it report in Theorem~\ref{TeoKraus}. Such Kraus representation is much simpler than that reported in Lemma~\ref{lemmakraus} in the Appendix, obtained from a standard calculation exploiting the `explicit formula of {BS}'. The Kraus representation reported in Theorem~\ref{TeoKraus} can be useful in calculating the action of the thermal attenuator on a generic operator. For example, as we show in Remark~\ref{lemmakraus3} in the Appendix, it significantly simplifies the action of the thermal attenuator on an operator of the form $\ketbraa{n}{i}$, where $\ket{n}$ and $\ket{i}$ are Fock states.  
\begin{thm}\label{TeoKraus}
For all $\lambda\in[0,1],\nu\ge 0$, the thermal attenuator $\Phi_{\lambda,\tau_\nu}$ admits the following Kraus representation:
\begin{equation}\label{krausformatt}
    \Phi_{\lambda,\tau_\nu}(\rho)=\sum_{k,m=0}^\infty M_{k,m}\rho M_{k,m}^\dagger\,,
\end{equation}
where 
\bb
    M_{k,m}\coloneqq&\sqrt{\frac{\nu^k(\nu+1)^m(1-\lambda)^{m+k}}{k!m![(1-\lambda)\nu+1]^{m+k+1}}} \\&\quad\times(a^\dagger)^k\left(\frac{\sqrt{\lambda}}{(1-\lambda)\nu+1}\right)^{a^\dagger a}a^m\,.
\ee
In particular, by letting $\ket{n}$ and $\ket{i}$ two Fock states, it holds that
\begin{equation}\label{action_ni}
    \Phi_{\lambda,\tau_\nu}(\ketbraa{n}{i})=\sum_{l=\max(i-n,0)}^\infty f_{n,i,l}(\lambda,\nu)\ketbraa{l+n-i}{l}\,,
\end{equation}
where
\bb
    f_{n,i,l}(\lambda,\nu)&\coloneqq \sum_{m=\max(i-l,0)}^{\min(n,i)}\frac{\sqrt{n!i!l!(l+n-i)!}}{(n-m)!(i-m)!m!(l+m-i)!}\\&\quad\times\frac{\nu^{l+m-i}(\nu+1)^m(1-\lambda)^{2m+l-i}\lambda^{\frac{n+i-2m}{2}}}{\left((1-\lambda)\nu+1\right)^{l+n+1}}\,.
\ee
\end{thm}
The proofs of~\eqref{krausformatt} and~\eqref{action_ni} are provided in Lemma~\ref{lemmakraus2} and Remark~\ref{lemmakraus3} in the Appendix, respectively.}

{Now, let us conclude the calculation of the coherent information in~\eqref{Icoh}.}
The von Neumann entropy of $\Phi_{\lambda,\ketbrasub{n}}(\tau_N)$ is the Shannon entropy of the probability distribution $\{P_l(N,n,\lambda)\}_{l\in\N}$ i.e. 
\bb
    S\left(\Phi_{\lambda,\ketbrasub{n}}\left(\tau_N\right)\right) &= H\left(P(N,n,\lambda)\right)\\
    &=-\sum_{l=0}^{\infty}P_l(N,n,\lambda)\log_2 P_l(N,n,\lambda)\,.
\ee
Finally, from~\eqref{Icoh}, we have that
\bb \label{cohprob}
	&I_{\text{coh}}\left(\Phi_{\lambda,\ketbrasub{n}},\tau_N\right) \\ &\qquad =H\left(P(N,n,\lambda)\right)-H\left(P(N,n,1-\lambda)\right)\, .
\ee
By using~\eqref{cohprob}, in Fig.~\ref{Icoh figure} we plot $I_{\text{coh}}\left(\Phi_{\lambda,\ketbrasub{n}},\tau_N\right)$ as a function of $\lambda$ for $N=0.5$ and increasing $n$. 
\begin{figure}[t]
\centering
\includegraphics[width=1.0\linewidth]{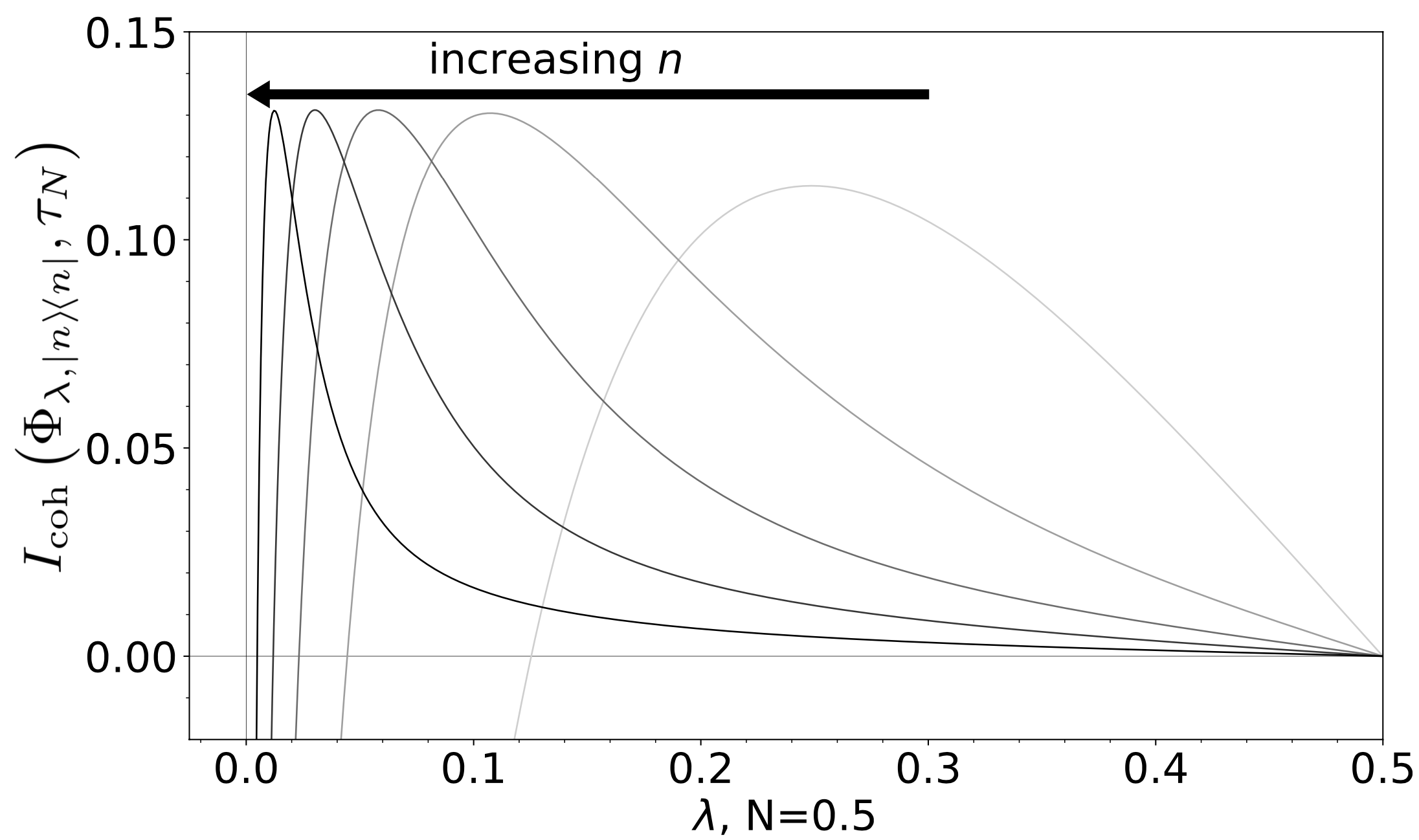}
\caption{The functions $I_{\text{coh}}\left(\Phi_{\lambda,\ketbrasub{n}},\tau_N\right)$ plotted with respect to the variable $\lambda$ for $N=0.5$ and for several values of $n$ from $3$ to $100$. We have computed $I_{\text{coh}}\left(\Phi_{\lambda,\ketbrasub{n}},\tau_N\right)$ by using~\eqref{cohprob}.}
\label{Icoh figure}
\end{figure}

For the other values of $N$, the trend of $I_{\text{coh}}\left(\Phi_{\lambda,\ketbrasub{n}},\tau_N\right)$ is similar to that in Fig.~\ref{Icoh figure}. This suggests that, for fixed $N>0$ and $n\in\N$, there exists a transmissivity value $\bar{\lambda}_N(n)\in (0,1/2]$ for which the function $\lambda\mapsto I_{\text{coh}}\left(\Phi_{\lambda,\ketbrasub{n}},\tau_N\right)$ is strictly positive for $\lambda\in(\bar{\lambda}_N(n),1/2)$ and negative for $\lambda\in\left(0,\bar{\lambda}_N(n)\right)$. In addition, $\bar{\lambda}_N(n)$ seems to decrease in $n$ and to satisfy $\lim\limits_{n\rightarrow\infty}\bar{\lambda}_N(n)=0$. This leads us to Conjecture~\ref{congI}.

Incidentally, a similar behaviour is displayed by $I_{\text{coh}} \left(\Phi_{\lambda,\ketbrasub{n}}, \frac{\ketbrasub{0}+\ketbrasub{1}}{2}\right)$, as discussed in~\cite{die-hard}. There, it was analytically shown that for all $\lambda>0$ such a coherent information is bounded away from zero, if $n$ is chosen appropriately.

\begin{cj} \label{congI}
For all $\lambda\in(0,1/2)$, $N>0$, and for $n\in\N$ sufficiently large it holds that
\begin{equation}
    I_{\text{coh}}\left(\Phi_{\lambda,\ketbrasub{n}},\tau_N\right)>0\,.
\end{equation}
\end{cj}

Conjecture~\ref{congI} has several implications since the quantity $I_{\text{coh}}\left(\Phi_{\lambda,\ketbrasub{n}},\tau_N\right)$ provides a lower bound on several capacities of $\Phi_{\lambda,\ketbrasub{n}}$. Indeed,~\eqref{bound_Cea}
and~\eqref{capId} imply that \begin{equation}\label{lboundCea}
    C_{\text{ea}}\left(\Phi_{\lambda,\ketbrasub{n}},N\right)\ge C\left(\Id,N\right)+I_{\text{coh}}\left(\Phi_{\lambda,\ketbrasub{n}},\tau_N\right)\,.
\end{equation}
In addition, the relation $C_{\text{ea}}=2Q_{\text{ea}}$ and~\eqref{capId} guarantee that 
\begin{equation}\label{lboundQea}
    Q_{\text{ea}}\left(\Phi_{\lambda,\ketbrasub{n}},N\right)\ge \frac{Q\left(\Id,N\right)+I_{\text{coh}}\left(\Phi_{\lambda,\ketbrasub{n}},\tau_N\right)}{2}\,,
\end{equation}
Finally,~\eqref{lsdth2} yields
\begin{equation}\label{lboundQ}
    Q\left(\Phi_{\lambda,\ketbrasub{n}},N\right)\ge I_{\text{coh}}\left(\Phi_{\lambda,\ketbrasub{n}},\tau_N\right)\,.
\end{equation}
The validity of Conjecture~\ref{congI}, together with the above lower bounds on the capacities, would prove Theorem~\ref{congCap}. 

\begin{thm} \label{congCap}
Suppose that Conjecture~\ref{congI} is valid. Then, for all $\lambda\in(0,1/2)$, $N>0$, and for $n\in\N$ sufficiently large it holds that 
\begin{align}
C_{\text{ea}}\left(\Phi_{\lambda,\ketbrasub{n}},N\right) &> C\left(\Id,N\right)\,, \label{SpecialCea} \\
Q_{\text{ea}}\left(\Phi_{\lambda,\ketbrasub{n}},N\right) &> Q\left(\Id,N\right)/2\,, \label{SpecialQea} \\
Q\left(\Phi_{\lambda,\ketbrasub{n}},N\right) &> 0\,. \label{SpecialQ}
\end{align}
\end{thm}

Let us analyse each of these lower bounds on the capacities.

First, since $C\left(\Id,N\right)$ is the best-achievable rate of bits reliably transmitted from Alice to Bob in the ideal case of absence of noise (i.e.~$\lambda=1$),~\eqref{SpecialCea} would imply the following statement. \emph{The pre-shared entanglement resource and the control of the environment state allow one to communicate with better performance than the ideal case of absence of noise, even if the transmissivity $\lambda>0$ is arbitrarily small}. This statement can be expressed by means of the inequality
\begin{equation}\label{sup_sigma}
    \sup_\sigma C_{\text{ea}}\left(\Phi_{\lambda,\sigma},N\right)> C\left(\Id,N\right)\,,
\end{equation}
for all $N>0$ and all $\lambda\in(0,1/2)$. Or, more specifically, as
\begin{equation}
    \sup_{n\in\N} C_{\text{ea}}\left(\Phi_{\lambda,\ketbrasub{n}},N\right)> C\left(\Id,N\right)\,.
\end{equation}
\eqref{sup_sigma} is well known to be true if $\lambda>1/2$ (it suffices to take $\sigma=\ketbra{0}$). Here we find that~\eqref{sup_sigma} can hold also for $0<\lambda<1/2$. In addition, it is well known that the resource of pre-shared entanglement between Alice and Bob helps communication. Here, we show that this resource, together with the ability to choose the environment state, is able to completely neutralise the action of the noise on an arbitrarily long optical fibre used to transmit bits.   

To put~\eqref{SpecialCea} into perspective, notice that if the environment is in a thermal state (as in the usual scheme of an optical fibre), for $\lambda<1/2$ it is not possible to transmit bits with better performance than the unassisted noiseless scenario, even when pre-shared entanglement is available. Indeed, for all $\lambda<1/2$ and all $N,\nu\ge0$ it holds that 
\begin{equation}
    C\left(\Phi_{\lambda,\tau_\nu},N\right)\le C_{\text{ea}}\left(\Phi_{\lambda,\tau_\nu},N\right)< C\left(\Id,N\right)\,.
\end{equation}
In addition, since $C_{\text{ea}}\left(\Phi_{\lambda,\tau_\nu},N\right)$ approaches $0$ as $\lambda$ approaches $0$, the environment control would infinitely improve the performance of entanglement-assisted communication for $\lambda\rightarrow 0^+$.

Second,~\eqref{SpecialQea} would imply {that the
pre-shared entanglement resource and the appropriate control of the environment state always allow one to transmit qubits with a rate equal to at least half of the energy-constrained quantum capacity of the noiseless channel.} 

Third,~\eqref{SpecialQ} would constitute another proof of the phenomenon of {D-HQCOM}~\cite{die-hard}. That is, for any arbitrarily low non-zero value of the transmissivity, there exists an environment state for which the quantum capacity of the corresponding general attenuator is strictly positive.~\eqref{SpecialQ} is proved in~\cite{die-hard} only for $N\ge1/2$ and for $n\in\N$ such that $1/\lambda-1\le n\le 1/\lambda$. Here, we are able to explore {D-HQCOM} also for $N\in(0,1/2)$.

Conjecture~\ref{congI}, and hence also Theorem~\ref{congCap}, rests on the hypothesis that $n$ has to be sufficiently large. A natural question that arises in this context is: for a fixed $N>0$ and $\lambda\in(0,1/2)$, what is the minimum $n\in\N$ such that $I_{\text{coh}}\left(\Phi_{\lambda,\ketbrasub{n}},\tau_N\right)>0$\,?
Let $\bar{n}_N(\lambda)$ denote such a minimum, i.e. 
\bb\label{nmineq}
&\bar{n}_N(\lambda)\coloneqq\min\{n\in\N:I_{\text{coh}}\left(\Phi_{\lambda,\ketbrasub{n}},\tau_N\right)> 0\} \,.
\ee
From numerical investigations, the above optimisation problem seems to admit a solution for all $N>0$ and $\lambda\in(0,1/2)$. More precisely, the plot of $\bar{n}_N(\lambda)$ in Fig.~\ref{nmin} suggests that
\begin{equation}\label{defKN}
	\bar{n}_N(\lambda)\sim \frac{K(N)}{\lambda}\quad\text{for }\lambda\rightarrow0^+\,,
\end{equation}
where $K(N)$ is a monotonically increasing positive function. In Fig.~\ref{K(N)} we plot an estimate of $K(N)$.
\begin{figure}[t]
	\centering
	\includegraphics[width=1.0\linewidth]{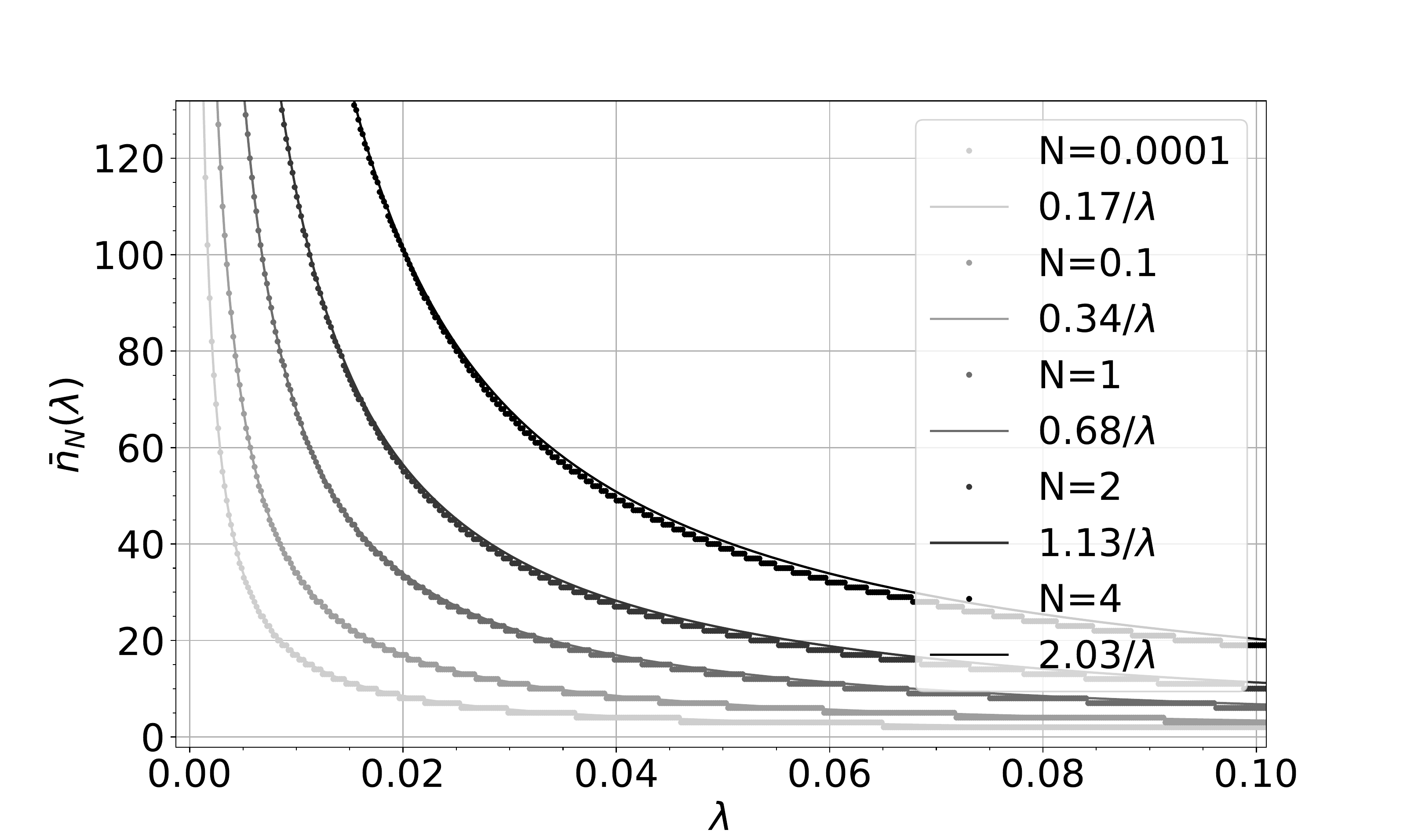}
	\caption{The function $\bar{n}_N(\lambda)$ plotted with respect to the transmissivity $\lambda$ with a step of $\Delta\lambda=0.0002$ for several values of the energy-constraint $N$. The continuous lines are functions of the form $K(N)/\lambda$.}
	\label{nmin}
\end{figure}
\begin{figure}[t]
	\centering
	\includegraphics[width=1.0\linewidth]{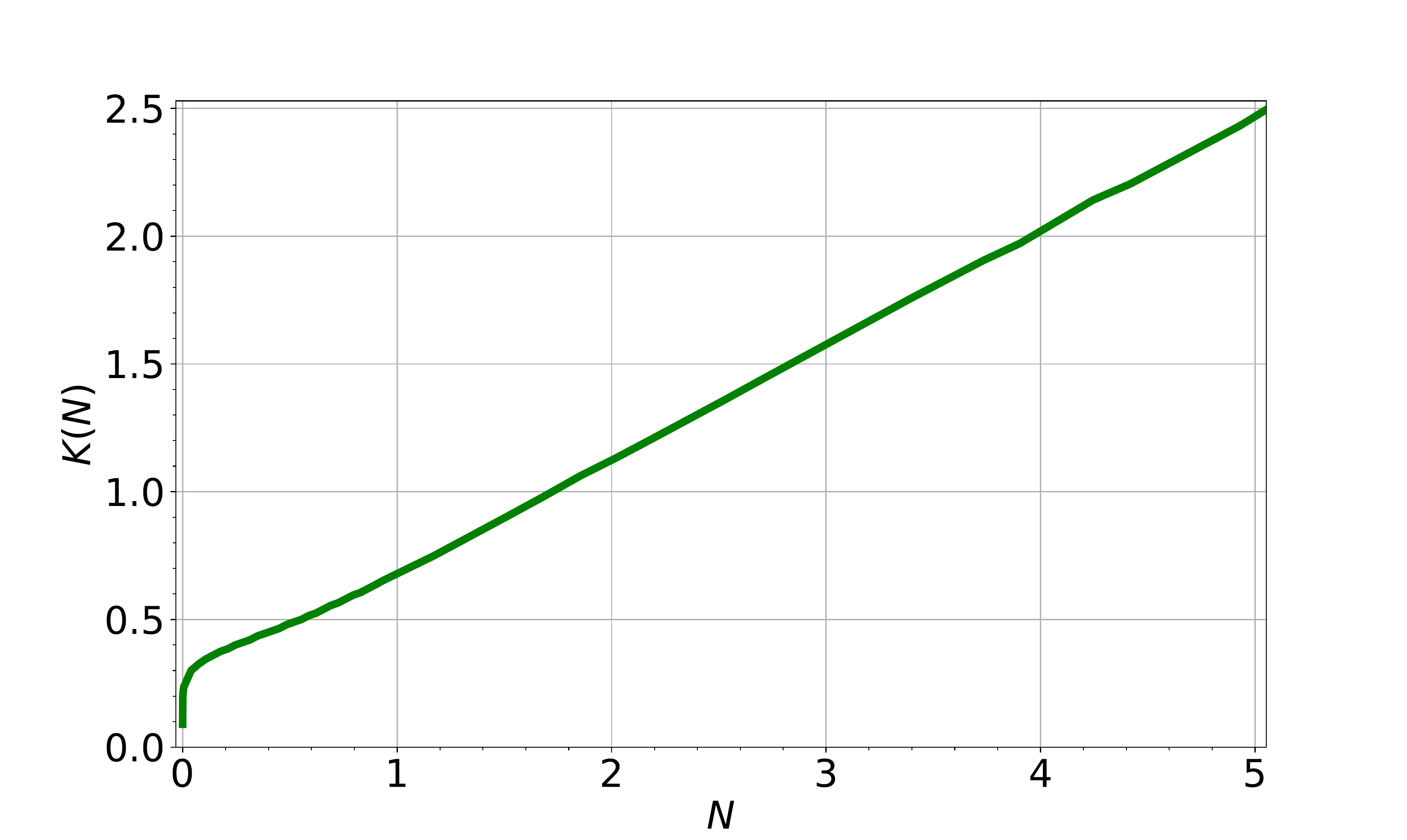}
	\caption{An estimate of the function $K(N)$ plotted with respect to the energy-constraint $N$.}
	\label{K(N)}
\end{figure}
From the above numerical observations it seems that the following more specific conjecture holds.
\begin{cj} \label{cong2}
	Let $N>0$ and $\lambda\in\left(0,1/2\right)$. It holds that $I_{\text{coh}}\left(\Phi_{\lambda,\ketbrasub{n}},\tau_N\right)> 0$ if and only if $n\ge\bar{n}_N(\lambda)$,
	where $\bar{n}_N(\lambda)\sim \frac{K(N)}{\lambda}$ for $\lambda>0$ sufficiently small with $K(N)$ being a proper increasing positive function. As a consequence, for all $n\ge \bar{n}_N(\lambda)$ the inequalities in~\eqref{SpecialCea},~\eqref{SpecialQea}, and~\eqref{SpecialQ} hold.
\end{cj} 
\subsection{Limit of vanishing transmissivity}\label{subsec_Cea} 
Above we have seen that the positivity of the coherent information {$I_{\text{coh}}\left(\Phi_{\lambda,\ketbrasub{n}},\tau_N\right)$ in~\eqref{Icoh}}
has strong implications in terms of performances of the communication across an optical fibre with transmissivity $\lambda$. The case in which the optical fibre is very long --- i.e.~the transmissivity $\lambda>0$ is very small --- is the most interesting one. In Theorem~\ref{congl0} we will take into account this case providing a lower bound on the coherent information in the limit in which $\lambda$ goes to $0^+$ as $c/n$ where $n\rightarrow\infty$ and $c>0$ is a constant. Interestingly, such a lower bound turns out to be positive, as we will see later.

Before stating Theorem~\ref{congl0}, let us explain why we take the transmissivity to be 
\begin{equation}
    \lambda_n=c/n\,.
\end{equation} First,~\eqref{defKN} suggests that we need to take $\lambda_n n\gtrsim K(N)$ in order to obtain $I_{\text{coh}}\left(\Phi_{\lambda_n,\ketbra{n}},\tau_N\right)>0$. Therefore, we are led to consider as ansatz a sequence {$\{\lambda_n\}_{n\in\N}$} 
such that $\lambda_n n$ tends either to infinity or to a constant $\gtrsim K(N)$. The choice of $\lambda_n n$ proportional to a constant makes the calculation simpler since, in this case, we will see that the sequence $\left\{\Phi_{1-\lambda_n,\ketbra{n}}(\tau_N)\right\}_{n\in\N}$ converges as $n\rightarrow\infty$. In addition, notice that if $\lambda_n n$ is proportional to a constant, then the mean photon number of $\Phi_{1-\lambda_n,\ketbra{n}}(\tau_N)$ converges. While, if $\lambda_n n\rightarrow\infty$, then this mean photon number diverges.

Observe that the mean photon number of $\Phi_{\lambda,\sigma}(\rho)$ can be calculated as
\bb\label{energia_general}
	&\langle a^\dagger a\rangle_{\Phi_{\lambda,\sigma}(\rho)}\\&=\lambda\langle a^\dagger a\rangle_\rho +(1-\lambda)\langle b^\dagger b\rangle_\sigma+2\sqrt{\lambda(1-\lambda)}\Re\left(\langle a\rangle_\rho \langle b^\dagger\rangle_\sigma\right)\,,
\ee
as proved in Lemma~\ref{LemmaEner} in the Appendix.
\begin{thm}\label{congl0}
	For all $N>0$ and $c>0$ it holds that
	\begin{equation}\label{entdifflow}
		\liminf\limits_{n\rightarrow\infty}I_{\text{coh}}\left(\Phi_{\frac{c}{n},\ketbrasub{n}},\tau_N\right)\ge H\left(q(N,c)\right)-H\left(p(N,c)\right)\,,
	\end{equation}
	where $\{q_{k}(N,c)\}_{k\in\mathbb{Z}}$ and $\{p_{k}(N,c)\}_{k\in\mathbb{Z}}$ are two probability distributions given by
	\begin{equation}\label{prinf}
	q_k(N,c)=e^{-c(2N+1)}\left(\frac{N}{N+1}\right)^{k/2}I_{|k|}\left(2c\sqrt{N(N+1)}\right)
	\end{equation}
	\bb\label{ninfprob}
	&p_k(N,c)=\\&
	\begin{cases}
	\frac{N^k}{(N+1)^{k+1}}e^{-c/(N+1)}L_k\left(-\frac{c}{N(N+1)}\right)\,, & \text{if $k\ge 0$,} \\
	0\,, & \text{otherwise.}
	\end{cases}
	\ee
	where $I_k(\cdot)$ is the $k$-th Bessel function of the first kind i.e.
	\begin{equation}\label{besself}
	I_k(z)\coloneqq\left(\frac{z}{2}\right)^k 	\sum_{m=0}^{\infty}\frac{\left(\frac{z}{2}\right)^{2m}}{m!(k+m)!}\,,
	\end{equation}
	and $L_k(\cdot)$ is the $k$-th Laguerre polynomial i.e.
	\begin{equation}
	L_k(x)\coloneqq\sum_{m=0}^k\frac{(-1)^m}{m!}\binom{k}{m}x^m\,.
	\end{equation}
\end{thm}
Before proving Theorem~\ref{congl0}, let us explain its implications. As discussed above, our aim is to show that the coherent information is positive in the limit of vanishing transmissivity. As guaranteed by Theorem~\ref{congl0}, if
\bb
H\left(q\left(N,c\right)\right)-H\left(p\left(N,c\right)\right)>0\,,
\ee
then for $n$ sufficiently large it holds that
\bb
I_{\text{coh}}\left(\Phi_{\frac{c}{n},\ketbrasub{n}},\tau_N\right)>0\,.
\ee
We do not expect that $H\left(q\left(N,c\right)\right)-H\left(p\left(N,c\right)\right)>0$ for all $c,N>0$. Indeed,~\eqref{defKN} suggests that, for fixed $N>0$, we need to take $c\gtrsim K(N)$ in order to obtain $I_{\text{coh}}\left(\Phi_{\frac{c}{n},\ketbrasub{n}},\tau_N\right)>0$ for $n$ sufficiently large ($K(N)$ is estimated in Fig.~\ref{K(N)}). 
This is the reason why we choose $c$ to be a function of $N$. In particular, we choose $c(N)$ so that the inequality $c(N)\gtrsim K(N)$ is satisfied.

In Fig.~\ref{entropy_diff} we plot the function $H\left(q\left(N,c(N)\right)\right)-H\left(p\left(N,c(N)\right)\right)$ where $c(N)$ is chosen to be $$c(N)\coloneqq N+\alpha$$ and $\alpha$ is a constant. 
\begin{figure}[t]
\centering
\includegraphics[width=1.0\linewidth]{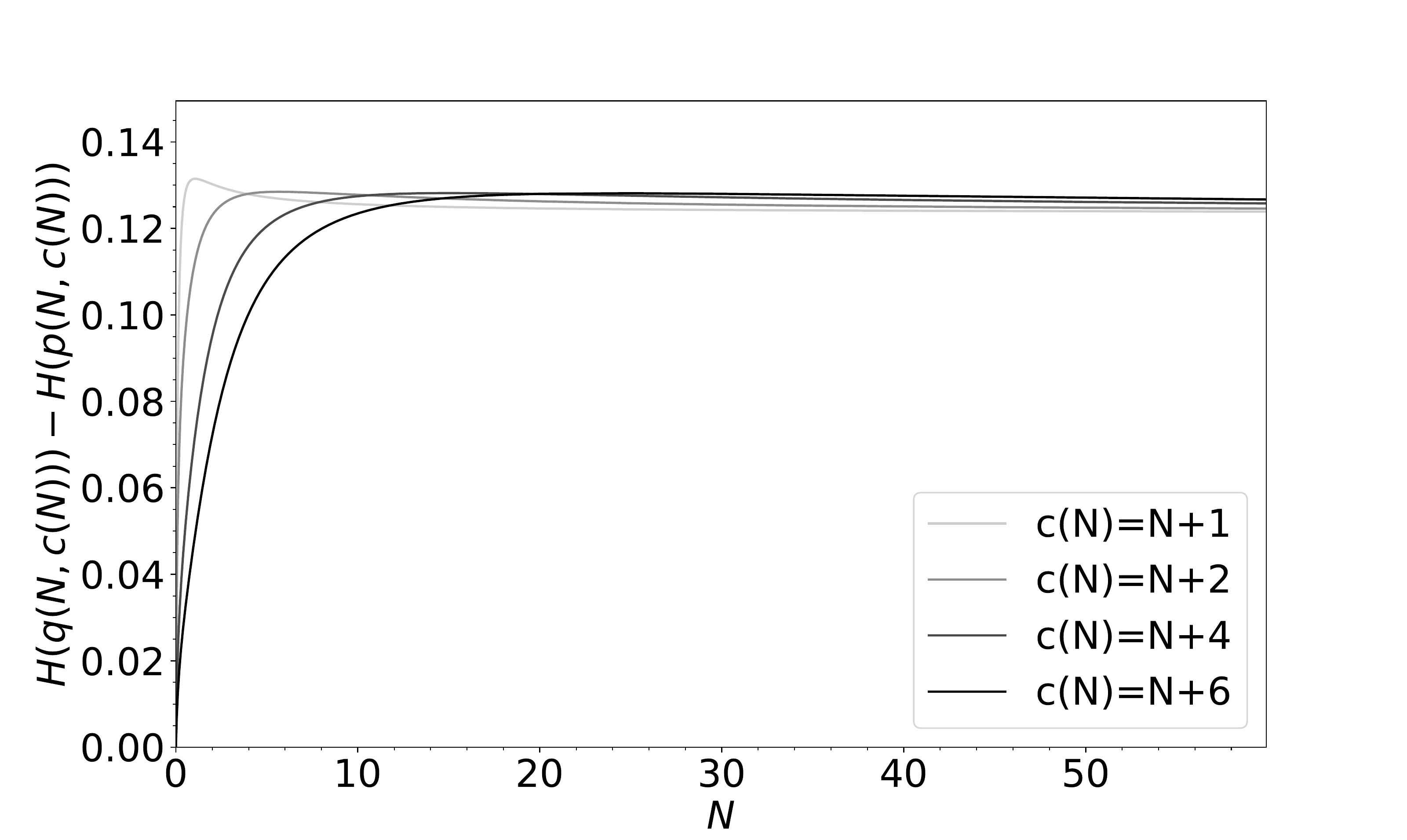}
\caption{The functions $H\left(q\left(N,c(N)\right)\right)-H\left(p\left(N,c(N)\right)\right)$ plotted with respect to the variable $N$ for several choices of the function $c(N)$.}
\label{entropy_diff}
\end{figure}
From Fig.~\ref{entropy_diff} we see that
\begin{equation}\label{ineqcong}
    H\left(q\left(N,c(N)\right)\right)-H\left(p\left(N,c(N)\right)\right)>0
\end{equation}
for proper choices of the function $c(N)$. We do not yet possess an analytical proof of the fact that~\eqref{ineqcong} holds for all $N>0$, but the numerical evidence we have gathered, depicted in Fig.~\ref{entropy_diff}, is very clear. Indeed, the function $H\left(q\left(N,c(N)\right)\right)-H\left(p\left(N,c(N)\right)\right)$ appears to have a positive horizontal asymptote. As a consequence, Theorem~\ref{congl0}, together with the lower bounds in~\eqref{lboundCea},~\eqref{lboundQea}, and~\eqref{lboundQ},  establishes that:
\bb\label{SpecialCea0}
\liminf_{n\rightarrow\infty} C_{\text{ea}}\left(\Phi_{\frac{c(N)}{n},\ketbrasub{n}},N\right)&>C\left(\Id,N\right)\,,  \\  
\liminf_{n\rightarrow\infty} Q_{\text{ea}}\left(\Phi_{\frac{c(N)}{n},\ketbrasub{n}},N\right)&>Q\left(\Id,N\right)/2\,, \\  
\liminf_{n\rightarrow\infty} Q\left(\Phi_{\frac{c(N)}{n},\ketbrasub{n}},N\right)&>0\,.  
\ee
These inequalities show that all the discussed implications of Theorem~\ref{congCap} are valid in the physically interesting case in which $\lambda>0$ is very small. Namely, we have proved that for arbitrarily small values of $\lambda>0$ there exist environment states that allow one to reliably transmit:
\begin{itemize}
    \item Bits with better performances than in the (unassisted) ideal case of absence of noise, provided that pre-shared entanglement assists communication. In other words, the ability of choosing the environment state and the entanglement resource suffice to neutralise the effect of the noise affecting an optical fibre used to transmit bits, no matter how small the signal-to-noise ratio is.
    \item Qubits with a rate larger than half of the best-achievable one in the (unassisted) ideal case of absence of noise, provided that pre-shared entanglement assists communication.
    \item Qubits at a strictly positive rate that is moreover independent of $\lambda$. This result was already known from~\cite{die-hard}.
\end{itemize}

We are now ready to prove Theorem~\ref{congl0}.

\begin{proof}[Proof of Theorem~\ref{congl0}]
From~\eqref{Icoh}, we have  
\bb\label{conscons}
&I_{\text{coh}}\left(\Phi_{\frac{c}{n},\ketbrasub{n}},\tau_N\right)\\&=S\left(\Phi_{\frac{c}{n},\ketbrasub{n}}(\tau_N)\right) -S\left(\Phi_{1-\frac{c}{n},\ketbrasub{n}}(\tau_N)\right)\,,
\ee
where, by using~\eqref{defP}, we can write:
\bb\label{phia}
\Phi_{\frac{c}{n},\ketbrasub{n}}(\tau_N)&=\sum_{l=0}^{\infty}P_l\left(N,n,\frac{c}{n}\right)\ketbra{l}\,,\\
\Phi_{1-\frac{c}{n},\ketbrasub{n}}(\tau_N)&=\sum_{l=0}^{\infty}P_l\left(N,n,1-\frac{c}{n}\right)\ketbra{l}\,.\
\ee
Before explaining our proof, let us forge our intuition by looking at Fig.~\ref{ProbIntuition}.
\begin{figure}[t]
	\centering
	\includegraphics[width=1\linewidth]{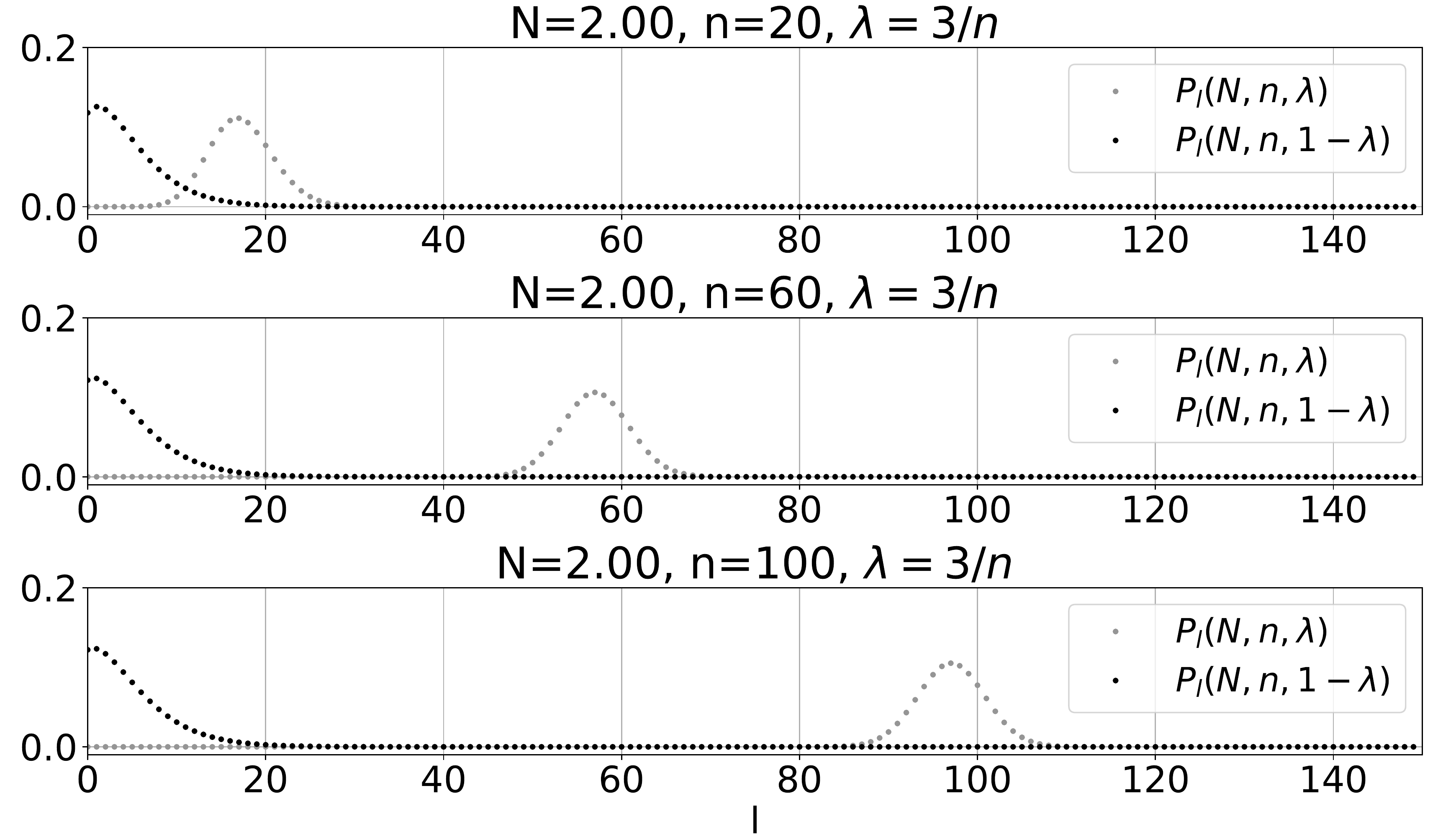}
	\caption{The probability distributions $P_l(2,n,\frac{3}{n})$ and $P_l(2,n,1-\frac{3}{n})$ plotted with respect to the index $l$ for several values of $n$. These probability distributions are computed by using~\eqref{simpler}.}
	\label{ProbIntuition}
\end{figure}
While $\left\{P_l(N,n,1-\frac{c}{n})\right\}_{l\ge 0}$ seems to converge for $n\rightarrow\infty$, the distribution $\left\{P_l(N,n,\frac{c}{n})\right\}_{l\ge 0}$ does not. Indeed, the distribution $\left\{P_l(N,n,\frac{c}{n})\right\}_{l\ge 0}$ can be seen to move to the right as $n$ increases, with a mean value of (see~\eqref{energia_general})
\bb
    \sum_{l=0}^\infty l\,P_l\left(N,n,\frac{c}{n}\right)&=\Tr\left[a^\dagger a\,\Phi_{\frac{c}{n},\ketbrasub{n}}(\tau_N)\right] \\&=\frac{c}{n}N+\left(1-\frac{c}{n}\right)n\stackrel{n\rightarrow\infty}{\simeq}n-c\,,
\ee
keeping constant its peak shape. This is the reason why we are going to unitarily shift to the left the state $\Phi_{\frac{c}{n},\ketbrasub{n}}(\tau_N)$ in order to have a new state which does converge in the limit $n\rightarrow\infty$.

Let us extend our harmonic oscillator Hilbert space by adding new orthonormal states $\{\ket{-i}\}_{i\in\N^+}$ so that $\{\ket{k}\}_{k\in\mathbb{Z}}$ is an orthonormal basis of the resulting Hilbert space $\HH_{\text{ext}}$.
We define the extension of the operators $\Phi_{\frac{c}{n},\ketbrasub{n}}(\tau_N)$ and $\Phi_{1-\frac{c}{n},\ketbrasub{n}}(\tau_N)$ to $\HH_{\text{ext}}$ such that for all $k\in\N^+$ it holds that (for the sake of simplicity we are going to denote an operator and its extension in the same way) 
\bb
\Phi_{\frac{c}{n},\ketbrasub{n}}(\tau_N)\ket{-k}&=0\,,\\
\Phi_{1-\frac{c}{n},\ketbrasub{n}}(\tau_N)\ket{-k}&=0\,.
\ee
Moreover, let us define for all $k\in\mathbb{Z}$ the shift operators $T_k$ which act on the extended Hilbert space as
\begin{equation}
T_{k}\ket{i}=\ket{i+k}\quad\forall\,  i\in\mathbb{Z}\,.
\end{equation}
Since $T_k$ is unitary, it holds that
\begin{equation}
S\left(\Phi_{\frac{c}{n},\ketbrasub{n}}(\tau_N)\right) =S\left(T_{-n}\Phi_{\frac{c}{n},\ketbrasub{n}}(\tau_N)T_{-n}^\dagger\right).
\end{equation}
Consequently,~\eqref{conscons} implies that
\bb\label{conscons2}
&I_{\text{coh}}\left(\Phi_{\frac{c}{n},\ketbrasub{n}},\tau_N\right) \\&=S\left(T_{-n}\Phi_{\frac{c}{n},\ketbrasub{n}}(\tau_N)T_{-n}^\dagger\right)  -S\left(\Phi_{1-\frac{c}{n},\ketbrasub{n}}(\tau_N)\right)\,.
\ee
Unlike the sequence of operators 
\bb
\left\{\Phi_{\frac{c}{n},\ketbrasub{n}}(\tau_N)\right\}_{n\in\N}\subseteq\mathfrak{S}(\HH_S)\,,
\ee
its shifted version, i.e.
\bb
\left\{T_{-n}\Phi_{\frac{c}{n},\ketbrasub{n}}(\tau_N)T_{-n}^\dagger\right\}_{n\in\N}\subseteq\mathfrak{S}(\HH_{\text{ext}})\,,
\ee
does converge in trace norm as $n\rightarrow\infty$, as we are going to prove now.

\eqref{phia} implies that
\bb\label{shiftedState}
T_{-n}\Phi_{\frac{c}{n},\ketbrasub{n}}(\tau_N)T_{-n}^\dagger&=\sum_{l=0}^{\infty}P_l\left(N,n,\frac{c}{n}\right)\ketbra{l-n} \\&=\sum_{k=-n}^{\infty}P_{n+k}\left(N,n,\frac{c}{n}\right)\ketbra{k}\,.
\ee
Hence, since we want to study $T_{-n}\Phi_{\frac{c}{n},\ketbrasub{n}}(\tau_N)T_{-n}^\dagger$ in the limit $n\rightarrow\infty$, let us define the probability distribution $\{q_k(N,c)\}_{k\in\mathbb{Z}}$ over the alphabet $\mathbb{Z}$ as
\begin{equation}\label{qk}
q_k(N,c)\coloneqq\lim\limits_{n\rightarrow\infty}P_{n+k}\left(N,n,\frac{c}{n}\right)\,.
\end{equation}
We will show below that such limit exists. To this end, let us first consider the case where $k\geq 0$. Then~\eqref{simpler} implies that
\begin{small}
\bb
&q_{-k}(N,c) \\
&= \lim\limits_{n\rightarrow\infty}\frac{\sum_{m=k}^n\left(\frac{c}{n}\right)^{2m-k}\!\left(1\!-\!\frac{c}{n}\right)^{n-m} \! N^{m-k}(N\!+\!1)^m\binom{n}{m}\!\binom{n-k}{n-m}}{(1+N\frac{c}{n})^{2n-k+1}} \\
&=\lim\limits_{n\rightarrow\infty}\frac{1}{(1+N\frac{c}{n})^{2n-k+1}}\sum_{m=0}^{n-k}\left(\frac{c}{n}\right)^{2m+k}\left(1-\frac{c}{n}\right)^{n-m-k} \\
&\hspace{9ex} \times N^{m}(N+1)^{m+k}\binom{n}{m+k}\binom{n-k}{m}  \\&=\lim\limits_{n\rightarrow\infty}\frac{\left(1-\frac{c}{n}\right)^{n-k}\left[c(N+1)\right] ^k}{(1+N\frac{c}{n})^{2n-k+1}} \\
&\hspace{9ex} \times\sum_{m=0}^{n-k}\left(1-\frac{c}{n}\right)^{-m}\left[c\sqrt{N(N+1)}\right]^{2m}\frac{\binom{n}{m+k}\binom{n-k}{m}}{n^{2m+k}} \\
&=e^{-c\left[2N+1\right]}\left[c(N+1)\right]^{k} \\
&\hspace{9ex} \times\lim\limits_{n\rightarrow\infty}\sum_{m=0}^{n-k}\left(1-\frac{c}{n}\right)^{-m}\left[c\sqrt{N(N+1)}\right]^{2m}\frac{\binom{n}{m+k}\binom{n-k}{m}}{n^{2m+k}}\,,\label{seriess}
\ee
\end{small}where in the last line we used the celebrated \emph{Nepero limit} i.e.$\lim\limits_{n\rightarrow\infty}\left(1-\frac{c}{n}\right)^{n}=e^{-c} $.
Now, we are going to invoke \emph{Tannery's theorem} in order to show that we can interchange the limit and the series, i.e.
\bb
	&\lim\limits_{n\rightarrow\infty}\sum_{m=0}^{n-k}\left(1-\frac{c}{n}\right)^{-m}\left[c\sqrt{N(N+1)}\right]^{2m}\frac{\binom{n}{m+k}\binom{n-k}{m}}{n^{2m+k}} \\&=\sum_{m=0}^{\infty}\lim\limits_{n\rightarrow\infty}\left(1-\frac{c}{n}\right)^{-m}\left[c\sqrt{N(N+1)}\right]^{2m}\frac{\binom{n}{m+k}\binom{n-k}{m}}{n^{2m+k}}\,.
\ee
The statement of the Tannery's theorem, which is nothing but a special case of Lebesgue's dominated convergence theorem, is the following. \emph{For all $n\in\N$, let $\{a_m(n)\}_{m\in\N}\subset\R$ be a sequence. Suppose that the limit
\bb
\lim_{n\rightarrow \infty }a_{m}(n)\,
\ee
exists for all $m\in\N$. If there exists a sequence $\{M_m\}_{m\in\N}\subset\R$ such that $\sum_{m=0}^{\infty}M_m<\infty$ and $$|a_m(n)|\le M_m$$ for all $m,n\in\N$, then}
\bb
\lim\limits_{n\rightarrow\infty}\sum _{m=0}^{\infty }a_{m}(n)=\sum_{m=0}^{\infty}\lim_{n\rightarrow \infty }a_{m}(n)\,.
\ee
By setting
\begin{equation}
	\chi_{[0,n-k]}(m)\coloneqq\begin{cases}
	1 & \text{if $m\in[0,n-k]$,} \\
	0 & \text{otherwise.}
	\end{cases}
\end{equation}
one obtains 
\bb
&\lim\limits_{n\rightarrow\infty}\sum_{m=0}^{n-k}\left(1-\frac{c}{n}\right)^{-m}\left[c\sqrt{N(N+1)}\right]^{2m}\frac{\binom{n}{m+k}\binom{n-k}{m}}{n^{2m+k}} \\&=\lim\limits_{n\rightarrow\infty}\sum_{m=0}^{\infty}\chi_{[0,n-k]}(m)\left(1-\frac{c}{n}\right)^{-m} \\&\quad\times\left[c\sqrt{N(N+1)}\right]^{2m}\frac{\binom{n}{m+k}\binom{n-k}{m}}{n^{2m+k}}\,.
\ee
Let us check whether the hypotheses of Tannery's theorem are fulfilled.
Thanks to the inequality $\binom{a}{b}\le\frac{a^b}{b!}$, valid for all integers $a\ge b\ge0$, and to the fact that $(1-\frac{c}{n})^{-m}\le 2^m$, valid for sufficiently large $n$ (more precisely, as soon as $n\ge2c$), the general term of the series (which is non-negative) satisfies
\bb
&\chi_{[0,n-k]}(m)\left(1-\frac{c}{n}\right)^{-m}\left[c\sqrt{N(N+1)}\right]^{2m}\frac{\binom{n}{m+k}\binom{n-k}{m}}{n^{2m+k}} \\&\le \frac{\left[c\sqrt{2N(N+1)}\right]^{2m}}{(m+k)!m!}\,,
\ee
for sufficiently large $n$. Since~\footnote{Notice that $\sum_{m=0}^{\infty}\frac{\left[c\sqrt{2N(N+1)}\right]^{2m}}{(m+k)!m!}<\exp\left[2N(N+1)c^2\right]<\infty$.} 
\bb
\sum_{m=0}^{\infty}\frac{\left[c\sqrt{2N(N+1)}\right]^{2m}}{(m+k)!m!}<\infty\,,
\ee
Tannery's theorem guarantees that
\bb
&\lim\limits_{n\rightarrow\infty}\sum_{m=0}^{n-k}\left(1-\frac{c}{n}\right)^{-m}\left[c\sqrt{N(N+1)}\right]^{2m}\frac{\binom{n}{m+k}\binom{n-k}{m}}{n^{2m+k}} \\&=\sum_{m=0}^{\infty}\lim\limits_{n\rightarrow\infty}\chi_{[0,n-k]}(m)\left(1-\frac{c}{n}\right)^{-m} \\&\quad\times\left[c\sqrt{N(N+1)}\right]^{2m}\frac{\binom{n}{m+k}\binom{n-k}{m}}{n^{2m+k}} \\&=\sum_{m=0}^{\infty}\lim\limits_{n\rightarrow\infty}\left(1-\frac{c}{n}\right)^{-m}\left[c\sqrt{N(N+1)}\right]^{2m}\frac{\binom{n}{m+k}\binom{n-k}{m}}{n^{2m+k}}\\&=\sum_{m=0}^\infty\frac{\left[c\sqrt{N(N+1)}\right]^{2m}}{(m+k)!m!}\,.\label{speroult}
\ee
As a consequence,~\eqref{seriess} implies that
\begin{equation}\label{primaeq}
q_{-k}(N,c)= e^{-c\left[2N+1\right]}\left[c(N+1)\right]^{k}\sum_{m=0}^{\infty}\frac{\left[c\sqrt{N(N+1)}\right]^{2m}}{(m+k)!m!}\,.
\end{equation}
By using~\eqref{besself}, we arrive at
\begin{equation}
q_{-k}(N,c)= e^{-c\left[2N+1\right]}\left(\frac{N+1}{N}\right)^{k/2}I_k\left(2c\sqrt{N(N+1)}\right)\,.
\end{equation}
Analogously, for $k\ge0$ it holds that:
\begin{small}
\bb
&q_{k}(N,c) \\
&=\lim\limits_{n\rightarrow\infty}\frac{\sum_{m=0}^n\left(\frac{c}{n}\right)^{2m+k}\left(1-\frac{c}{n}\right)^{n-m}N^{m+k}(N+1)^m\binom{n}{m}\binom{n+k}{n-m}}{(1+N\frac{c}{n})^{2n+k+1}} \\
&=\lim\limits_{n\rightarrow\infty}\frac{1}{(1+N\frac{c}{n})^{2n+k+1}}\left(1-\frac{c}{n}\right)^{n}\left(cN\right)^k \\
&\quad\times\sum_{m=0}^{n}\left(1-\frac{c}{n}\right)^{-m}\left[c\sqrt{N(N+1)}\right]^{2m}\frac{\binom{n}{m}\binom{n}{m+k}}{n^{2m+k}} \\
&=e^{-c\left[2N+1\right]}\left(cN\right)^{k}\sum_{m=0}^{\infty}\frac{\left[c\sqrt{N(N+1)}\right]^{2m}}{(m+k)!m!} \\
&=e^{-c\left[2N+1\right]}\left(\frac{N}{N+1}\right)^{k/2}I_k\left(2c\sqrt{N(N+1)}\right)\,.&
\ee
\end{small}In summary, we have shown that for all $k\in\mathbb{Z}$ it holds that
\begin{equation}
q_k(N,c)=e^{-c\left[2N+1\right]}\left(\frac{N}{N+1}\right)^{k/2}I_{|k|}\left(2c\sqrt{N(N+1)}\right)\,.
\end{equation}
Now, let us define the following state on $\HH_{\text{ext}}$:
\begin{equation}
\rho_{q(N,c)}\coloneqq\sum_{k=-\infty}^{\infty}q_{k}(N,c)\ketbra{k}\,.
\end{equation}
\eqref{shiftedState} and~\eqref{qk} guarantee that the sequence of density operators $\left\{T_{-n}\Phi_{\frac{c}{n},\ketbrasub{n}}(\tau_N)T_{-n}^\dagger\right\}_{n\in\N}$ weakly converges to $\rho_{q(N,c)}$. Consequently, we can apply~\cite[Lemma 4.3]{Davies1969}, which states that: \emph{if a sequence of density operators converges to another density operators in the weak operator topology, then the convergence is in trace norm}.  Hence,
$\left\{T_{-n}\Phi_{\frac{c}{n},\ketbrasub{n}}(\tau_N)T_{-n}^\dagger\right\}_{n\in\N}$ converges to $\rho_{q(N,c)}$ in trace norm i.e.
\begin{equation}\label{trconvq}
\lim\limits_{n\rightarrow\infty}\left\| T_{-n}\Phi_{\frac{c}{n},\ketbrasub{n}}(\tau_N)T_{-n}^\dagger-\rho_{q(N,c)}\right\|_1=0\,.
\end{equation}
Analogously, we define the probability distribution $\{p_k(N,c)\}_{k\in\mathbb{Z}}$ over the alphabet $\mathbb{Z}$ as
\begin{equation}\label{pk}
p_k(N,c)\coloneqq
\begin{cases}
\lim\limits_{n\rightarrow\infty}P_k\left(N,n,1-\frac{c}{n}\right), & \text{if $k\ge 0$,} \\
0, & \text{otherwise.}
\end{cases}
\end{equation}
Let $k\ge0$.  Analogously to what has been done previously i.e.~starting from~\eqref{simpler}, carrying out the calculations and expanding for large $n$, one obtains that:  
\begin{small}
\bb
&p_k(N,c) \\
&\quad =\lim\limits_{n\rightarrow\infty}\frac{1}{\left(N+1-N\frac{c}{n}\right)^{k+n+1}}\sum_{m=n-k}^n\left(1-\frac{c}{n}\right)^{2m+k-n}  \\
&\quad\times\left(\frac{c}{n}\right)^{n-m}N^{k+m-n}(N+1)^m\binom{n}{m}\binom{k}{n-m} \\&=\lim\limits_{n\rightarrow\infty}\frac{1}{(N+1)^{k+n+1}\left(1-\frac{Nc}{(N+1)n}\right)^{k+n+1}}\left(1-\frac{c}{n}\right)^{n+k}  \\
&\quad\times\sum_{m=0}^k\left(1-\frac{c}{n}\right)^{-m}		\left(\frac{c}{n}\right)^{m}N^{k-m}(N+1)^{n-m}\binom{n}{m}\binom{k}{m}	 \\
&=\frac{N^k}{(N+1)^{k+1}}e^{-c/(N+1)}\sum_{m=0}^k	\left(\frac{c}{N(N+1)}\right)^{m}\frac{1}{m!}\binom{k}{m} \\
&=\frac{N^k}{(N+1)^{k+1}}e^{-c/(N+1)}L_k\left(-\frac{c}{N(N+1)}\right)\,.
\ee
\end{small}Moreover, let us define
\begin{equation}
	\rho_{p(N,c)}\coloneqq \sum_{k=0}^{\infty}p_{k}(N,c)\ketbra{k}\,.
\end{equation}
Since $\left\{\Phi_{1-\frac{c}{n},\ketbrasub{n}}(\tau_N)\right\}_{n\in\N}$ weakly converges to $\rho_{p(N,c)}$, then it converges to $\rho_{p(N,c)}$ in trace norm~\cite[Lemma 4.3]{Davies1969} i.e.
\begin{equation}\label{trconvp}
\lim\limits_{n\rightarrow\infty}\left\| \Phi_{1-\frac{c}{n},\ketbrasub{n}}(\tau_N)-\rho_{p(N,c)}\right\|_1=0\,.
\end{equation}
From~\eqref{conscons2} we have that:
\bb \label{conscons3}
&\liminf_{n\rightarrow\infty} I_{\text{coh}}\left(\Phi_{\frac{c}{n},\ketbrasub{n}},\tau_N\right)\\
&\qquad = \liminf_{n\rightarrow\infty}\left[S\left(T_{-n}\Phi_{\frac{c}{n},\ketbrasub{n}}(\tau_N)T_{-n}^\dagger\right)  \right. \\
& \hspace{20ex} \left.-S\left(\Phi_{1-\frac{c}{n},\ketbrasub{n}}(\tau_N)\right)\right]\,.
\ee
At this point we would like to show that we can lower bound the expression~\eqref{conscons3} by taking the limit inside the function $S$. For this to be a legal move, we would need the entropy to be a continuous function of the states we consider. However, although in the finite-dimensional scenario Fannes' inequality~\cite[Chapter~11]{NC} guarantees the continuity in trace norm of the von Neumann entropy~\cite{Fannes1973}, in the infinite-dimensional setting the continuity does not hold any more. However, the von Neumman entropy is still lower semi-continuous in trace norm~\cite[Theorem 11.6]{HOLEVO-CHANNELS-2}.
Continuity of the von Neumann entropy is restored when restricting to quantum systems satisfying the Gibbs hypothesis (i.e.\ quantum systems having finite partition function) and to states with 
bounded energy, as established by~\cite[Lemma 11.8]{HOLEVO-CHANNELS-2}.
First, we are going to apply~\cite[Lemma 11.8]{HOLEVO-CHANNELS-2} to the sequence $\left\{\Phi_{1-\frac{c}{n},\ketbrasub{n}}(\tau_N)\right\}_{n\in\N}$. Let us check whether its hypotheses are fulfilled.
\begin{itemize}
	\item The partition function of the quantum harmonic oscillator is finite for all $\beta>0$, indeed:
	$$\Tr_S\left[e^{-\beta a^\dagger a}\right]=\frac{1}{1-e^{-\beta}}<\infty\quad\text{for all }\beta>0\,.$$
	\item The sequence $\left\{\Phi_{1-\frac{c}{n},\ketbrasub{n}}(\tau_N)\right\}_{n\in\N}\subseteq\mathfrak{S}(\HH_S)$ has 
	bounded energy. Indeed,~\eqref{energia_general} implies that for all $n\in\N$ it holds that 
	$$\Tr_S\left[\Phi_{1-\frac{c}{n},\ketbrasub{n}}(\tau_N)\,a^\dagger a\right]=\left(1-\frac{c}{n}\right)N+c\le N+c\,.$$
\end{itemize}
As a consequence, we can apply~\cite[Lemma 11.8]{HOLEVO-CHANNELS-2}. By using also~\eqref{trconvp}, we deduce that
\bb
\lim\limits_{n\rightarrow\infty}S\left(\Phi_{1-\frac{c}{n},\ketbrasub{n}}(\tau_N)\right)=S\left(\rho_{p(N,c)}\right)=H\left(p(N,c)\right)\,.
\ee
Then,~\eqref{conscons3} implies 
\bb\label{conscons4}
&\liminf_{n\rightarrow\infty} I_{\text{coh}}\left(\Phi_{\frac{c}{n},\ketbrasub{n}},\tau_N\right) \\&=\liminf_{n\rightarrow\infty}S\left(T_{-n}\Phi_{\frac{c}{n},\ketbrasub{n}}(\tau_N)T_{-n}^\dagger\right)  -H\left(p(N,c)\right)\,.
\ee
Second, for the sequence $\left\{T_{-n}\Phi_{\frac{c}{n},\ketbrasub{n}}(\tau_N)T_{-n}^\dagger\right\}_{n\in\N}\subseteq\mathfrak{S}(\HH_{\text{ext}})$, we apply the lower semi-continuity of the von Neumann entropy and~\eqref{trconvq} to deduce that
\bb
    \liminf\limits_{n\rightarrow\infty}S\left(T_{-n}\Phi_{\frac{c}{n},\ketbrasub{n}}(\tau_N)T_{-n}^\dagger\right)&\ge S(\rho_{q(N,c)}) \\&=H\left(q(N,c)\right) \,.
\ee
Hence, by substituting in~\eqref{conscons4}, we finally obtain~\eqref{entdifflow}.
\end{proof}

\subsection{The capacities are not necessarily monotonic in the transmissivity} \label{subsec_mon}

The results we have discussed in Section~\ref{subsec_Cea} (i.e. the implications of~\eqref{SpecialCea0}) are valid for sufficiently small values of $\lambda>0$, which is the physically interesting case where only an infinitesimal fraction of the output energy comes from the input and where we expect the channel to be noisier than what it would be for larger values of $\lambda$. Therefore, one can be tempted to conclude that such results hold not only for sufficiently small values of $\lambda>0$ but also for all $\lambda>0$. However, this reasoning does not quite work as stated since it assumes that the capacities are monotonic in $\lambda$. This is not obvious and, as a matter of fact, it is already known that the quantum capacity is not necessarily monotonically increasing in the transmissivity for a fixed environment state~\cite{die-hard}.
Analogously, \emph{$C_{\text{ea}}\left(\Phi_{\lambda,\sigma},N\right)$ can fail to be monotonically increasing in $\lambda$ for fixed $\sigma$ and $N$}. Indeed:
\begin{itemize}
     \item If $\lambda=0$, the general attenuators are completely noisy and hence $C_{\text{ea}}\left(\Phi_{0,\ketbrasub{n}},N\right)=0$ for all $N>0$, $n\in\N$;
 	 \item If $\lambda=1/2$, as we will see later thanks to Theorem~\ref{teofirst}, it holds that 
 	 \begin{equation}\label{onehalf}
 	     C_{\text{ea}}\left(\Phi_{1/2,\ketbrasub{n}},N\right)=g(N)
 	 \end{equation}
 	 for all $N>0$, $n\in\N$;
 	 \item As guaranteed by~\eqref{SpecialCea0} (or also by Fig.~\ref{Icoh figure}), for some $\lambda\in(0,1/2)$, $n\in\N$, $N>0$, it holds that $C_{\text{ea}}\left(\Phi_{\lambda,\ketbrasub{n}},N\right)>g(N)$ .
\end{itemize}
Numerical investigations suggest that $C_{\text{ea}}\left(\Phi_{\lambda,\ketbrasub{n}},N\right)$ is non-monotonic in $\lambda$ if $n\gg N$. Notice that monotonicity still holds when the environment state is a thermal state, as it can be shown by exploiting~\eqref{func}. A more elegant way to prove the monotonicity in $\lambda$ of $C_{\text{ea}}\left(\Phi_{\lambda,\tau_\nu},N\right)$ is as follows. Lemma~\ref{thermal_composition} in the Appendix establishes that for all $\lambda_1\text{, }\lambda_2\in[0,1]$ and all $N>0$ it holds that
\begin{equation}
\Phi_{\lambda_1\lambda_2,\tau_N}=\Phi_{\lambda_1,\tau_N}\circ\Phi_{\lambda_2,\tau_N}\,.
\end{equation}
By using this, for all $0<\lambda_1<\lambda_2<1$ we deduce that $\Phi_{\lambda_1,\tau_\nu}=\Phi_{\lambda_1/\lambda_2,\tau_\nu}\circ\Phi_{\lambda_2,\tau_\nu}$. Hence, the channel $\Phi_{\lambda_1,\tau_\nu}$ can be simulated by $\Phi_{\lambda_2,\tau_\nu}$ if Bob applies $\Phi_{\lambda_1/\lambda_2,\tau_\nu}$ on the received quantum carrier. As a consequence, it holds that every capacity of $\Phi_{\lambda_1,\tau_\nu}$ is lower than that of $\Phi_{\lambda_2,\tau_\nu}$. Hence, all the capacities of the thermal attenuator are monotonic in $\lambda$. In addition, the previous reasoning suggests that the composition law $\Phi_{\lambda_1\lambda_2,\ketbrasub{n}}=\Phi_{\lambda_1,\ketbrasub{n}}\circ\Phi_{\lambda_2,\ketbrasub{n}}$ does not hold for $n\ge1$.

Now, we just have to prove~\eqref{onehalf}. To do this, let us show the following more general theorem.
\begin{thm}\label{teofirst}
	Let $\alpha\in\mathbb{C}$. Let $\sigma\in\mathfrak{S}(\HH_E)$ be of the form $\sigma=D(\alpha)\sigma_0{D(\alpha)}^\dagger$, where $V \sigma_0 V^\dagger=\sigma_0$ with $V \coloneqq (-1)^{b^\dagger b}$.
	Then for all $N>0$ it holds that
	\begin{equation}\label{bounds}
	g(N)-S(\sigma)\le C_{\text{ea}}\left(\Phi_{1/2,\sigma},N\right)\le g(N)\,.
	\end{equation}
	If $\sigma$ is also pure, it follows in particular that
	\begin{equation}\label{lambda1/2}
	C_{\text{ea}}\left(\Phi_{1/2,\sigma},N\right)= g(N)\,.
	\end{equation}
\end{thm}
\begin{proof}
	Let $N>0$ and let $\rho\in\mathfrak{S}(\HH_S)$ such that $\Tr[a^\dagger a\,\rho]\le N$. Since the hypotheses are the same as Lemma~\ref{lemma_diehard}, we can apply~\eqref{weak_formula}. The latter, together with the invariance of the von Neumann entropy under unitary transformations, guarantees that
	\begin{equation}\label{forweak}
	S\left(\tilde{\Phi}_{\lambda,\sigma}^{\text{wc}}(\rho)\right)= S\left(\Phi_{1-\lambda,\sigma}(\rho)\right)\,.
	\end{equation}
	Thanks to the subadditivity of the von Neumann entropy, one can prove that
	\begin{equation}\label{ineq_compl}
	    S\left(\tilde{\Phi}_{\lambda,\sigma}(\rho)\right)\le S(\sigma)+S\left(\tilde{\Phi}_{\lambda,\sigma}^{\text{wc}}(\rho)\right)\,.
	\end{equation}
	A proof of~\eqref{ineq_compl} can be found in Lemma~\ref{subadd} in the Appendix.
	\eqref{ineq_compl} implies that
	\bb
	&I(\Phi_{1/2,\sigma},\rho)=S(\rho)+S\left(\Phi_{1/2,\sigma}(\rho)\right)-S\left(\tilde{\Phi}_{1/2,\sigma}(\rho)\right) \\&\ge S(\rho)+S\left(\Phi_{1/2,\sigma}(\rho)\right)-S(\sigma)-S\left(\tilde{\Phi}_{1/2,\sigma}^{\text{wc}}(\rho)\right) \\&=S(\rho)-S(\sigma)\,,
	\ee 
	where in the last line we used~\eqref{forweak}.
	Therefore, Lemma~\ref{maxthermstate} concludes the proof of 
	\begin{equation}
	    C_{\text{ea}}\left(\Phi_{1/2,\sigma},N\right)\ge g(N)-S(\sigma)\,.
	\end{equation}	Now, we just have to prove the upper bound. Lemma~\ref{lemma_diehard} guarantees that $Q\left(\Phi_{1/2,\sigma}\right)=0$ and hence for all $\rho$ it holds $I_{\text{coh}}\left(\Phi_{1/2,\sigma},\rho\right)\le 0$. As a consequence, $I\left(\Phi_{1/2,\sigma},\rho\right)\le S(\rho)$.
	Taking the maximum over all $\rho$ satisfying the energy constraint $\Tr[a^\dagger a\,\rho]\le N$, this implies
	\begin{equation}
	    C_{\text{ea}}\left(\Phi_{1/2,\sigma},N\right)\le g(N).
	\end{equation}
	Finally, since the von Neumann entropy of a pure state is zero,~\eqref{lambda1/2} immediately follows from the upper and lower bounds just proved.
\end{proof}
The states $\sigma$ satisfying the hypothesis of Theorem~\ref{teofirst} are of the form
	$\sigma= D(\alpha)\sigma_0{D(\alpha)}^{\dagger}
	,$ where $\sigma_0$ and $V$ commute and hence admit a common basis of eigenstates. Hence, we can write 
	\begin{equation}
	    \sigma = D(\alpha)\left(\sum_i p_i \ketbra{\psi_i}\right){D(\alpha)}^{\dagger}\,,
	\end{equation}
	where:
	\begin{itemize}
		\item $\{p_i\}_{i\in\N}$ is a probability distribution (since $\sigma_0$ is a density operator);
		\item Either $\ket{\psi_i}$ is an \emph{even state} i.e.
		$$V\ket{\psi_i}=\ket{\psi_i}\Longleftrightarrow \ket{\psi_i}=\sum_{n=0}^{\infty} c^{(i)}_n\ket{2n}\,,$$
		or it is an \emph{odd state} i.e.
		$$V\ket{\psi_i}=-\ket{\psi_i}\Longleftrightarrow \ket{\psi_i}=\sum_{n=0}^{\infty} c^{(i)}_n\ket{2n+1}\,.$$
	\end{itemize}
	To summarise, if $\sigma$ is a displaced convex combination of odd and even states (e.g.\ displaced thermal states) then 
	\begin{equation}
	    g(N)-S(\sigma)\le C_{\text{ea}}\left(\Phi_{1/2,\sigma},N\right)\le g(N)\,.
	\end{equation}
	In addition, if $\sigma$ is a displaced odd or even pure state (e.g.\ a (displaced) Fock state, coherent state, squeezed state, or cat state) then 
	\begin{equation}
	    C_{\text{ea}}\left(\Phi_{1/2,\sigma},N\right)= g(N)\,.
	\end{equation}
	This result is in agreement with the known expression for the pure loss channel $\Phi_{1/2,\ketbra{0}}$ (see~\eqref{pureloss}). Since $C\left(\Id,N\right)=g(N)$, for such a class of environment states $\sigma$ it holds that \begin{equation}
	    C_{\text{ea}}\left(\Phi_{1/2,\sigma},N\right)= C\left(\Id,N\right)\,.
	\end{equation} In other words, we can state that: for $\lambda=1/2$ if the environment is in a displaced odd or even pure state, the resource of pre-shared entanglement allows classical information to be transmitted with the same performance as the unassisted noiseless scenario. 
	Previously, we have seen that a stronger property holds: for $0<\lambda<1/2$, if the environment state is a suitable Fock state, entanglement-assisted classical communication is possible with better performance than in the unassisted noiseless scenario. 
\section{Control of the environment state}\label{sec_control}
In the best studied models of communication, one usually employs the \emph{memoryless approximation}. In our scenario, under this approximation, the environment state is always the same every time the transmission line is used. However, if Alice feeds signals into the line separated by a sufficiently short temporal interval, the memoryless approximation is challenged and memory effects arise, thereby altering the environment state. From the above results we know that the manipulation of the environment state can greatly improve the communication performance. Consequently, a natural question that arises is: how to implement in a realistic and operational way the control of the environment state?

Now, we introduce a communication model which takes into account memory effects. Then, we will explain how to implement a \emph{partial} manipulation of the environment state that however suffices to our purposes, i.e.\ that of boosting the capacities.

We analyse the model formulated by~\cite{Dynamical-Model} and depicted in Fig.~\ref{comm_model}, which consists of three elements: the signals, the channel environment $E$, and also a huge reservoir $R$.
\begin{figure}[t]
\centering
\includegraphics[scale=0.20]{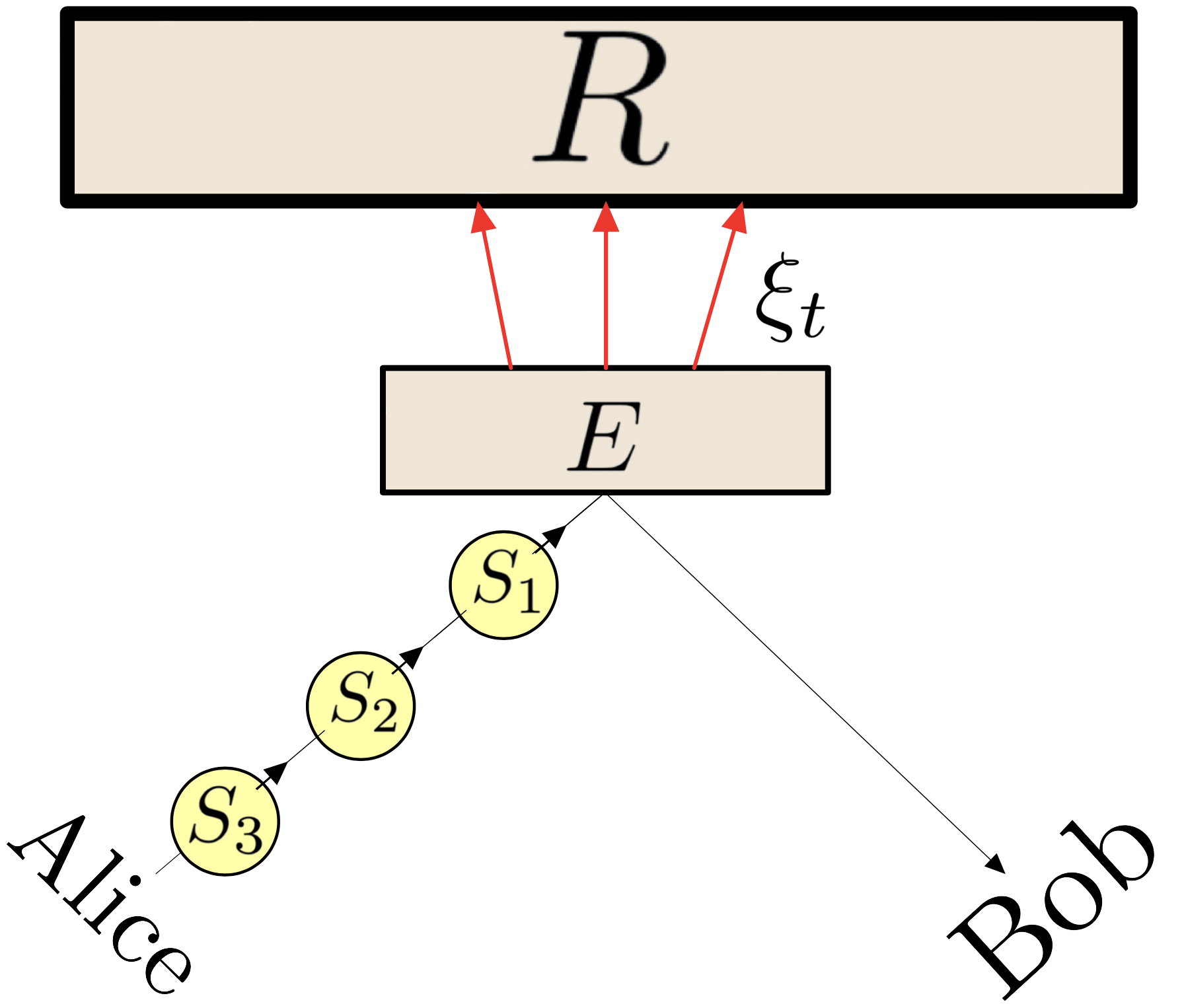}
\caption{Schematic of the communication model. The reservoir $R$ is coupled with the environment $E$, but not with the signals $S_i$. The environment $E$ is perturbed also by $S_i$.}
\label{comm_model}
\end{figure}
The signals and the environment $E$ are single-mode systems. The signals interact one at a time with the environment $E$ through the {BS} unitary, while $E$ undergoes a dissipative evolution induced by $R$. The dissipative evolution is implemented by a one-parameter family of quantum channels $\{\xi_{\delta t}\}_{\delta t\ge0}$, which transforms any state $\sigma$ into a stationary state $\sigma_0$ after a thermalisation time $t_E$. We request that 
{\begin{alignat}{3}
\xi_{\delta t}(\sigma)&&&\simeq\sigma\qquad &&\text{if $\delta t\ll t_E$ ,}\label{assumption1xi} \\
\xi_{\delta t}(\sigma)&&&=\sigma_0\qquad &&\text{if $\delta t\ge t_E$ ,} \label{assumption2xi}
\end{alignat}}
for all $\sigma\in\mathfrak{S}(\HH_E)$.
An optical fibre is described by the above model with $\sigma_0$ being a thermal state $\tau_\nu$.

{Let us justify the assumption in~\eqref{assumption2xi}. In the best studied model of bosonic quantum communication across optical fibres, the attenuation noise which affects each signal is given by the same quantum channel, i.e.~by the thermal attenuator $\Phi_{\lambda,\tau_\nu}$. This is valid when the memoryless assumption is justified, i.e.\ when the time interval $\delta t$ between two subsequent signals is sufficiently large. Within our model, this translates to the assumption that when a signal is transmitted across the fibre, the thermalisation process $\xi_{\delta t}$ has already \emph{exactly} brought back the environment state into the thermal state $\tau_\nu$. Mathematically, this means that if $\delta t$ is larger than a certain timescale --- which we call the thermalisation time $t_E$ --- it holds that $\xi_{\delta t}(\sigma)= \tau_\nu$ for all $\sigma$. To summarise, if we require that our model reduces to the best studied model of bosonic quantum communication across optical fibres, the assumption in~\eqref{assumption2xi} is necessary.}

If Alice is able to send signals only separated by a temporal interval {$\delta t\ge t_E$, }{}
then the memoryless approximation is recovered and communication occurs via $\Phi_{\lambda,\sigma_0}$. Vice versa, if the temporal interval between signals is shorter than $t_E$, then memory effects arise. And, if the interval is much shorter than $t_E$, then the dissipative evolution is negligible and the environment state is only modified by the interaction with the signals.

{{\begin{ex}
Suppose that Alice sends $k$ signals $S_1,S_2,\ldots,S_k$, separated by a time interval $\delta t$ and initialised in a state $\rho^{(k)}$. Then the environment state $\sigma$, right after a time $\delta t$ that the $k$th signal has interacted with the environment, can be expressed as
	\bb\label{expl_sigma}
    	\sigma &= \Tr_{S_1\ldots S_k}\left[\xi_{\delta t}\circ \mathcal{U}_{\lambda}^{(S_kE)}\circ \xi_{\delta t}\circ\mathcal{U}_{\lambda}^{(S_{k-1}E)}\circ\ldots \right.  \\
    &\left.\qquad\ldots \circ \xi_{\delta t}\circ\mathcal{U}_{\lambda}^{(S_{1}E)} \left(\rho^{(k)}\otimes \sigma_0\right) \right]\,,
    \ee
	where $\mathcal{U}_{\lambda}^{(S_iE)}$ is a quantum channel defined by $\mathcal{U}_{\lambda}^{(S_iE)}(\cdot)=U_{\lambda}^{(S_iE)}(\cdot) \left(U_{\lambda}^{(S_iE)}\right)^\dagger$.

Now, suppose that once Alice has sent the signal $S_k$, she waits for a time $\delta t$ and sends another signal $S$ initialised in $\rho$. The output state of $S$ that Bob receives is 
	\bb\label{out_of_env}
    	&\Tr_{S_1\ldots S_kE}\left[\mathcal{U}_{\lambda}^{(SE)}\circ\xi_{\delta t}\circ \mathcal{U}_{\lambda}^{(S_kE)}\circ \xi_{\delta t}\circ\mathcal{U}_{\lambda}^{(S_{k-1}E)}\circ\ldots \right.  \\
    &\left.\qquad\ldots \circ \xi_{\delta t}\circ\mathcal{U}_{\lambda}^{(S_{1}E)} \left(\rho\otimes\rho^{(k)}\otimes \sigma_0\right) \right]\\&=\Tr_E\left[\mathcal{U}_{\lambda}^{(SE)}\left(\rho\otimes\xi_{\delta t}(\sigma)\right)\right]=\Phi_{\lambda,\xi_{\delta t}(\sigma)}(\rho)\,.
    \ee
\end{ex}}}
Let us put these ideas into practice by introducing a protocol dubbed `noise attenuation protocol'. This protocol aims to obtain an effective memoryless channel $\Phi_{\lambda,\sigma}$ that is less noisy than $\Phi_{\lambda,\sigma_0}$. More precisely, the aim of the noise attenuation protocol is to obtain an environment state $\sigma$ such that the capacity of interest of $\Phi_{\lambda,\sigma}$ is larger than that of $\Phi_{\lambda,\sigma_0}$. The basic idea of the protocol is that Alice first attempts to manipulate the environment state by sending signals that do not encode information, and then transmits the actual message.

\bigskip
\textbf{\emph{Noise attenuation protocol}}:
\begin{itemize}
        \item \emph{step~1}: Alice waits for a time $t_E$ (in order to have the environment reset into $\sigma_0$);
		\item \emph{step~2}: Alice sends $k$ signals 
		so that, after the interaction with the environment $E$ has occurred, the latter transforms into a chosen state $\sigma$;
		\item \emph{step~3}: Alice sends the signal which encodes information. Then, she restarts from step~1 until the communication is complete.
\end{itemize}
Any signal sent during step~2 is dubbed `trigger signal'. Alice can encode information into a multipartite entangled state of the signals sent at step~3. Step~1 allows one to treat the input-output relations of the signals of step~3 as implemented by the memoryless channel $\Phi_{\lambda,\sigma}$. Hence, Alice and Bob have to apply encoding and decoding strategies to communicate via $\Phi_{\lambda,\sigma}$, instead of via $\Phi_{\lambda,\sigma_0}$. Bob has to take into account in his decoding strategy only the signals sent during step~3, throwing away the received trigger signals. We give a simple example where the noise attenuation protocol is beneficial in terms of entanglement-assisted communication in {the} Appendix~\ref{example_noiseatt}.

Let $\HH_{S_i}$ denote the Hilbert space of the $i$-th trigger signal and let $a_i$ denote the corresponding annihilation operator.  We will refer to $\sigma_0$ as the `stationary environment state', to the signal of 
step~3 as an `information-carrier signal', to $\Phi_{\lambda,\sigma_0}$ as the `original channel' and to $\Phi_{\lambda,\sigma}$ as the `resulting channel'.
A pictorial schematic of the noise attenuation protocol is provided in Fig.~\ref{protocol}.
\begin{figure}[t]
\centering
\includegraphics[scale=0.17]{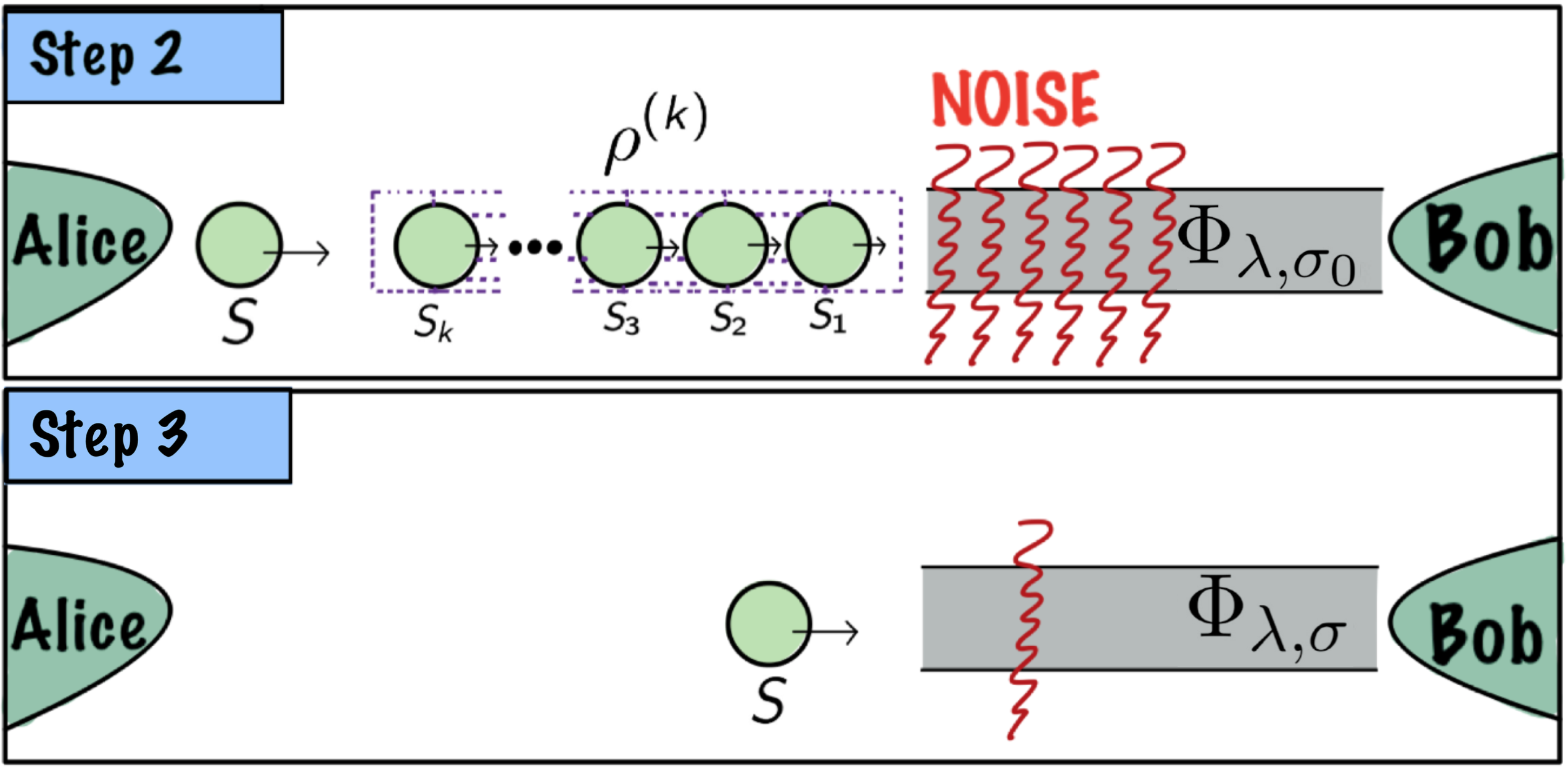}
\caption{Alice alternates between sequences of trigger signals (step~2), which do not carry information but are sacrificed to turn the environment into a state that facilitates communication, and information-carrying signals (step~3).}
\label{protocol}
\end{figure}

{
Suppose that the $k$ trigger signals are initialised in $\rho^{(k)}\in\mathfrak{S}(\HH_{S_{1}}\otimes\ldots\otimes\HH_{S_{k}})$ and are separated by a time interval $\delta t$. Suppose further that the time interval between the $k$th trigger signal and the information-carrier signal is $\delta t$. Consequently, at the beginning of step~3 the environment is in a state $\sigma$ given by~\eqref{expl_sigma}. From now on, suppose that $\delta t \ll t_E$. In this regime we can use the approximation $\xi_{\delta t}\simeq\Id$. In Theorem~\ref{teo_new} we will formalise such an approximation and we will see that it is consistent. 
By exploiting~\eqref{expl_sigma}, $\sigma$ can be approximated as}
 \bb
\sigma=\Delta_{\lambda,\sigma_0}^{(k)}(\rho^{(k)})\,,
\ee
where $\Delta_{\lambda,\sigma_0}^{(k)}:\mathfrak{S}(\HH_{S_{1}}\otimes\ldots\otimes\HH_{S_{k}})\mapsto\mathfrak{S}(\HH_{E})$
is defined by 
\bb\label{achievable}
    \Delta_{\lambda,\sigma_0}^{(k)}(\rho^{(k)})\coloneqq&\Tr_{S_1S_2\ldots S_k}  \left[U_{\lambda}^{(S_kE)}\ldots U_{\lambda}^{(S_1E)}\rho^{(k)}\otimes\sigma_0 \right.  \\
  &\left. \left(U_{\lambda}^{(S_1E)}\right)^\dagger\ldots     \left(U_{\lambda}^{(S_kE)}\right)^\dagger\right]\,.
\ee
Hence, the information-carrier signal is affected by the channel $\Phi_{\lambda,\Delta_{\lambda,\sigma_0}^{(k)}(\rho^{(k)})}$.
\begin{definition}
An environment state $\sigma\in\mathfrak{S}(\HH_{E})$ is said to be $(\lambda,\sigma_0)$-achievable if there exists $k\in\N$ and $\rho^{(k)}\in\mathfrak{S}(\HH_{S_{1}}\otimes\ldots\otimes\HH_{S_{k}})$ such that $ \sigma=\Delta_{\lambda,\sigma_0}^{(k)}(\rho^{(k)})$.
\end{definition}
Not all environment states can be obtained by Alice, but only those which are $(\lambda,\sigma_0)$-achievable.
 In order to activate the special communication performance discussed in Sec.~\ref{subsec_Cea} (e.g.\ obtaining a resulting channel with strictly positive quantum capacity, or with entanglement-assisted classical capacity larger than that of the noiseless channel), a natural question that arises is whether there exists a suitable choice of the trigger-signals state $\rho^{(k)}$ that creates an environmental Fock state. In other words, we may wonder whether the Fock states are $(\lambda,\sigma_0)$-achievable. Unfortunately, this is not the case if the stationary state $\sigma_0$ is thermal (as in the scheme of an optical fibre), except in the trivial case in which both $\sigma_0$ and the Fock state are the vacuum, as guaranteed by the forthcoming Theorem~\ref{NoTheorem}. Before stating Theorem~\ref{NoTheorem}, let us prove the following lemma.
\begin{lemma}\label{bHeisenberg}
For all $\lambda\in(0,1)$, $k\in\N$, it holds that
\bb\label{trb}
  b_H^{(k)}&\coloneqq\left(U_{\lambda}^{(S_1E)}\right)^\dagger\ldots\left(U_{\lambda}^{(S_kE)}\right)^\dagger b \,U_{\lambda}^{(S_kE)}\ldots U_{\lambda}^{(S_1E)} \\&=-\sqrt{1-\lambda^k}h_{\lambda,k}+\lambda^{k/2} b\,,
\ee
where
\begin{equation}
    h_{\lambda,k}\coloneqq \sqrt{\frac{1-\lambda}{1-\lambda^k}} \sum_{l=1}^k \lambda^{\frac{k-l}{2}}a_l\,,
\end{equation}
with $\left[h_{\lambda,k},\left(h_{\lambda,k}\right)^\dagger\right]=\mathbb{1}$.
\end{lemma}
\begin{proof}
Let us prove this by induction. 
The case $k=1$ is true since $h_{\lambda,1}=a_1$ and since it holds that
\begin{equation}
    \left({U_{\lambda}^{(S_1E)}}\right)^\dagger b\, U_{\lambda}^{(S_1E)}=-\sqrt{1-\lambda}a_1+\sqrt{\lambda}b\,,
\end{equation}
as established by Lemma~\ref{lemmatrasf} in the Appendix.
Assuming~\eqref{trb} to be true for $k$, we have to prove it for $k+1$. By using that
\begin{equation}
    \left({U_{\lambda}^{(S_{k+1}E)}}\right)^\dagger b\, U_{\lambda}^{(S_{k+1}E)}=-\sqrt{1-\lambda}a_{k+1}+\sqrt{\lambda}b\,
\end{equation}
and by exploiting the inductive hypothesis, one obtains that
	\bb
b_H^{(k+1)} &=-\sqrt{1-\lambda}a_{k+1}+\sqrt{\lambda}b_H^{(k)}\\&=-\sqrt{1-\lambda}a_{k+1}+\sqrt{\lambda}\left[-\sqrt{1-\lambda^k}h_{\lambda,k}+\lambda^{k/2} b\right]\\&= -\sqrt{1-\lambda^{k+1}}h_{\lambda,k+1}+\lambda^{(k+1)/2}  b \,.
	\ee
Finally, by using $\left[a_i,\left(a_j\right)^\dagger\right]=\delta_{ij}\mathbb{1}$, one concludes that $\left[h_{\lambda,k},\left(h_{\lambda,k}\right)^\dagger\right]=\mathbb{1}$.
\end{proof}
Lemma~\ref{bHeisenberg} helps in the calculation of mean values on the output environment state $\Delta_{\lambda,\sigma_0}^{(k)}(\rho^{(k)})$. For example, the mean photon number $\langle b^\dagger b\rangle_{\Delta_{\lambda,\sigma_0}^{(k)}(\rho^{(k)})}$ can be evaluated as 
\bb
    \langle b^\dagger b\rangle_{\Delta_{\lambda,\sigma_0}^{(k)}(\rho^{(k)})}&=\Tr\left[b^\dagger b\,\Delta_{\lambda,\sigma_0}^{(k)}(\rho^{(k)})\right] \\&=\Tr_{S_1S_2\ldots S_k E}  \left[b^\dagger b\,U_{\lambda}^{(S_kE)}\ldots U_{\lambda}^{(S_1E)} \right.  \\
  &\left. \qquad\rho^{(k)}\otimes\sigma_0\left(U_{\lambda}^{(S_1E)}\right)^\dagger\ldots     \left(U_{\lambda}^{(S_kE)}\right)^\dagger\right] \\&=\Tr_{S_1S_2\ldots S_k E}\left[\left(b_H^{(k)}\right)^\dagger b_H^{(k)} \rho^{(k)}\otimes \sigma_0\right] \\&=\left\langle \left(b_H^{(k)}\right)^\dagger b_H^{(k)}\right\rangle_{\rho^{(k)}\otimes\sigma_0}\,.
\ee
In addition, let us recall that the photon number variance of a state $\sigma$ is defined by
\begin{equation}
    V_\sigma\coloneqq\left\langle \left(b^\dagger b-\langle b^\dagger b\rangle_\sigma \right)^2\right\rangle_{\sigma}\,,
\end{equation}
and it holds that $V_\sigma\ge0$, $V_\sigma=\langle \left(b^\dagger b\right)^2 \rangle_\sigma -\langle b^\dagger b\rangle_\sigma^2$.
We are now ready to show Theorem~\ref{NoTheorem}. 
\begin{thm}\label{NoTheorem}
Let $n\in\N$, $\nu\ge0$, and $\lambda\in(0,1)$. The environmental Fock state $\ketbra{n}$ is not $(\lambda,\tau_\nu)$-achievable, except the trivial case in which $n=\nu=0$.
\end{thm}
\begin{proof}
Suppose that $n\ne0$ or $\nu\ne0$. Assume by contradiction that there exist $k\in\mathbb{N}$ and $\rho^{(k)}\in\mathfrak{S}(\HH_{S_{1}}\otimes\ldots\otimes\HH_{S_{k}})$ such that $\ketbra{n}=\Delta_{\lambda,\tau_\nu}^{(k)}(\rho^{(k)})$.
As a consequence, by applying Lemma~\ref{bHeisenberg}, the mean photon number can be evaluated as
\bb\label{energy_n}
	n&=\langle b^\dagger b\rangle_{\ketbra{n}}=\left\langle \left(b_H^{(k)}\right)^\dagger b_H^{(k)}\right\rangle_{\rho^{(k)}\otimes\tau_\nu} \\&=(1-\lambda^k)\left\langle \left(h_{\lambda,k}\right)^\dagger h_{\lambda,k}\right\rangle_{\rho^{(k)}}+\lambda^k\langle b^\dagger b\rangle_{\tau_\nu}\,,
	\ee
where we used that $\langle b\rangle_{\tau_\nu}=0$.
Moreover, by using Lemma~\ref{bHeisenberg}, one can show that
\bb
	&n^2=\langle (b^\dagger b)^2\rangle_{\ketbra{n}}=\left\langle \left(\left(b_H^{(k)}\right)^\dagger b_H^{(k)}\right)^2\right\rangle_{\rho^{(k)}\otimes\tau_\nu} \\&=\left\langle\left[(1-\lambda^k)\left(h_{\lambda,k}\right)^\dagger h_{\lambda,k}+\lambda^kb^\dagger b\right.\right.  \\
  &\left.\left. \quad-\lambda^{k/2}\sqrt{1-\lambda^k}\left(h_{\lambda,k}\right)^\dagger b-\lambda^{k/2}\sqrt{1-\lambda^k}h_{\lambda,k} b^\dagger\right]^2 \right\rangle_{\rho^{(k)}\otimes\tau_\nu} \\&=(1-\lambda^k)^2\left\langle \left(\left(h_{\lambda,k}\right)^\dagger h_{\lambda,k}\right)^2\right\rangle_{\rho^{(k)}}+\lambda^{2k}\langle (b^{\dagger}b)^2\rangle_{\tau_\nu} \\&\quad+\lambda^k(1-\lambda^k)\left[ \left(4\nu+1\right)\left\langle \left(h_{\lambda,k}\right)^\dagger h_{\lambda,k}\right\rangle_{\rho^{(k)}} +\nu\right]\,,
\ee
	where we used $\langle b^\dagger b\rangle_{\tau_\nu}=\nu$ and $\langle b b^\dagger b\rangle_{\tau_\nu}=0$. By setting 
	\begin{equation}\label{varianceRho}
	    V_{\rho^{(k)}}\coloneqq \left\langle \left(\left(h_{\lambda,k}\right)^\dagger h_{\lambda,k}\right)^2\right\rangle_{\rho^{(k)}}-\left\langle \left(h_{\lambda,k}\right)^\dagger h_{\lambda,k}\right\rangle_{\rho^{(k)}}^2\,,
	\end{equation}
	one obtains
	\bb
	0&=\langle (b^\dagger b)^2\rangle_{\ketbra{n}}-\langle b^\dagger b\rangle_{\ketbra{n}}^2 \\&=(1-\lambda^k)^2V_{\rho^{(k)}}+\lambda^{2k}V_{\tau_\nu} \\&\quad+\lambda^k(1-\lambda^k)\left[ \left(2\nu+1\right)\left\langle \left(h_{\lambda,k}\right)^\dagger h_{\lambda,k}\right\rangle_{\rho^{(k)}} +\nu\right]\,.
	\ee
	We deduce that $\nu= 0$ and, also, $\left\langle \left(h_{\lambda,k}\right)^\dagger h_{\lambda,k}\right\rangle_{\rho^{(k)}}=0$. By inserting the latter into~\eqref{energy_n}, we conclude that $n=0$. Hence, we have a contradiction.

	Now suppose that $\nu=n=0$. Fixed any $k\in\N$, the choice $$\rho^{(k)}=\ketbra{0}_{S_1}\otimes\ldots\otimes\ketbra{0}_{S_k}$$ 
	satisfies that
	\begin{equation}
	    \Delta_{\lambda,\tau_\nu}^{(k)}(\rho^{(k)})=\ketbra{0}_{E} \,.  
	\end{equation}
	This is a direct consequence of the fact that the {BS} preserves the vacuum i.e.~$$U_{\lambda}^{(S_iE)}\ket{0}_{S_i}\otimes\ket{0}_E=\ket{0}_{S_i}\otimes\ket{0}_E\,.$$
	This concludes the proof.
\end{proof}
As a consequence of the previous theorem, environmental Fock states can not be \emph{exactly} achieved at the output of the environment by feeding trigger pulses close to each other into an optical fibre. However, we will show that Fock states can be \emph{approximately} achieved. More precisely, by sending $k$ trigger signals initialised in an appropriate state, then the environment transforms into a state that is as close in trace norm to the chosen Fock state as desired if $k$ is large, as we will rigorously see later with Theorem~\ref{FockAchiev}.

Let us give a heuristic explanation of this fact. Since $\lambda\in(0,1)$ then Lemma~\ref{bHeisenberg} implies that for $k\rightarrow\infty$ it holds that $b_H^{(k)}\simeq-h_{\lambda,k}$. As a consequence, for $k\rightarrow\infty$ the mean photon number and the photon number variance of the output environment state are
\begin{equation}\label{MeanD}
    \left\langle b^\dagger b\right\rangle_{\Delta_{\lambda,\sigma_0}^{(k)}(\rho^{(k)})}\simeq   \left\langle \left(h_{\lambda,k}\right)^\dagger h_{\lambda,k}\right\rangle_{\rho^{(k)}}\,,
\end{equation}
\begin{equation}\label{VarD}
    V_{\Delta_{\lambda,\sigma_0}^{(k)}(\rho^{(k)})}\simeq   \left\langle \left(\left(h_{\lambda,k}\right)^\dagger h_{\lambda,k}\right)^2\right\rangle_{\rho^{(k)}}-\left\langle \left(h_{\lambda,k}\right)^\dagger h_{\lambda,k}\right\rangle_{\rho^{(k)}}^2\,,
\end{equation}
respectively. Fixed $k\gg 1$, our hope is to obtain $\Delta_{\lambda,\sigma_0}^{(k)}(\rho^{(k)})\simeq \ketbra{n}$ for a suitable choice of $\rho^{(k)}$. This is equivalent to requiring that $\Delta_{\lambda,\sigma_0}^{(k)}(\rho^{(k)})$ has mean photon number approximately equal to $n$ and vanishing photon number variance. By looking at~\eqref{MeanD} and~\eqref{VarD}, this is accomplished by choosing $\rho^{(k)}$ to be the $n$-th excited state of the harmonic oscillator defined by the annihilation operator $h_{\lambda,k}$ i.e.~
\begin{equation}
    \rho_{\lambda,n}^{(k)}\coloneqq \ketbra{n,\lambda}_{S_1\ldots S_k}\,,
\end{equation}
where
\begin{equation}\label{statenl}
 \ket{n,\lambda}_{S_1\ldots S_k}\coloneqq\frac{({h^\dagger_{\lambda,k}})^{n}}{\sqrt{n!}}\ket{0}_{S_1}\ket{0}_{S_2}\ldots\ket{0}_{S_k}\,.
\end{equation}
Hence, we expect that the output environment state $\Delta_{\lambda,\sigma_0}^{(k)}\left(\rho_{n,\lambda}^{(k)}\right)$ is as close to $\ketbra{n}$ as desired if $k$ is sufficiently large. A schematic of the interaction between the trigger signals and the environment in terms of {BS}s is shown in Fig.~\ref{inter}.
\begin{figure}[t]
\centering
\includegraphics[scale=0.40]{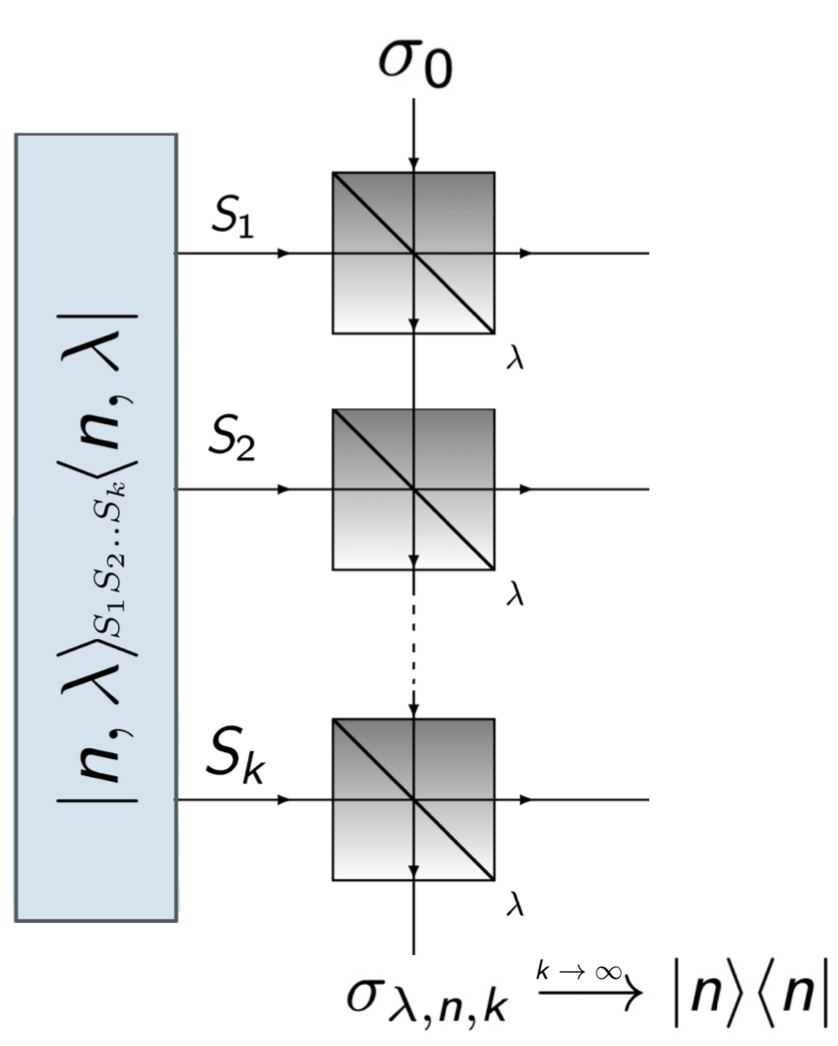}
\caption{A schematic of the interaction in terms of {BS}s. The initial environment state is $\sigma_0$. After the interaction with all the trigger signals $S_1$, $S_2$, .., $S_k$, the environment transforms into the state $\sigma_{\lambda,n,k}\coloneqq \Phi_{\lambda,\Delta_{\lambda,\sigma_0}^{(k)}\left(\rho_{n,\lambda}^{(k)}\right)} $. If $k$ is large, then $\sigma_{\lambda,n,k}$ is approximately the $n$-th Fock state.}
\label{inter}
\end{figure}

Intuitively, for $k$ large the fact that $$\Delta_{\lambda,\sigma_0}^{(k)}\left(\rho_{n,\lambda}^{(k)}\right)\approx \ketbra{n}$$ implies that $$\Phi_{\lambda,\Delta_{\lambda,\sigma_0}^{(k)}\left(\rho_{n,\lambda}^{(k)}\right)}\approx \Phi_{\lambda,\ketbrasub{n}}\,.$$ Therefore, we expect the capacities of $\Phi_{\lambda,\Delta_{\lambda,\sigma_0}^{(k)}\left(\rho_{n,\lambda}^{(k)}\right)}$ to be approximately equal to those of $\Phi_{\lambda,\ketbrasub{n}}$ for $k$ large. Lemma~\ref{Stability} assures that this is indeed the case.
\begin{lemma}\label{Stability}
	The energy-constrained quantum capacity and the energy-constrained entanglement-assisted classical capacity of the general attenuators are continuous with respect to the environment state, over the subset of environment states having a finite mean photon number.

	More explicitly, let $\sigma\in\mathfrak{S}(\HH_E)$ be an environment state such that $\langle b^\dagger b\rangle_\sigma<\infty$ and let $\{\sigma_k\}_{k\in\mathbb{N}}\subset\mathfrak{S}(\HH_E)$ be a sequence of environment states such that $$\lim\limits_{k\rightarrow\infty}\|\sigma_k-\sigma\|_1=0\,,$$
	then for all $N>0$ and $\lambda\in[0,1]$ it holds that:
	\bb
	\lim\limits_{k\rightarrow\infty}Q\left(\Phi_{\lambda,\sigma_k},N\right)&=Q\left(\Phi_{\lambda,\sigma},N\right)\,,\\
	\lim\limits_{k\rightarrow\infty}C_{\text{ea}}\left(\Phi_{\lambda,\sigma_k},N\right)&=C_{\text{ea}}\left(\Phi_{\lambda,\sigma},N\right)\,.
	\ee
\end{lemma}
\begin{proof}
Thanks to Lemma~\ref{continuityBound}, it suffices to show that for all $\lambda\in[0,1]$ it holds that: (a)~$\Phi_{\lambda,\sigma}$ is energy-limited, (b)~$\Phi_{\lambda,\sigma_k}$ is energy-limited for $k$ sufficiently large, and (c)~$\lim\limits_{k\rightarrow\infty}\|\Phi_{\lambda,\sigma_k}-\Phi_{\lambda,\sigma}\|_{\diamond N}=0$ for all $N>0$.
    
Let $\lambda\in[0,1]$ and $\rho\in\mathfrak{S}(\HH_S)$. By using~\eqref{energia_general}, we have that
\bb \label{enerbound}
&\Tr_S\left[\Phi_{\lambda,\sigma}(\rho)\,a^\dagger a\right] \\&\le\lambda \langle a^\dagger a \rangle_\rho +(1-\lambda)\langle b^\dagger b\rangle_\sigma+2\sqrt{\lambda(1-\lambda)}|\langle a\rangle_\rho| |\langle b^\dagger\rangle_\sigma|\,.
\ee
The Cauchy--Schwarz inequality for the Hilbert--Schmidt inner product yields 
\bb
&|\langle a\rangle_\rho|=|\langle a^\dagger\rangle_\rho|=|\Tr[\sqrt{\rho}\sqrt{\rho} a^\dagger]|\le\sqrt{\Tr\rho}\sqrt{\Tr[a\rho a^\dagger]} \\
&\qquad\,=\sqrt{\langle a^\dagger a \rangle_\rho}\,, \\
&|\langle b^\dagger\rangle_\sigma|\le  \sqrt{\langle b^\dagger b\rangle_\sigma}\,.
\ee
By inserting this into~\eqref{enerbound} and by noting that 
\bb
\sqrt{\langle a^\dagger a \rangle_\rho}\le \frac{\langle a^\dagger a \rangle_\rho}{4}+1\,,
\ee
one obtains that
\begin{align}
&\Tr_S[\Phi_{\lambda,\sigma}(\rho)\,a^\dagger a] \nonumber\\&\le \lambda \langle a^\dagger a \rangle_\rho +(1-\lambda)\langle b^\dagger b\rangle_\sigma+2\sqrt{\lambda(1-\lambda)\langle b^\dagger b\rangle_\sigma}\sqrt{\langle a^\dagger a \rangle_\rho}\nonumber\\&\le \left(\lambda+\frac{1}{2}\sqrt{\lambda(1-\lambda)\langle b^\dagger b\rangle_\sigma}\right)\langle a^\dagger a \rangle_\rho\nonumber \\&\quad+\left[(1-\lambda)\langle b^\dagger b\rangle_\sigma+2\sqrt{\lambda(1-\lambda)\langle b^\dagger b\rangle_\sigma}\right]\label{engenatt4}\\&
=\alpha \langle a^\dagger a \rangle_\rho+N_0\,.\label{engenatt}
\end{align}
As a result, since $\langle b^\dagger b\rangle_\sigma<\infty$, it follows that $\Phi_{\lambda,\sigma}$ is energy-limited.
	
Since $\sigma_k$ tends to $\sigma$, then $\langle b^\dagger b\rangle_{\sigma_k}$ has to be upper bounded for $k$ sufficiently large. Therefore,~\eqref{engenatt4} implies that $\Phi_{\lambda,\sigma_k}$ is energy-limited for $k$ sufficiently large.
	
Furthermore, let $N>0$. Take an arbitrary Hilbert space $\HH_C$ and $\rho\in\mathfrak{S}(\HH_S\otimes\HH_C)$ such that $\Tr\left[\rho\,a^\dagger a\right]\le N$. By applying~\cite[Lemma 23]{QCLT}, we have that
\bb\label{diamond_dist_gen}
&\left\|\left(\left(\Phi_{\lambda,\sigma_k}-\Phi_{\lambda,\sigma}\right)\otimes I_C\right)(\rho)\right\|_1 \le \left\|\sigma_k-\sigma\right\|_1\,.
\ee
\eqref{diamond_dist_gen} can be derived from the contractivity of the trace norm under partial traces and from the invariance of the trace norm under unitary transformations. Thanks to~\eqref{diamond_dist_gen}, it holds that
\begin{equation}\label{diamondtrace}
\left\|\Phi_{\lambda,\sigma_k}-\Phi_{\lambda,\sigma}\right\|_{\diamond N}\le\left\|\sigma_k-\sigma\right\|_1\,.
\end{equation}
Since $\lim\limits_{k\rightarrow\infty}\left\|\sigma_k-\sigma\right\|_1=0$, one concludes that $\lim\limits_{k\rightarrow\infty}\left\|\Phi_{\lambda,\sigma_k}-\Phi_{\lambda,\sigma}\right\|_{\diamond N}=0$.
\end{proof}

Before continuing our analysis, let us formalise the fact that: if a sequence of single-mode states $\{\rho_k\}_{k\in\mathbb{N}}$ is such that its mean photon number $\langle a^\dagger a \rangle_{\rho_k}$ tends to $n$ and its photon number variance $V_{\rho_k}$ tends to $0$, then $\{\rho_k\}_{k\in\mathbb{N}}$ converges to $\ketbra{n}$ in trace norm. This property is guaranteed by the following lemma. 
\begin{lemma}\label{lemmavar}
For all single-mode state $\rho$ and all $n\in\N$, it holds that
\begin{equation}\label{ineq_var}
\| \rho-\ketbra{n} \|_1\le2\sqrt{V_\rho+\left(\langle a^\dagger a\rangle_\rho -n\right)^2}\,,
\end{equation}
where $V_\rho\coloneqq\left\langle\left(a^\dagger a-\left\langle a^\dagger a\right\rangle_\rho\mathbb{1} \right)^2\right\rangle_\rho$ is the photon number variance of $\rho$.
\end{lemma}
\begin{proof}
Let us observe that
	$$\left(a^\dagger a-n\mathbb{1}\right)^2\ge \mathbb{1}-\ketbra{n} \,,$$
	since it can be rewritten as
	$$\sum_{k=0}^\infty\left[(k-n)^2-1+\delta_{kn}\right]\ketbra{k}\ge 0$$
	and $(k-n)^2-1+\delta_{kn}\ge0$ for all $n,k\in\mathbb{N}$. Hence, one obtains
	\bb\label{eq_fid}
	\Tr\left[\rho\left(a^\dagger a-n\mathbb{1}\right)^2\right]&\ge	\Tr\left[\rho\left(\mathbb{1}-\ketbra{n}\right)\right] \\&=1-F^2\left(\rho,\ketbra{n}\right)	 \,,
	\ee
	where 
	$F\left(\rho,\ketbra{n}\right)=\sqrt{\bra{n}\rho\ket{n}}$ is the fidelity between $\rho$ and the pure state $\ketbra{n}$.~\eqref{fidel_trace_ineq} establishes that $$\frac{1}{2}\left\|\rho-\ketbra{n}\right\|_1\le\sqrt{1-F^2(\rho,\ketbra{n})}\,.$$ Consequently,~\eqref{eq_fid} implies
	 \begin{equation}
	\| \rho-\ketbra{n} \|_1\le2\sqrt{\Tr\left[\rho\left(a^\dagger a-n\mathbb{1}\right)^2\right]} \,.
	 \end{equation}
	Finally, the identity
	\bb
	&\Tr\left[\rho\left(a^\dagger a-n\mathbb{1}\right)^2\right] \\&=\Tr\left[\rho\left(a^\dagger a-\langle a^\dagger a \rangle_\rho\mathbb{1}\right)^2\right]+\left(\langle a^\dagger a\rangle_\rho-n\right)^2 \,
	\ee
	concludes the proof.
\end{proof}
\begin{ex}
Let $n\in\N$, $\varepsilon\in(0,1)$, and $\rho_{\varepsilon,n}\coloneqq\ketbra{\psi_{\varepsilon,n}}$ be a single-mode state where
\bb
\ket{\psi_{\varepsilon,n}}\coloneqq \sqrt{1-\varepsilon}\ket{n}+\sqrt{\varepsilon}\ket{n+1}\,.
\ee
By using the explicit formula for the trace distance between pure states~\cite[Eq.~10.17]{HOLEVO-CHANNELS-2}, the left-hand side of~\eqref{ineq_var} is given by
\bb
\left\| \rho_{\varepsilon,n}-\ketbra{n} \right\|_1=2\sqrt{1-|\langle\psi_{\varepsilon,n}|n\rangle|^2}=2\sqrt{\varepsilon}\,.
\ee
Moreover, one can verify that the right-hand side of~\eqref{ineq_var} is
\bb
2\sqrt{V_{\rho_{\varepsilon,n}}+\left(\langle a^\dagger a\rangle_{\rho_{\varepsilon,n}} -n\right)^2}=2\sqrt{\varepsilon}\,.
\ee
Therefore, Lemma~\ref{lemmavar} is satisfied with equality. Hence, the inequality in~\eqref{ineq_var} is generally tight.
\end{ex}
Now, we are ready to state 
Theorem~\ref{FockAchiev}.
\begin{thm}\label{FockAchiev}
Let $\lambda\in(0,1)$, $n\in\N$, and let $\sigma_0$ be the stationary environment state. Suppose $\langle (b^\dagger b)^2\rangle_{\sigma_0}<\infty$. Suppose further that, during step~2 of the noise attenuation protocol, Alice sends $k$ trigger signals initialised in 
\begin{equation}\label{trigger_signals_k}
    \rho_{\lambda,n}^{(k)}\coloneqq \ketbra{n,\lambda}_{S_1\ldots S_k}\,.
\end{equation}
For $k\ge2$, the state $\ket{n,\lambda}_{S_1\ldots S_k}$, defined in~\eqref{statenl}, can be written as
\begin{equation}\label{recipe}
\ket{n,\lambda}_{S_1S_2\ldots S_k}=V^{(S_1S_2\ldots S_k)}_{k,\lambda}\ket{0}_{S_1}\ldots \ket{0}_{S_{k-1}}\ket{n}_{S_k}\,,
\end{equation}
where
\begin{equation} \label{V_unitary}
V^{(S_1S_2\ldots S_k)}_{k,\lambda}\coloneqq U^{(S_{1}S_2)}_{\frac{1-\lambda}{1-\lambda^2}}\ldots U^{(S_{k-1}S_k)}_{\frac{1-\lambda}{1-\lambda^k}}\,.
\end{equation}
Otherwise, for $k=1$ it holds that
\begin{equation}\label{case_k1}
    \ket{n,\lambda}_{S_1}=\ket{n}_{S_1}\,.
\end{equation}
Then, at the beginning of step~3, the environment transforms into the state
\begin{equation}
    \sigma_{\lambda,n,k}\coloneqq \Delta_{\lambda,\sigma_0}^{(k)}\left(\rho_{n,\lambda}^{(k)}\right)\,,
\end{equation}
(the dependence of $\sigma_{\lambda,n,k}$ on $\sigma_0$ is implicit) which satisfies
\begin{equation}\label{distance_from_Fock}
    \|\sigma_{\lambda,n,k}-\ketbra{n}\|_1 \le \eta\, \lambda^{k/2}\,,
\end{equation}
where $\eta\in\R^+ $ is a constant with respect to $\lambda$ and $k$ (but it is related to $n$, $\langle b^\dagger b\rangle_{\sigma_0}$, $\langle (b^\dagger b)^2\rangle_{\sigma_0}$). Hence
\begin{equation}\label{limit_to_fock}
	\lim\limits_{k\rightarrow\infty}\|\sigma_{\lambda,n,k}-\ketbra{n}\|_1=0\,.
\end{equation} 
Moreover, for all $\lambda\in (0,1/2)$ and for $k$ sufficiently large it holds that
\begin{equation}\label{achiev_quantum}
    Q\left(\Phi_{\lambda,\sigma_{\lambda,n_\lambda,k}}\right)\ge     Q\left(\Phi_{\lambda,\sigma_{\lambda,n_\lambda,k}},1/2\right)>0\,,
\end{equation}
where $n_\lambda\in\mathbb{N}$ with $$\frac{1}{\lambda}-1\le n_\lambda\le \frac{1}{\lambda} \,.$$ In addition, if Conjecture~\ref{congI} is valid, then for all $N>0$, $\lambda\in(0,1/2)$, and for $\bar{n}\in\N$ sufficiently large, it holds that
\bb\label{limit_k_Cea}
    \lim\limits_{k\rightarrow\infty}C_{\text{ea}}\left(\Phi_{\sigma_{\lambda,\bar{n},k}},N\right)&>C\left(\Id,N\right)\,,\\
    \lim\limits_{k\rightarrow\infty}Q_{\text{ea}}\left(\Phi_{\lambda,\sigma_{\lambda,\bar{n},k}},N\right)&>Q\left(\Id,N\right)/2\,,\\
    \lim\limits_{k\rightarrow\infty}Q\left(\Phi_{\lambda,\sigma_{\lambda,\bar{n},k}},N\right)&>0\,.
\ee
\end{thm}
Before proving Theorem~\ref{FockAchiev}, let us explain its meaning. First, notice that Theorem~\ref{FockAchiev} can be applied when the stationary environment state is a thermal state $\tau_\nu$, since
\bb
\langle(b^\dagger b)^2\rangle_{\tau_\nu}=\nu(2\nu+1)<\infty\,.
\ee
Suppose that Alice and Bob share an optical fibre {with} 
very low transmissivity i.e.~$0<\lambda\ll 1/2$. The stationary environment state associated with the quantum channel description of the optical fibre is a thermal state i.e.~$\sigma_0=\tau_\nu$ with $\nu\ge0$. Since for $\lambda\ll 1/2$ the thermal attenuator $\Phi_{\lambda,\tau_\nu}$ has zero quantum capacity and vanishing entanglement-assisted capacities, Alice can not reliably transmit qubits to Bob via the channel $\Phi_{\lambda,\tau_\nu}$, and they can not efficiently communicate even if entanglement is pre-shared between them.

However, Theorem~\ref{diehard_th} guarantees that there exists a suitable environmental Fock state $\ketbra{n}$ that makes the quantum capacity of the corresponding channel $\Phi_{\lambda,\ketbrasub{n}}$ strictly positive. Moreover, 
if Conjecture~\ref{congI} is valid, then Theorem~\ref{congCap} guarantees that the channel $\Phi_{\lambda,\ketbrasub{n}}$ allows for very efficient communication performance for $n$ sufficiently large. Therefore, a reasonable strategy is to apply the noise attenuation protocol in order to communicate via a channel that is approximately equal to $\Phi_{\lambda,\ketbrasub{n}}$, for a choice of $n$ that activates the discussed properties 
of Theorem~\ref{diehard_th} or Theorem~\ref{congCap}.

Now, we apply Theorem~\ref{FockAchiev}. If Alice sends $k$ trigger signals separated by a sufficiently short temporal interval and initialised in $\rho_{\lambda,n}^{(k)}$ (defined in~\eqref{trigger_signals_k}), then the environment transforms into a state $\sigma_{\lambda,n,k}$ which satisfies $$\sigma_{\lambda,n,k}\sim \ketbra{n}$$ for $k$ large. More precisely, the trace distance between $\sigma_{\lambda,n,k}$ and $\ketbra{n} $ goes exponentially to zero for $k\rightarrow\infty$. In other words, Alice is able to achieve the desired environmental Fock state if she sends a sufficiently large number of trigger signals.

The information-carrier signals, sent at the Steps 3, are affected by the quantum channel $\Phi_{\lambda,\sigma_{\lambda,n,k}}$. If $n$ is chosen to be such that $1/\lambda-1\le n\le 1/\lambda$, then $\Phi_{\lambda,\sigma_{\lambda,n,k}}$ has strictly positive quantum capacity for $k$ sufficiently large (i.e.\ the phenomenon of {D-HQCOM} is activated). 
Moreover, under the assumption that Conjecture~\ref{congI} is valid, for $n$ and $k$ sufficiently large the channel $\Phi_{\lambda,\sigma_{\lambda,n,k}}$ allows for an efficient entanglement-assisted communication and satisfies the same interesting inequalities as the environmental $n$-th Fock state of Theorem~\ref{congCap}:   
\bb
C_{\text{ea}}\left(\Phi_{\sigma_{\lambda,n,k}},N\right)&>C\left(\Id,N\right)\,,\\
Q_{\text{ea}}\left(\Phi_{\lambda,\sigma_{\lambda,n,k}},N\right)&>Q\left(\Id,N\right)/2\,,\\
Q\left(\Phi_{\lambda,\sigma_{\lambda,n,k}},N\right)&>0\,.
\ee

Notice that Theorem~\ref{FockAchiev} provides an explicit procedure to construct the state $\ket{n,\lambda}_{S_1S_2\ldots S_k}$ of the $k$ trigger signals. Indeed,~\eqref{recipe} implies that it suffices to input $k-1$ vacuum states and the $n$-th Fock state into the interferometer depicted in Fig.~\ref{InterfK}.

Moreover, since the trigger $\ket{n,\lambda}_{S_1S_2\ldots S_k}$ state does not depend on $\sigma_0$, Alice is able to put (asymptotically in $k$) the environment in a Fock state even if she does not know the initial environment state $\sigma_0$.
\begin{figure}[t]
\centering
\includegraphics[scale=0.20]{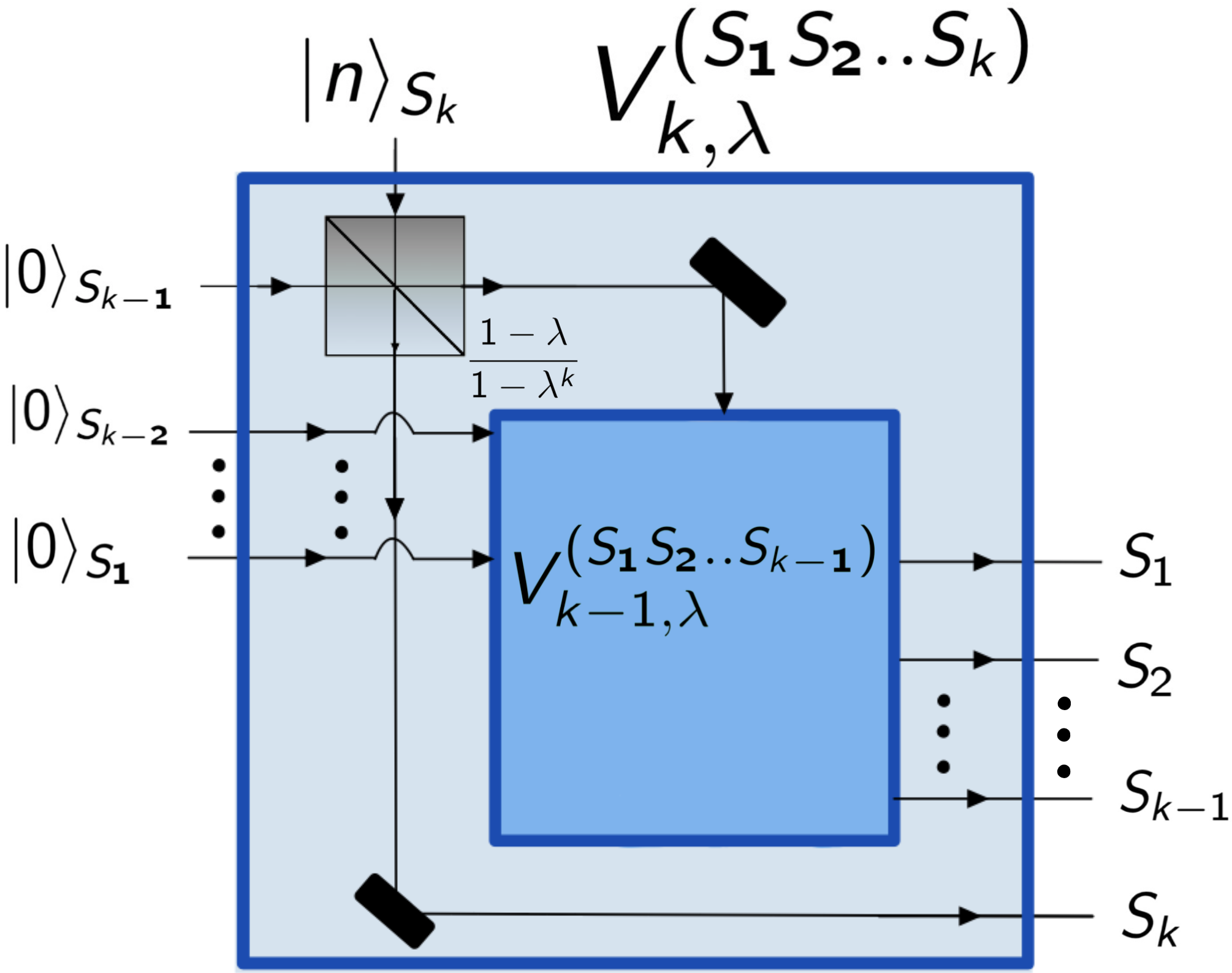}
\caption{A recursive schematic of the interferometer which implements the unitary operator $V^{(S_1S_2\ldots S_k)}_{k,\lambda}$ defined in~\eqref{V_unitary}. The base case $k=2$ is depicted in Fig.~\ref{Interf2}. If the $n$-th Fock state and $k-1$ vacuum states are input into this interferometer, the $k$ output signals are set in $\ket{n,\lambda}_{S_1S_2\ldots S_k}$. These $k$ signals can be sent through the channel to alter the environment into a state which is as close to the environmental $n$-th Fock state as desired if $k$ is large, as guaranteed by Theorem~\ref{FockAchiev}.}
\label{InterfK}
\end{figure}
\begin{figure}[t]
\centering
\includegraphics[scale=0.20]{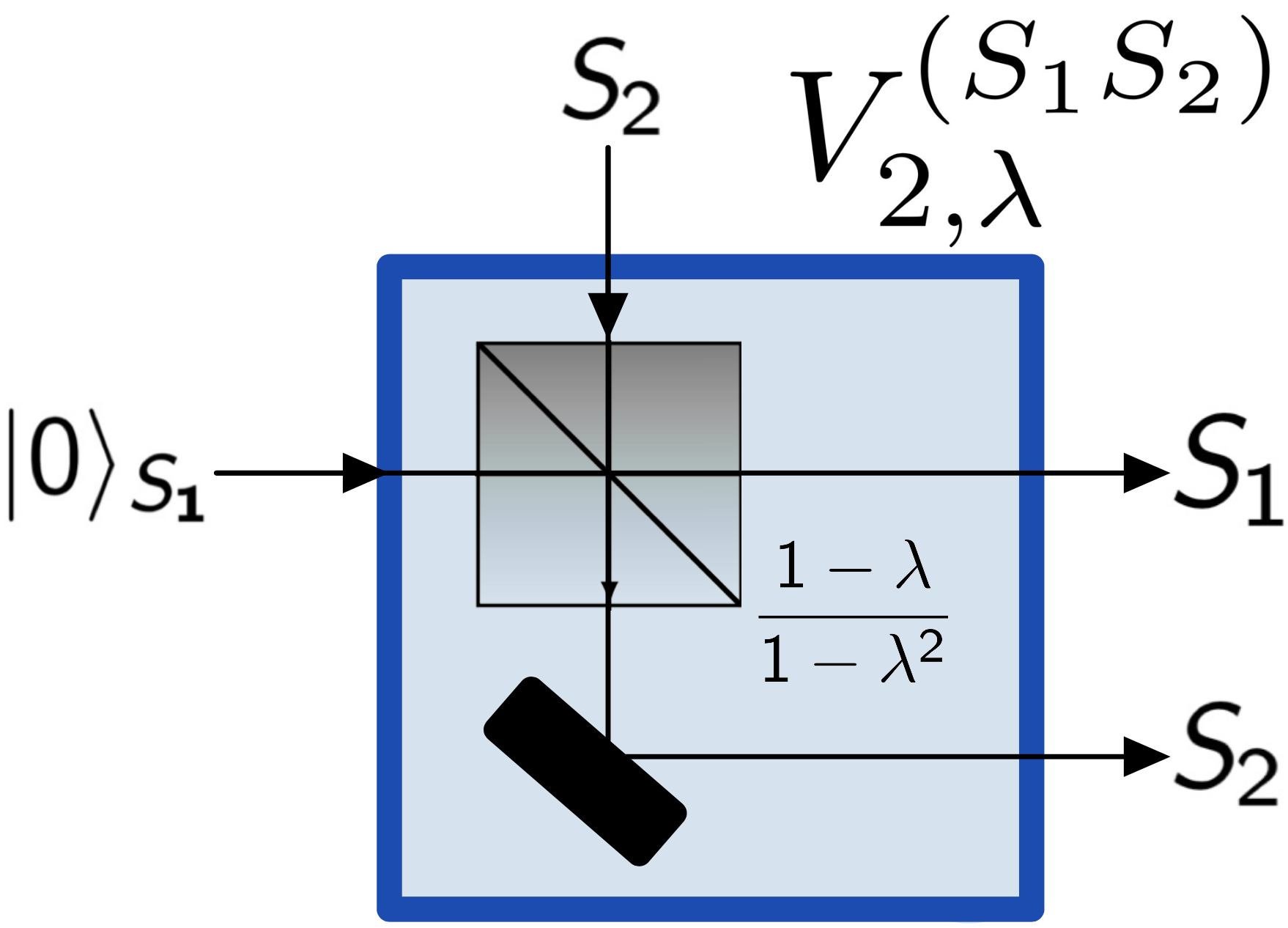}
\caption{A schematic of the interferometer which implements the unitary operator $V^{(S_1S_2)}_{2,\lambda}$ defined in~\eqref{V_unitary}. If the $n_\lambda$-th Fock state with $1/\lambda-1\le n_\lambda\le 1/\lambda$ and a vacuum state are input into this interferometer, the two output signals are set in $\ket{n_\lambda,\lambda}_{S_1S_2}$. As guaranteed by Theorem~\ref{theorem_lambda_small}, if Alice sends these two output signals into the channel, the environment transforms into a state $\sigma_\lambda$ such that for $\lambda>0$ sufficiently small the corresponding general attenuator $\Phi_{\lambda,\sigma_\lambda}$ has quantum capacity bounded away from $0$.}
\label{Interf2}
\end{figure}

We are now ready to prove Theorem~\ref{FockAchiev}.
\begin{proof}[Proof of Theorem~\ref{FockAchiev}]
Notice that~\eqref{case_k1} is a consequence of~\eqref{statenl} and of the fact that $h_{\lambda,1}=a_1$.
Now, let us show that the state $\ket{n,\lambda}_{S_1\ldots S_k}$ can be rewritten as in~\eqref{recipe} for $k\ge2$. For induction it can be shown that
\begin{equation}\label{hvav}
h_{\lambda,k}=V^{(S_1S_2\ldots S_k)}_{k,\lambda}a_k\left({V^{(S_1S_2\ldots S_k)}_{k,\lambda}}\right)^\dagger\,.
\end{equation}
It holds that
\begin{equation}\label{trasf_a2}
    U_{\lambda}^{(S_{1}S_{2})}a_{2}\left({U_{\lambda}^{(S_{1}S_{2})}}\right)^\dagger=\sqrt{1-\lambda}a_{1}+\sqrt{\lambda}a_2\,,
\end{equation}
as guaranteed by~\eqref{trasfheisb2} in {the} Appendix.
Consequently, one obtains that
\bb
&V^{(S_1S_2)}_{2,\lambda}a_2\left({V^{(S_1S_2)}_{2,\lambda}}\right)^\dagger=U^{(S_{1}S_2)}_{\frac{1-\lambda}{1-\lambda^2}}a_2\left({U^{(S_{1}S_2)}_{\frac{1-\lambda}{1-\lambda^2}}}\right)^\dagger  \\&=\sqrt{\frac{1-\lambda}{1-\lambda^2}}\left(\sqrt{\lambda}a_1+a_2\right)=h_{\lambda,2}\text{ .}
\ee
Hence,~\eqref{hvav} holds for $k=2$. Assuming~\eqref{hvav} to be true for $k$, we have to prove it for $k+1$. By using~\eqref{trasf_a2} again we obtain
\bb
&V^{(S_1S_2\ldots S_kS_{k+1})}_{k+1,\lambda}a_{k+1}\left({V^{(S_1S_2\ldots S_kS_{k+1})}_{k+1,\lambda}}\right)^\dagger\\ &=V^{(S_1S_2\ldots S_k)}_{k,\lambda}U^{(S_{k}S_{k+1})}_{\frac{1-\lambda}{1-\lambda^{k+1}}}a_{k+1}\left({U^{(S_{k}S_{k+1})}_{\frac{1-\lambda}{1-\lambda^{k+1}}}}\right)^\dagger\left({V^{(S_1S_2\ldots S_k)}_{k,\lambda}}\right)^\dagger \\&=V^{(S_1S_2\ldots S_k)}_{k,\lambda}\left(\sqrt{1-\frac{1-\lambda}{1-\lambda^{k+1}}}a_{k}\right.  \\
  &\left.\quad+\sqrt{\frac{1-\lambda}{1-\lambda^{k+1}}}a_{k+1}\right)\left({V^{(S_1S_2\ldots S_k)}_{k,\lambda}}\right)^\dagger \\&=\sqrt{\frac{1-\lambda}{1-\lambda^{k+1}}}\left(\sqrt{\lambda}\sqrt{\frac{1-\lambda^k}{1-\lambda}}h_{\lambda,k}+a_{k+1}\right) \\&= h_{\lambda,k+1}\text{ .}
\ee
Hence, by using that 
$$\left({V^{(S_1S_2\ldots S_k)}_{k,\lambda}}\right)^\dagger\ket{0}_{S_1}\ket{0}_{S_2}\ldots\ket{0}_{S_k}=\ket{0}_{S_1}\ket{0}_{S_2}\ldots\ket{0}_{S_k}\text{ ,}$$
\eqref{hvav} implies
\bb
&\ket{n,\lambda}_{S_1S_2\ldots S_k}=\frac{{\left(h^\dagger_{\lambda,k}\right)}^n}{\sqrt{n!}}\ket{0}_{S_1}\ket{0}_{S_2}\ldots\ket{0}_{S_k}\\&=V^{(S_1S_2\ldots S_k)}_{k,\lambda}\frac{{\left(a_k^\dagger\right)}^n}{\sqrt{n!}}{V^{(S_1S_2\ldots S_k)}_{k,\lambda}}^\dagger\ket{0}_{S_1}\ket{0}_{S_2}\ldots\ket{0}_{S_k}\\&=V^{(S_1S_2\ldots S_k)}_{k,\lambda}\ket{0}_{S_1}\ket{0}_{S_2}\ldots\ket{0}_{S_{k-1}}\ket{n}_{S_k}\,.
\ee
Hence,~\eqref{recipe} is proved. Now, let us show~\eqref{distance_from_Fock}.
Since $\rho_{\lambda,n}^{(k)}= \ketbra{n,\lambda}_{S_1\ldots S_k}$ and since $\ket{n,\lambda}_{S_1\ldots S_k}$ is the $n$-th excited state of the quantum harmonic oscillator defined by the annihilation operator $h_{\lambda,k}$, then the following identities hold:
\bb
    \left\langle \left(h_{\lambda,k}\right)^\dagger h_{\lambda,k}\right\rangle_{\rho^{(k)}_{\lambda,n}}&=n\,,\\
    \left\langle \left(\left(h_{\lambda,k}\right)^\dagger h_{\lambda,k}\right)^2\right\rangle_{\rho^{(k)}_{\lambda,n}}&=n^2\,,\\
    \left\langle  h_{\lambda,k}\right\rangle_{\rho^{(k)}_{\lambda,n}}&=0\,,\\
    \left\langle  {h_{\lambda,k}}^2\left(h_{\lambda,k}\right)^\dagger\right\rangle_{\rho^{(k)}_{\lambda,n}}&=0\,.
\ee
By using these identities and Lemma~\ref{bHeisenberg}, one can obtain that
\bb
\left\langle b^\dagger b\right\rangle_{\sigma_{\lambda,n,k}}&=\left\langle \left(b_H^{(k)}\right)^\dagger b_H^{(k)}\right\rangle_{\rho^{(k)}_{\lambda,n}\otimes\sigma_0}
 \\&=(1-\lambda^k)n+\lambda^k\langle b^\dagger b\rangle_{\sigma_0}\text{}\label{bb}
\ee
and
\bb
&\left\langle \left(b^\dagger b\right)^2\right\rangle_{\sigma_{\lambda,n,k}}=\left\langle \left(\left(b_H^{(k)}\right)^\dagger b_H^{(k)}\right)^2\right\rangle_{\rho^{(k)}_{\lambda,n}\otimes\sigma_0}
 \\&=(1-\lambda^k)^2n^2+4n\lambda^k(1-\lambda^k)\langle b^\dagger b\rangle_{\sigma_0}+\lambda^{2k}\left\langle (b^\dagger b)^2\right\rangle_{\sigma_0}+ \\&\quad+\lambda^k(1-\lambda^k)n+\lambda^k(1-\lambda^k)\langle b^\dagger b\rangle_{\sigma_0}\text{ .}\label{bb2}
\ee
Hence, from~\eqref{bb} and~\eqref{bb2}, the photon number variance of $\sigma_{\lambda,n,k}$ can be evaluated as
\bb
V_{\sigma_{\lambda,n,k}}&= 2n\lambda^k(1-\lambda^k)\langle b^\dagger b\rangle_{\sigma_0} \\&\quad+\lambda^{2k}\left(\left\langle (b^\dagger b)^2\right\rangle_{\sigma_0}-\langle b^\dagger b\rangle_{\sigma_0}^2\right) \\&\quad+\lambda^k(1-\lambda^k)n+\lambda^k(1-\lambda^k)\langle b^\dagger b\rangle_{\sigma_0}\,.\label{variances22}
\ee
By applying Lemma~\ref{lemmavar}, one obtains
\begin{align}
&\frac{1}{4}\left\|\sigma_{\lambda,n,k}-\ketbra{n} \right\|_1^2\le V_{\sigma_{\lambda,n,k}}+\left(\langle b^\dagger b\rangle_{\sigma_{\lambda,n,k}} -n\right)^2\label{inreqcop}\\&= \left[n^2-(4n+1)\langle b^\dagger b\rangle_{\sigma_0}+\langle (b^\dagger b)^2\rangle_{\sigma_0}-n\right]\lambda^{2k} \nonumber\\&\quad+\left[(2n+1)\langle b^\dagger b\rangle_{\sigma_0}+n\right]\lambda^{k}\label{distance_from_Fock2}\\&\le  \left[n^2+\left\langle (b^\dagger b)^2\right\rangle_{\sigma_0}+(2n+1)\langle b^\dagger b\rangle_{\sigma_0}+n\right]\lambda^{k}\text{ .}\label{concpr2}
\end{align}
Hence, by setting
\begin{equation}
\eta\coloneqq 2\sqrt{n^2+\langle (b^\dagger b)^2\rangle_{\sigma_0}+(2n+1)\langle b^\dagger b\rangle_{\sigma_0}+n}\,,
\end{equation}
it holds that
\begin{equation}\label{distrasp}
\left\|\sigma_{\lambda,n,k}-\ketbra{n} \right\|_1\le \eta\lambda^{k/2}\,.
\end{equation}
Since by hypothesis $\langle (b^\dagger b)^2\rangle_{\sigma_0}< \infty$ (notice that this implies $\langle b^\dagger b\rangle_{\sigma_0}< \infty$), then $\eta$ is finite and, as a consequence,~\eqref{distance_from_Fock} is proved. In the limit $k\rightarrow\infty$ the right hand side of~\eqref{distance_from_Fock} vanishes and hence~\eqref{limit_to_fock} holds.
Let $\lambda\in(0,1/2)$ and $n_\lambda\in \N$ with $1/\lambda-1\le n_\lambda\le 1/\lambda$. Lemma~\ref{Stability} and Theorem~\ref{diehard_th} guarantee that
\begin{equation}
    \lim\limits_{k\rightarrow\infty}Q\left(\Phi_{\lambda,\sigma_{\lambda,n_\lambda,k}},1/2\right)=Q\left(\Phi_{\lambda,\ketbrasub{n_\lambda} },1/2\right)>0\,.
\end{equation}
Since $Q\left(\Phi\right)\ge Q\left(\Phi,1/2\right)$ for any quantum channel $\Phi$, we deduce that $Q\left(\Phi_{\lambda,\sigma_{\lambda,n_\lambda,k}}\right)>0$ for $k$ sufficiently large.

In addition, under the assumption that Conjecture~\ref{congI} is valid, we can apply Theorem~\ref{congCap}.
The latter, together with~\eqref{limit_to_fock} and Lemma~\ref{Stability}, would imply the validity of the limits in~\eqref{limit_k_Cea}.
\end{proof}

By adopting the noise attenuation protocol with $k$ trigger signals, the best rate of qubits (i.e.~number of qubits reliably transferred per number of transmission line uses) is equal to the quantum capacity of the resulting channel (that obtained at step~3) divided by $k+1$. Hence, this rate is lower and lower as $k$ increases, and~\eqref{achiev_quantum} of Theorem~\ref{FockAchiev} may require a large $k$ to be valid.
Fortunately, Theorem~\ref{theorem_lambda_small} guarantees that sending just $k=2$ trigger signals is enough for $\lambda>0$ sufficiently small (which is the physically interesting case) to obtain a resulting channel with strictly positive quantum capacity. This is reasonable, since~\eqref{distance_from_Fock} implies that the lower $\lambda$ the faster the convergence in $k$ of $\{\sigma_{\lambda,n,k}\}_{k\in\mathbb{N}}$ to $\ketbra{n}$.  
\begin{thm}\label{theorem_lambda_small}
Let $\sigma_0$ be the stationary environment state and $\lambda\in(0,1]$ the transmissivity of the channel. Suppose $\langle (b^\dagger b)^2\rangle_{\sigma_0}<\infty$. Suppose further that Alice sends two trigger signals initialised in 
\begin{equation}
    \ket{n_\lambda,\lambda}_{S_1S_2}=U_{\frac{1}{1+\lambda}}^{(S_1S_2)}\ket{0}_{S_1}\ket{n_\lambda}_{S_2}\,,
\end{equation}
where $n_\lambda\in\mathbb{N}$ such that $1/\lambda-1\le n_\lambda\le 1/\lambda$. Then the environment transforms into a state $\sigma_\lambda$ such that
\begin{equation}\label{Fock_lambda}
	\left\|\sigma_{\lambda}-\ketbra{n_\lambda}\right\|_1\le \eta_0\,\lambda^{1/2}\,,
\end{equation}
where $\eta_0>0$ is a constant with respect to $\lambda$ (but it is related to $\langle b^\dagger b\rangle_{\sigma_0}$ and $\langle (b^\dagger b)^2\rangle_{\sigma_0}$). Moreover, for $\lambda>0$ sufficiently small it holds that 
\begin{equation}\label{die_hard_achiev}
    Q\left(\Phi_{\lambda,\sigma_{\lambda}}\right)\ge Q\left(\Phi_{\lambda,\sigma_{\lambda}},1/2\right)\ge c\,,
\end{equation}
where $c>0$ is a universal constant.
\end{thm}
\begin{proof}
Notice that we can apply the formulae obtained in the proof of Theorem~\ref{FockAchiev} with $\sigma_{\lambda,n_\lambda,2}= \sigma_\lambda$. In particular,~\eqref{distance_from_Fock2} implies that
\bb
\frac{1}{4}\left\|\sigma_{\lambda}-\ketbra{n_\lambda}\right\|_1^2 &\le \lambda^4 n_\lambda^2+\lambda^4 \langle (b^\dagger b)^2\rangle_{\sigma_0}\\
&\quad +\lambda^2(2n_\lambda+1)\langle b^\dagger b\rangle_{\sigma_0}+\lambda^2n_\lambda\,.
\ee
By using $n_\lambda\le 1/\lambda$ and $\lambda\le1$, one obtains
\begin{equation}\label{tracenorm2}
\frac{1}{2}\left\|\sigma_\lambda-\ketbra{n_\lambda}\right\|_1\le k_0 \sqrt{\lambda}\,,
\end{equation}
where
\begin{equation}
k_0\coloneqq \left[2+3\langle b^\dagger b\rangle_{\sigma_0}+\langle(b^\dagger b)^2\rangle_{\sigma_0}\right]^{1/2}\,,
\end{equation}
which is finite since $\langle (b^\dagger b)^2\rangle_{\sigma_0}< \infty$. This concludes the proof of~\eqref{Fock_lambda}.

In order to prove~\eqref{die_hard_achiev}, we want to exploit the continuity bound with respect to the energy-constrained diamond norm of Lemma~\ref{continuityBound}. 
First, we need to find $\alpha(\lambda)$ and $N_0(\lambda)$ such that for any state $\rho$ it holds that
\bb
\Tr\left[a^\dagger a\, \Phi_{\lambda,\sigma_\lambda}(\rho)\right]&\le \alpha(\lambda) \Tr\left[a^\dagger a\,\rho\right]+N_0(\lambda)\,,\\
\Tr\left[a^\dagger a\, \Phi_{\lambda,\ketbrasub{n_\lambda}}(\rho)\right]&\le \alpha(\lambda) \Tr\left[a^\dagger a\,\rho\right]+N_0(\lambda)\,.
\ee
Notice that 
\begin{equation}
\langle b^\dagger b\rangle_{\ketbrasub{n_\lambda}}=n_\lambda \le 1/\lambda\,,
\end{equation}
and that for $\lambda>0$ sufficiently small it holds that
\bb
\langle b^\dagger b \rangle_{\sigma_\lambda}&=(1-\lambda^2)n_\lambda+\lambda^2\langle b^\dagger b \rangle_{\sigma_0} \\&=  n_\lambda-\lambda^2\left(n_\lambda-\langle b^\dagger b \rangle_{\sigma_0}\right)\le n_\lambda\le 1/\lambda\,,
\ee
where we used~\eqref{bb}. As a consequence,~\eqref{engenatt} guarantees that for $\lambda>0$ sufficiently small we can take 
\bb
    \alpha(\lambda)&\equiv 1\,,\\
    N_0(\lambda)&\equiv \frac{1}{\lambda}+2\,.
\ee
Second, we need to find an upper bound on the energy-constrained diamond distance between $\Phi_{\lambda,\ketbrasub{n_\lambda}}$ and $\Phi_{\lambda,\sigma_\lambda}$. By using~\eqref{diamondtrace} and~\eqref{tracenorm2}, we obtain
\begin{equation}
\frac{1}{2}\left\|\Phi_{\lambda,\sigma_{\lambda}}-\Phi_{\lambda,\ketbrasub{n_\lambda}}\right\|_{\diamond N}\le k_0\sqrt{\lambda}\,,
\end{equation}
for all $N>0$.
By setting $\varepsilon(\lambda)\coloneqq k_0\sqrt{\lambda}$, we apply Lemma~\ref{continuityBound} to obtain
\bb
&\left|Q\left(\Phi_{\lambda,\sigma_{\lambda}},N\right)-Q\left(\Phi_{\lambda,\ketbrasub{n_\lambda}},N\right)\right| \\&\le 56\sqrt{\varepsilon(\lambda)}\, g\left(4\frac{\alpha(\lambda) N + N_0(\lambda)}{\sqrt{\varepsilon(\lambda)}}\right)+6g\left(4\sqrt{\varepsilon(\lambda)}\right) \\&=56\sqrt{k_0}\lambda^{1/4}\text{ } g\left(4\frac{1+\lambda(N + 2)}{\sqrt{k_0}\lambda^{5/4}}\right)+6g\left(4\sqrt{k_0}\lambda^{1/4}\right)\,, \label{distquantum3}
\ee
for all $N>0$ and $\lambda>0$ sufficiently small.
By using that
\bb
\lim\limits_{x\rightarrow0^+}g(x)&=0\,,\\
\lim\limits_{x\rightarrow0+}x\log_2 x&=0\,,\\
g(x)&=\log_2(ex)+o(1)\quad\text{($x\rightarrow\infty$)}\,,
\ee
it follows that for all $N>0$ it holds that
\begin{equation}\label{lim_lambda_Q}
    \lim\limits_{\lambda\rightarrow 0^+}\left|Q\left(\Phi_{\lambda,\sigma_{\lambda}},N\right)-Q\left(\Phi_{\lambda,\ketbrasub{n_\lambda}},N\right)\right|=0\,.
\end{equation}
In addition, Theorem~\ref{diehard_th} establishes that there exists $\bar{\epsilon}>0$ and $\bar{c}>0$ such that for all $\lambda\in(0,1/2-\bar{\varepsilon})$ it holds that $Q\left(\Phi_{\lambda,\ketbrasub{n_\lambda}},1/2\right)\ge \bar{c}$. As a consequence,~\eqref{lim_lambda_Q} guarantees that for $\lambda>0$ sufficiently small it holds that
\begin{equation}
    Q\left(\Phi_{\lambda,\sigma_{\lambda}},1/2\right)\ge Q\left(\Phi_{\lambda,\ketbrasub{n_\lambda}},1/2\right)-\frac{\bar{c}}{2}\ge \frac{\bar{c}}{2}\,.
\end{equation}
We conclude that the inequality
\begin{equation}
Q\left(\Phi_{\lambda,\sigma_{\lambda}}\right)\ge Q\left(\Phi_{\lambda,\sigma_{\lambda}},1/2\right)\ge \bar{c}/2\,,
\end{equation}
holds for $\lambda>0$ sufficiently small.
\end{proof}
Theorem~\ref{theorem_lambda_small} shows that qubits can be transmitted reliably across a very lossy optical fibre. This is accomplished if one follows the noise attenuation protocol, by using the two signals produced at the output of the interferometer depicted in Fig.~\ref{Interf2} as trigger signals. The key idea of the proof of Theorem~\ref{theorem_lambda_small} was that $\sigma_\lambda$ becomes closer and closer to the $n_\lambda$-th Fock state as $\lambda$ approaches $0$~\footnote{By looking at the signal $S_2$ path in Fig.~\ref{k2approx}, one can be tempted to think that the approximation  $\sigma_\lambda\stackrel{\lambda\rightarrow0^+}{\simeq}\ketbra{n_\lambda}$ is trivial. Indeed, in the limit $\lambda\rightarrow0^+$, the {BS} with transmissivity $\frac{1}{1+\lambda}$ is totally transmitting and the {BS}s with transmissivity $\lambda$ are totally reflective. Hence, in this limit, if we send a state $\bar{\rho}$ (independent of $\lambda$) in the $S_2$ arm of the {BS} with transmissivity $\frac{1}{1+\lambda}$ it is trivial that the final environment state satisfies $\sigma_\lambda\stackrel{\lambda\rightarrow0^+}{\simeq}\bar{\rho}$. However, since $\ket{n_\lambda}_{S_2}$ depends on $\lambda$, the fact that $\sigma_\lambda\stackrel{\lambda\rightarrow0^+}{\simeq}\ketbra{n_\lambda}$ is not trivial.}. In {the} Appendix~\ref{Appendix_one_trigger} we explain why sending just one trigger signal is not enough to obtain an environment state which is close to the $n_\lambda$-th Fock state. Moreover, notice that Theorem~\ref{theorem_lambda_small} holds in the interesting limit $\lambda\rightarrow 0^+$, where the quantum capacity of the thermal attenuator is zero and its energy-constrained entanglement-assisted capacities tend to zero. A schematic of the interaction between the two trigger signals and the environment, as analysed in Theorem~\ref{theorem_lambda_small}, is shown in Fig.~\ref{k2approx}. 
\begin{figure}[t]
	\centering
	\includegraphics[width=0.9\linewidth]{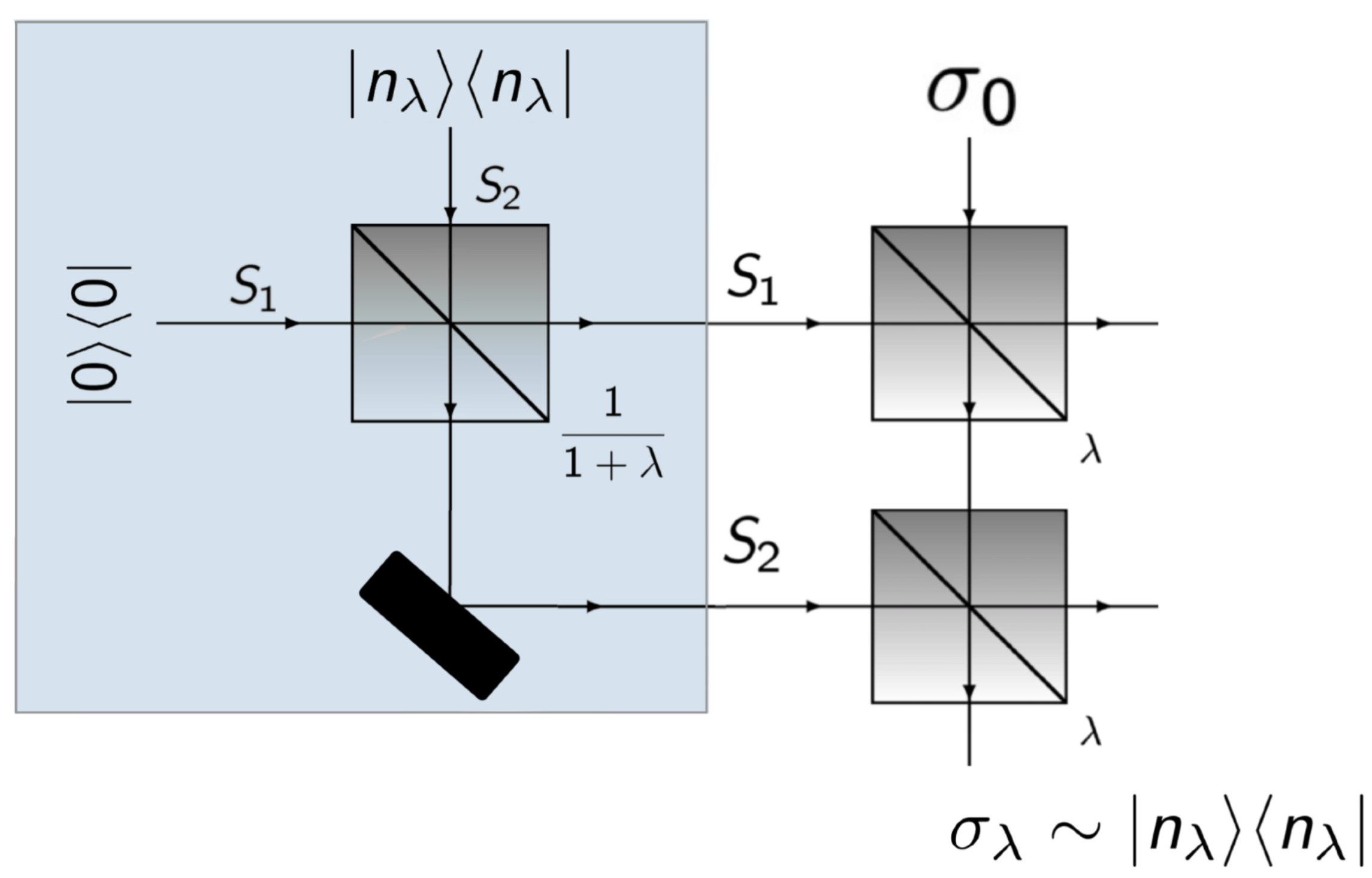}
	\caption{A schematic of what described in Theorem~\ref{theorem_lambda_small}. The blue box builds the two-signals trigger state $\ket{n_\lambda,\lambda}_{S_1S_2}$ which transforms the environment state in a state $\sigma_\lambda$ closer and closer to $\ketbra{n_\lambda}$ as $\lambda$ decreases. The approximation $\sigma_\lambda\sim \ketbra{n_\lambda}$ is valid if $\lambda>0$ is sufficiently small.}
	\label{k2approx}
\end{figure}

{The results of this section are based on the approximate expression in~\eqref{achievable} of the environment state $\sigma$ at the beginning of step~3 of the noise attenuation protocol. This approximation is justified when the time interval $\delta t$ between subsequent signals is sufficiently short so that we can use $\xi_{\delta t}\simeq \Id$. The exact expression of $\sigma$ is given by~\eqref{expl_sigma}. The expressions in~\eqref{achievable} and~\eqref{expl_sigma} coincide if $\xi_{\delta t}=\Id$.
In the forthcoming Theorem~\ref{teo_new} we formalise this approximation and we show that it is consistent. More precisely, we show that if the time interval $\delta t$ is sufficiently short, any scheme for quantum communication derived with the approximate expression in~\eqref{achievable} will also work with exact expression in~\eqref{expl_sigma}.

Theorem~\ref{FockAchiev} and Theorem~\ref{theorem_lambda_small} establish that if Alice sends a suitable number $k$ of trigger signals initialised in the state in~\eqref{trigger_signals_k}, then at the beginning of step~3 the environment is in a state $\sigma$, given by the approximate expression~\eqref{achievable}, such that
\begin{equation}\label{boundQ_inexacT}
    Q(\Phi_{\lambda,\sigma})\ge Q(\Phi_{\lambda,\sigma},1/2)>0\,.
\end{equation}
In the forthcoming Theorem~\ref{teo_new} we provide a generalisation of this result where we take into account the thermalisation process $\xi_{\delta t}$ in the regime where $\delta t$ is sufficiently short. Indeed, Theorem~\ref{teo_new} guarantees that if Alice sends a suitable number $k$ of trigger signals initialised in the state in~\eqref{trigger_signals_k} and separated by a sufficiently short time interval $\delta t$, then at the beginning of step~3 the environment is in a state $\sigma'_{\delta t}$, given by the exact expression~\eqref{expl_sigma}, such that 
\begin{equation}\label{boundQ_exacT}
    Q(\Phi_{\lambda,\sigma'_{\delta t}})\ge Q(\Phi_{\lambda,\sigma'_{\delta t}},1/2)>0\,.
\end{equation}
We show this result under suitable hypotheses on $\xi_{\delta t}$. Namely, we require that for $\delta t$ sufficiently short $\xi_{\delta t}$ is close in energy-constrained diamond norm to the identity channel and it maps the subset of states with a finite mean photon number in itself.

\begin{thm}\label{teo_new}
Let $\sigma_0\in\mathfrak{S}(\HH_E)$ such that $\langle b^\dagger b\rangle_{\sigma_0}<\infty$. Let $\lambda\in(0,1)$, $k\in\mathbb{N}$. Let $\rho^{(k)}\in\mathfrak{S}\left(\HH_{S_1}\otimes\ldots\otimes\HH_{S_k}\right)$ such that: $$\left\langle \sum_{i=1}^k(a_i)^\dagger a_i\right\rangle_{\rho^{(k)}}<\infty\,,$$
$$Q(\Phi_{\lambda,\sigma},1/2)>0\,,$$
where
\begin{align}\label{sigma_expr}
	\sigma&\coloneqq \Tr_{S_1\ldots S_k}\left[  \mathcal{U}_{\lambda}^{(S_kE)} \circ\mathcal{U}_{\lambda}^{(S_{k-1}E)}\ldots \right. \nonumber \\
    &\left.\qquad\ldots \circ\mathcal{U}_{\lambda}^{(S_{1}E)} \left(\rho^{(k)}\otimes \sigma_0\right) \right]\,,
\end{align}
with $\mathcal{U}_{\lambda}^{(S_iE)}$ being a quantum channel defined by $\mathcal{U}_{\lambda}^{(S_iE)}(\cdot)=U_{\lambda}^{(S_iE)}(\cdot) \left(U_{\lambda}^{(S_iE)}\right)^\dagger$.
Let $\{\xi_{\delta t}\}_{\delta t\ge 0}$ a one-parameter family of quantum channels on $\mathfrak{S}(\HH_E)$ such that:
\begin{enumerate}[(1)]
    \item For all $N>0$ there exists a function $\varepsilon_N:\mathbb{R}\mapsto\mathbb{R}$, which is continuous in $\delta t=0$ and satisfies $\varepsilon_N(0)=0$, for which the energy-constrained diamond distance between $\xi_{\delta t}$ and $\Id$ satisfies
\begin{equation}\label{vareps}
\left\|\xi_{\delta t}-\Id\right\|_{\diamond N}\le \varepsilon_N(\delta t)\,.
\end{equation}
    \item For $\delta t$ sufficiently short it holds that $\langle b^\dagger b\rangle_{\xi_{\delta t}(\tilde{\sigma})}<\infty$ for all $\tilde{\sigma}\in\mathfrak{S}(\HH_E)$ such that $\langle b^\dagger b \rangle_{\tilde{\sigma}}<\infty$.
\end{enumerate}
Let us define $\sigma'_{\delta t}\in\mathfrak{S}(\HH_E)$ as
\begin{align}\label{sigma'}
	\sigma'_{\delta t} &\coloneqq \Tr_{S_1\ldots S_k}\left[\xi_{\delta t}\circ \mathcal{U}_{\lambda}^{(S_kE)}\circ \xi_{\delta t}\circ\mathcal{U}_{\lambda}^{(S_{k-1}E)}\circ\ldots \right. \nonumber \\
    &\left.\qquad\ldots \circ \xi_{\delta t}\circ\mathcal{U}_{\lambda}^{(S_{1}E)} \left(\rho^{(k)}\otimes \sigma_0\right) \right]
\end{align}
For $\delta t$ sufficiently short it holds that
\begin{equation}\label{tesi_teo_nuovo}
    Q(\Phi_{\lambda,\sigma_{\delta t}})\ge Q(\Phi_{\lambda,\sigma_{\delta t}},1/2)> \frac{Q(\Phi_{\lambda,\sigma},1/2)}{2} >0\,.
\end{equation}
\end{thm}
 
\begin{proof}
Let us define $\mathcal{U}_{\lambda,0}(\delta t)\coloneqq \Id$ and
\begin{equation}
    \mathcal{U}_{\lambda,j}(\delta t)\coloneqq \xi_{\delta t}\circ \mathcal{U}_{\lambda}^{(S_jE)}\circ \xi_{\delta t}\circ\mathcal{U}_{\lambda}^{(S_{j-1}E)}\ldots \circ \xi_{\delta t}\circ\mathcal{U}_{\lambda}^{(S_{1}E)}\,
\end{equation}
for all $j\in\{1,2,\ldots,k\}$. By exploiting the hypothesis (2) on $\xi_{\delta t}$, we can fix $\delta t$ sufficiently short such that $\langle b^\dagger b\rangle_{\xi_{\delta t}(\tilde{\sigma})}<\infty$ for all $\tilde{\sigma}\in\mathfrak{S}(\HH_E)$ such that $\langle b^\dagger b \rangle_{\tilde{\sigma}}<\infty$.

We start by proving by induction that for all $j\in\{0,1,\ldots,k\}$ it holds that
\begin{small}
\bb\label{induction_ener}
    &N_j\coloneqq \Tr_{ES_1\ldots S_k}\left[\mathcal{U}_{\lambda,j}(\delta t) \left(\rho^{(k)}\otimes \sigma_0\right)\right.  \\
    &\left.\qquad \times \left(b^\dagger b+\sum_{i=1}^k(a_i)^\dagger a_i\right)\right]<\infty\,.
\ee
\end{small}
For $j=0$, the left hand side reduces to $\left\langle \sum_{i=1}^k(a_i)^\dagger a_i\right\rangle_{\rho^{(k)}}+\langle b^\dagger b\rangle_{\sigma_0}$, which is finite thanks to the hypotheses on $\rho^{(k)}$ and $\sigma_0$.

By assuming~\eqref{induction_ener} valid for $j-1$, we need to prove it for $j$. By using 
\begin{small}
\begin{equation}\label{recursive_U}
    \mathcal{U}_{\lambda,j}(\delta t)=\xi_{\delta t}\circ\mathcal{U}_{\lambda}^{(S_jE)}\circ \mathcal{U}_{\lambda,j-1}(\delta t)\,,
\end{equation}
\end{small}
in order to prove~\eqref{induction_ener} it suffices to prove that 
\begin{small}
\begin{equation}\label{first_term}
    \Tr_{ES_1\ldots S_k}\left[\xi_{\delta t}\circ\mathcal{U}_{\lambda}^{(S_jE)}\circ \mathcal{U}_{\lambda,j-1}(\delta t) \left(\rho^{(k)}\otimes \sigma_0\right) b^\dagger b \right]<\infty\,
\end{equation}
\end{small}
 and
\begin{small}
\bb\label{second_term}
    &\Tr_{ES_1\ldots S_k}\left[\xi_{\delta t}\circ\mathcal{U}_{\lambda}^{(S_jE)}\circ \mathcal{U}_{\lambda,j-1}(\delta t)\left(\rho^{(k)}\otimes \sigma_0\right) \right.  \\
    &\left.\qquad\times \left(\sum_{i=1}^k(a_i)^\dagger a_i\right)\right]<\infty\,.
\ee
\end{small}First, let us show~\eqref{first_term}. By exploiting the hypothesis (2) on $\xi_{\delta t}$,~\eqref{first_term} is valid if 
\begin{small}
\begin{equation}\label{eq_first_termine}
    \Tr_{ES_1\ldots S_k}\left[\mathcal{U}_{\lambda}^{(S_jE)}\circ \mathcal{U}_{\lambda,j-1}(\delta t) \left(\rho^{(k)}\otimes \sigma_0\right) b^\dagger b \right]<\infty\,.
\end{equation}
\end{small}
\eqref{eq_first_termine} is true since
\begin{small}
\begin{align*}
    & \Tr_{ES_1\ldots S_k}\left[\mathcal{U}_{\lambda}^{(S_jE)}\circ \mathcal{U}_{\lambda,j-1}(\delta t) \left(\rho^{(k)}\otimes \sigma_0\right) b^\dagger b\right]\\&\le  \Tr_{ES_1\ldots S_k}\left[\mathcal{U}_{\lambda}^{(S_jE)}\circ \mathcal{U}_{\lambda,j-1}(\delta t) \left(\rho^{(k)}\otimes \sigma_0\right)\right.\nonumber  \\
    &\left.\qquad\times \left(b^\dagger b+\sum_{i=1}^k(a_{i})^\dagger a_i\right) \right]\\&= \Tr_{ES_1\ldots S_k}\left[ \mathcal{U}_{\lambda,j-1}(\delta t) \left(\rho^{(k)}\otimes \sigma_0\right)\right.\nonumber  \\
    &\left.\qquad\times \left(b^\dagger b+\sum_{i=1}^k(a_{i})^\dagger a_i\right) \right]<\infty\,,
\end{align*}
\end{small}
where we used the fact that $\mathcal{U}_{\lambda}^{(S_jE)}$ preserves the mean of the total photon number operator $b^\dagger b+\sum_{i=1}^k(a_{i})^\dagger a_i$ (as one can show by exploiting Lemma~\ref{lemmatrasf} in the Appendix) and we used the inductive hypothesis.

Second, let us show~\eqref{second_term}. By exploiting the fact that $\xi_{\delta t}$ is trace-preserving, it follows that  
\begin{small}
\begin{align*}
    &\Tr_{ES_1\ldots S_k}\left[\xi_{\delta t}\circ\mathcal{U}_{\lambda}^{(S_jE)}\circ \mathcal{U}_{\lambda,j-1}(\delta t) \left(\rho^{(k)}\otimes \sigma_0\right)\right.\nonumber  \\
    &\left.\qquad\times  \left(\sum_{i=1}^k(a_i)^\dagger a_i\right)\right]\\&=\Tr_{ES_1\ldots S_k}\left[\mathcal{U}_{\lambda}^{(S_jE)}\circ \mathcal{U}_{\lambda,j-1}(\delta t) \left(\rho^{(k)}\otimes \sigma_0\right)\right.\nonumber  \\
    &\left.\qquad\times \left(\sum_{i=1}^k(a_i)^\dagger a_i\right)\right]\\&\le  \Tr_{ES_1\ldots S_k}\left[\mathcal{U}_{\lambda}^{(S_jE)}\circ \mathcal{U}_{\lambda,j-1}(\delta t) \left(\rho^{(k)}\otimes \sigma_0\right)\right.\nonumber  \\
    &\left.\qquad\times  \left(b^\dagger b+\sum_{i=1}^k(a_{i})^\dagger a_i\right) \right]\\&= \Tr_{ES_1\ldots S_k}\left[ \mathcal{U}_{\lambda,j-1}(\delta t) \left(\rho^{(k)}\otimes \sigma_0\right)\right.\nonumber  \\
    &\left.\qquad\times  \left(b^\dagger b+\sum_{i=1}^k(a_{i})^\dagger a_i\right) \right]<\infty\,.
\end{align*}
\end{small}Hence, we have proved~\eqref{induction_ener}. Now, let us show that 
\begin{equation}\label{diff_sigma}
    \|\sigma'_{\delta t}-\sigma\|_1\le k\,\varepsilon_{\bar{N}}(\delta t)\,,
\end{equation}
where 
\begin{equation}
    \bar{N}\coloneqq \max_{j\in\{0,1,\ldots,k\}}N_j\,
\end{equation}
and $N_j$ is defined in~\eqref{induction_ener}. Notice that~\eqref{induction_ener} implies that $\bar{N}$ is finite. By noting that
\begin{small}
\bb
    	\sigma&= \Tr_{S_1\ldots S_k}\left[  \mathcal{U}_{\lambda,k}(0) \left(\rho^{(k)}\otimes \sigma_0\right) \right]\,,\\
    	\sigma'_{\delta t} &= \Tr_{S_1\ldots S_k}\left[  \mathcal{U}_{\lambda,k}(\delta t) \left(\rho^{(k)}\otimes \sigma_0\right) \right]\,,
\ee
\end{small}
it follows that in order to prove~\eqref{diff_sigma} it suffices to prove by induction that for all $j\in\{0,1,\ldots,k\}$ it holds that
\begin{small}
\bb\label{induction_eq}
    &\left\|\Tr_{S_1\ldots S_j}\left[  \mathcal{U}_{\lambda,j}(\delta t) \left(\rho^{(k)}\otimes \sigma_0\right) \right]\right.  \\
    &\left.\qquad -\Tr_{S_1\ldots S_j}\left[  \mathcal{U}_{\lambda,j}(0) \left(\rho^{(k)}\otimes \sigma_0\right) \right]\right\|_1\le j\,\varepsilon_{\bar{N}}(\delta t)\,,
\ee
\end{small}
where $\Tr_{S_1\ldots S_j}$ is understood to be $\Id$ if $j=0$.
For $j=0$,~\eqref{induction_eq} is valid since both its left hand side and its right hand side are equal to $0$. By assuming~\eqref{induction_eq} valid for $j-1$, we need to prove it for $j$. By using~\eqref{recursive_U} and the triangular inequality of the trace norm, we have that
\begin{small}
\begin{align}
      &\left\|\Tr_{S_1\ldots S_j}\left[  \mathcal{U}_{\lambda,j}(\delta t) \left(\rho^{(k)}\otimes \sigma_0\right) \right]\right. \nonumber \\
    &\left.\qquad-\Tr_{S_1\ldots S_j}\left[  \mathcal{U}_{\lambda,j}(0) \left(\rho^{(k)}\otimes \sigma_0\right) \right]\right\|_1\nonumber\\& \le
    \left\|(\xi_{\delta t}-\Id)\left(\Tr_{S_1\ldots S_j}\left[ \mathcal{U}_{\lambda}^{(S_jE)}\circ \mathcal{U}_{\lambda,j-1}(\delta t) \left(\rho^{(k)}\otimes \sigma_0\right) \right]\right)\right\|_1 \nonumber\\& \quad +\left\|\Tr_{S_1\ldots S_j}\left[ \mathcal{U}_{\lambda}^{(S_jE)}\left( \mathcal{U}_{\lambda,j-1}(\delta t) \left(\rho^{(k)}\otimes \sigma_0\right) \right.\right.\right. \nonumber \\
    &\left.\left.\left.\qquad- \mathcal{U}_{\lambda,j-1}(0) \left(\rho^{(k)}\otimes \sigma_0\right) \right)\right]\right\|_1\,.\label{ineq_new_th}
\end{align}
\end{small}
We can upper bound the first term with $\varepsilon_{\bar{N}}(\delta t)$. Indeed, the mean environmental photon number of the state $\Tr_{S_1\ldots S_j}\left[ \mathcal{U}_{\lambda}^{(S_jE)}\circ \mathcal{U}_{\lambda,j-1}(\delta t) \left(\rho^{(k)}\otimes \sigma_0\right) \right]$ satisfies
\begin{small}
\begin{align*}
    &\Tr_{ES_{j+1}\ldots S_k}\left[\Tr_{S_1\ldots S_j}\left[ \mathcal{U}_{\lambda}^{(S_jE)}\circ \mathcal{U}_{\lambda,j-1}(\delta t) \left(\rho^{(k)}\otimes \sigma_0\right) \right] b^\dagger b\right]\\&= \Tr_{ES_1\ldots S_k}\left[\mathcal{U}_{\lambda}^{(S_jE)}\circ \mathcal{U}_{\lambda,j-1}(\delta t) \left(\rho^{(k)}\otimes \sigma_0\right) b^\dagger b\right]\\&\le \Tr_{ES_1\ldots S_k}\left[\mathcal{U}_{\lambda}^{(S_jE)}\circ \mathcal{U}_{\lambda,j-1}(\delta t) \left(\rho^{(k)}\otimes \sigma_0\right)\right.\nonumber  \\
    &\left.\qquad\times \left(b^\dagger b+\sum_{i=1}^k (a_i)^\dagger a_i\right)\right] \\&= \Tr_{ES_1\ldots S_k}\left[\mathcal{U}_{\lambda,j-1}(\delta t)\left(\rho^{(k)}\otimes \sigma_0\right) \right.\nonumber  \\
    &\left.\qquad\times\left(b^\dagger b+\sum_{i=1}^k (a_i)^\dagger a_i\right)\right] =N_{j-1}\le \bar{N}\,,
\end{align*}
\end{small}
where we used the fact that $\mathcal{U}_{\lambda}^{(S_jE)}$ preserves the mean of the total photon number operator $b^\dagger b+\sum_{i=1}^k(a_{i})^\dagger a_i$.
Hence, the state $\Tr_{S_1\ldots S_j}\left[ \mathcal{U}_{\lambda}^{(S_jE)}\circ \mathcal{U}_{\lambda,j-1}(\delta t) \left(\rho^{(k)}\otimes \sigma_0\right) \right]$ can be used to lower bound the supremum problem which defines $\left\|\xi_{\delta t}-\Id\right\|_{\diamond \bar{N}}$, which is less or equal to $\varepsilon_{\bar{N}}(\delta t)$ thanks to~\eqref{vareps}.

In addition, by exploiting the contractivity of the trace norm under partial traces and the invariance of the trace norm under unitary transformations,~\eqref{ineq_new_th} implies that
\begin{small}
\begin{align*}
      &\left\|\Tr_{S_1\ldots S_j}\left[  \mathcal{U}_{\lambda,j}(\delta t) \left(\rho^{(k)}\otimes \sigma_0\right) \right]\right.  \\
    &\left.\qquad-\Tr_{S_1\ldots S_j}\left[  \mathcal{U}_{\lambda,j}(0) \left(\rho^{(k)}\otimes \sigma_0\right) \right]\right\|_1\\& \le
      \varepsilon_{\bar{N}}(\delta t)+\left\|\Tr_{S_1\ldots S_{j-1}}\left[ \mathcal{U}_{\lambda,j-1}(\delta t) \left(\rho^{(k)}\otimes \sigma_0\right)\right.\right. \\
    &\left.\left.\qquad - \mathcal{U}_{\lambda,j-1}(0) \left(\rho^{(k)}\otimes \sigma_0\right) \right]\right\|_1\,.
\end{align*}
\end{small}
By using the inductive assumption, we have proved~\eqref{induction_eq}. Consequently, also~\eqref{diff_sigma} is proved.

Now, we invoke Lemma~14 which states that the energy-constrained quantum capacity of the general attenuators is continuous
with respect to the environment state over the subset of environment states having a finite mean photon number. Consequently, by exploiting~\eqref{diff_sigma}, the continuity of $\varepsilon_{\bar{N}}(\delta t)$ in $\delta t=0$, and the fact that $\varepsilon_{\bar{N}}(0)=0$, it follows that $Q(\Phi_{\lambda,\sigma'_{\delta t}},1/2)>Q(\Phi_{\lambda,\sigma},1/2)/2$ for $\delta t$ sufficiently small. Hence,~\eqref{tesi_teo_nuovo} is finally proved.

Note that we can apply Lemma~14 since we can prove that $\sigma$ and $\sigma'_{\delta t}$ have finite mean photon number. Indeed, the mean photon number of $\sigma'_{\delta t}$ satisfies
\begin{small}
\begin{align*}
    \Tr_E \left[\sigma'_{\delta t}b^\dagger b\right]&\le \Tr_{ES_1\ldots S_k}\left[\mathcal{U}_{\lambda,k}(\delta t) \left(\rho^{(k)}\otimes \sigma_0\right)\right.\nonumber  \\
    &\left.\qquad\times \left(b^\dagger b+\sum_{i=1}^k(a_i)^\dagger a_i\right)\right]=N_k<\infty\,,
\end{align*}
\end{small}where we used~\eqref{induction_ener}. 
Since $\sigma=\sigma'_{\delta t=0}$, also the mean photon number of $\sigma$ is finite.
\end{proof}

\begin{ex}
Let $\sigma_0=\ketbra{0}$. An example of thermalisation process $\xi_{\delta t}$ which satisfies the hypotheses of Theorem~\ref{teo_new} is $\xi_{\delta t}\coloneqq \Phi_{\eta(\delta t),\ketbra{0}}$, where $\eta(\cdot):[0,\infty)\mapsto[0,1]$ is such that: it is continuous in $\delta t=0$, $\eta(0)=1$, and $\eta(\delta t)=0$ for $\delta t\ge t_E$. Notice that $\xi_{0}=\Id$ and $\xi_{\delta t}(\sigma)=\ketbra{0}$ for all $\sigma$ and all $\delta t\ge t_E$.

The hypothesis (1) of Theorem~\ref{teo_new} is satisfied since for all $N>0$ it holds that~\cite[Section 4.B]{VV-diamond}
\bb\label{ECdiamon_pureloss}
\lim\limits_{\lambda\rightarrow 1^-}\|\Phi_{\lambda,\ketbra{0}}-\Id\|_{\diamond N}=0\,.
\ee
The reason why we expressed hypothesis (1) of Theorem~\ref{teo_new} in terms of the energy-constrained diamond norm and not simply in terms of the unconstrained diamond norm is that the latter has undesirable properties in the context of continuous variables systems. For example, the unconstrained version of~\eqref{ECdiamon_pureloss}, i.e.~ $\lim\limits_{\lambda\rightarrow 1^-}\|\Phi_{\lambda,\ketbra{0}}-\Id\|_{\diamond }=0$, is false~\cite[Proposition 1]{VV-diamond}.

The hypothesis (2) of Theorem~\ref{teo_new} is satisfied since it holds that $\langle b^\dagger b\rangle_{\Phi_{\lambda,\ketbra{0}}(\sigma)}=\lambda \langle b^\dagger b\rangle_{\sigma}$, thanks to Lemma~\ref{LemmaEner} in the Appendix.
\end{ex}

}

\section{Discussion and conclusions}
In Sec.~\ref{sec_Cea}, we {have} studied the transmission of classical and quantum information on general attenuators. Most notably, we {have} found that for arbitrarily small non-zero values of the transmissivity $\lambda$, if the environment is in a suitable state $\sigma$, then entanglement-assisted classical communication through $\Phi_{\lambda,\sigma}$ achieves better performance than unassisted classical communication through the noiseless channel. In other words, by controlling the environment state and by consuming pre-shared entanglement, it is possible to reliably transmit bits with better performance than in the ideal case of absence of noise. This property holds even when the channel is so noisy that $\lambda>0$ is very small. 

In mathematical terms, from numerical investigations regarding the positivity of $I_{\text{coh}}\left(\Phi_{\lambda,\ketbrasub{n}},\tau_N\right)$, we {have been} 
led to conjecture that
$$\sup_{n\in\mathbb{N}}C_{\text{ea}}\left(\Phi_{\lambda,\ketbrasub{n}},N\right)> C\left(\Id,N\right)$$
for any $N>0$ and $\lambda\in(0,1/2)$. In Sec.~\ref{subsec_Cea} we {have} provided a proof of this in the (seemingly worst) case where $\lambda>0$ is sufficiently small for a wide range of values of $N$ (those depicted in the plot in Fig.~\ref{entropy_diff}).

Since $C_{\text{ea}}=2Q_{\text{ea}}$, all the conclusions we have analysed so far can be reformulated in the context in which the task is the transmission of qubits.


On the technical level, we {have} introduced the `master equation trick', that {has} allowed us to simplify the calculations involving thermal attenuators. We believe that this trick may be of independent interest: for example, {it provides a simple Kraus representation of the thermal attenuator (see Theorem~\ref{TeoKraus}), allowing one to obtain simple expressions for the action of the thermal attenuator on a generic operator. In addition,} it can also be useful in estimating the capacities of those general attenuators which have a Fock-diagonal environment state (see Lemma~\ref{lowcap} in the Appendix).

In Sec.~\ref{subsec_mon} we {have} shown that whenever $\sigma$ is a pure state of definite parity, either even or odd, we have that $$C_{\text{ea}}\left(\Phi_{1/2,\sigma},N\right)=g(N)\quad\forall\,N>0\,.$$ Then we {have} found examples of Fock states $\ketbra{n}$ such that $C_{\text{ea}}\left(\Phi_{\lambda,\ketbrasub{n}},N\right)$ is not monotonic in $\lambda$.

In Sec.~\ref{sec_control}, we {have} explained how to implement in a operational way the control of the environment state. This can be done by exploiting the memory effects that arise in a realistic model of communication when the temporal intervals at which signals are fed into the optical fibre are sufficiently short. This fact {has} led us to devise the `noise attenuation protocol', whose goal is to manipulate the environment state in order to facilitate communication. Basically, one could transmit using sequences of $k+1$ very close signals: the first $k$ --- dubbed `trigger signals' --- are used to induce suitable modifications of the environment state, while the last one carries the actual information. We {have} shown that (Theorem~\ref{FockAchiev}) the noise attenuation protocol allows one to achieve not only the entanglement-assisted communication performance discussed in Sec.~\ref{sec_Cea}, but also the quantum communication performance discovered in~\cite{die-hard} (i.e.\ the phenomenon of {D-HQCOM}). In addition, the multipartite state of $k$ trigger signals needed to activate these enhanced capabilities can be produced by using the interferometer in Fig.~\ref{InterfK}. In particular (Theorem~\ref{theorem_lambda_small}), for sufficiently low non-zero values of the transmissivity, the noise attenuation protocol works with only two trigger signals. 
This implies that sufficiently long optical fibres can transmit qubits at a constant rate. 

We have proved that arbitrarily long optical fibres can have very efficient communication performance, provided that the noise attenuation protocol is applied with the trigger signals produced by the interforemeter in Fig.~\ref{InterfK}. Such an inteferometer requires a large Fock state in input. This is a problem since nowadays it is experimentally not known how to generate deterministically Fock states with high fidelity and Fock number larger than five. However, it is expected that this will be known in the near future~\cite{generation_fock_state}. From the theoretical point of view, one can try to overcome this problem by finding other environment states $\sigma(\lambda)$ such that $Q(\Phi_{\lambda,\sigma(\lambda)})>c>0$ and such that $\sigma(\lambda)$ is $(\lambda,\tau_\nu)-$achievable by means of trigger signals which can be produced without large Fock states.

Another intriguing open problem, relevant for physical realisations, is to consider the case in which Alice is not able to send signals separated by $\delta t\ll t_E$, where $t_E$ is the time after which the thermalisation process resets the environment state into the stationary environment state. In this case, we should take into account the resetting dynamics $\xi_{\delta t}$ during step~2 of the noise attenuation protocol and hence in~\eqref{achievable}.

{Furthermore, including the queuing framework, that has been recently
proposed for quantum communication networks~\cite{queue1,queue2,queue3}, can be an interesting development of the present paper. The queuing framework takes into account the decoherence, due to imperfections of Alice's quantum memory, which affects the signals as they wait to be fed into the optical fibre. }

An interesting experimental research would be to study memory effects in optical fibres, e.g.~to estimate $t_E$.

To summarise, our work shows that repeaterless quantum communication across arbitrarily long optical fibres is possible at a constant rate, provided that one exploits the noise attenuation protocol. This protocol, combined with entanglement-assistance, allows a sender to transfer bits or qubits across an arbitrarily long optical fibre at a rate of the same order of the maximum achievable in the unassisted noiseless scenario.

\smallskip
\textbf{\em Acknowledgements.}--- {FAM and VG acknowledge financial support by MIUR (Ministero dell'Istruzione, dell'Universit\`a e della Ricerca) via project PRIN 2017 ``Taming complexity via Quantum Strategies: a Hybrid Integrated Photonic approach'' (QUSHIP) Id.\ 2017SRNBRK, and via project PRO3 ``Quantum Pathfinder''.} LL acknowledges financial support from the Alexander von Humboldt Foundation.

\bibliographystyle{unsrt}
\bibliography{biblio}
\onecolumngrid
\appendix
\counterwithin{lemma_app}{section}
\counterwithin{definition_app}{section}
\counterwithin{remark}{section}
\section{Useful lemmas}
\begin{lemma_app}\label{subadd}
	Let $\rho\in\mathfrak{S}(\HH_S)$, $\sigma\in\mathfrak{S}(\HH_E)$ and $\lambda\in[0,1]$. It holds that
	$$S\left(\tilde{\Phi}_{\lambda,\sigma}(\rho)\right)\le S(\sigma)+S\left(\tilde{\Phi}_{\lambda,\sigma}^{\text{wc}}(\rho)\right)\,.$$
\end{lemma_app}
\begin{proof}
	Let $\ket{0}_{EE'}\in\HH_E\otimes\HH_{E'}$ be a purification of $\sigma$, where $\HH_{E'}$ is a fictitious purifying Hilbert space: $\Tr_{E'}\ketbra{0}_{EE'}=\sigma\,.$
	Then the map $\tilde{\Phi}_{\lambda,\sigma}: \mathfrak{S}(\HH_S)\mapsto\mathfrak{S}(\HH_E\otimes\HH_{E'})$, defined by
	$${\tilde{\Phi}}_{\lambda,\sigma}(\rho)\coloneqq\Tr_{S}\left[U_{\lambda}^{(SE)}\otimes\mathbb{1}_{E'}\left(\rho\otimes\ketbra{0}_{EE'}\right)\left(U_{\lambda}^{(SE)}\right)^\dagger\otimes\mathbb{1}_{E'}\right]\,,$$
	is a complementary quantum channel of $\Phi_{\lambda,\sigma}$, indeed: 
	\bb
	    &\Tr_{EE'}\left[U_{\lambda}^{(SE)}\otimes\mathbb{1}_{E'}\left(\rho\otimes\ketbra{0}_{EE'}\right)\left(U_{\lambda}^{(SE)}\right)^\dagger\otimes\mathbb{1}_{E'}\right] =\Tr_E\left[U_{\lambda}^{(SE)}\rho\otimes\sigma\left(U_{\lambda}^{(SE)}\right)^\dagger\right]=\Phi_{\lambda,\sigma}(\rho)\,. 
	\ee
	Moreover, it holds that
	\begin{equation}\label{trE'}
	\Tr_{E'}{\tilde{\Phi}}_{\lambda,\sigma}(\rho)={\tilde{\Phi}^\text{wc}}_{\lambda,\sigma}(\rho) 
	\end{equation}
	and
	\bb\label{smidt}
	S\left(\Tr_{E}{\tilde{\Phi}}_{\lambda,\sigma}(\rho)\right) &=S\left(\Tr_{ES}\left[U_{\lambda}^{(SE)}\otimes\mathbb{1}_{E'}\left(\rho\otimes\ketbra{0}_{EE'}\right)\left(U_{\lambda}^{(SE)}\right)^\dagger\otimes\mathbb{1}_{E'}\right]\right) \\&=S\left(\Tr_{E}\ketbra{0}_{EE'}\right)=S\left(\Tr_{E'}\ketbra{0}_{EE'}\right)=S(\sigma)\,,
	\ee
	where we used the fact that $\ketbra{0}_{EE'}$ is pure. Exploiting the subadditivity inequality of the von Neumann entropy and subsequently the identities~\eqref{trE'} and~\eqref{smidt}, we conclude that
	\bb
	&S\left({\tilde{\Phi}}_{\lambda,\sigma}(\rho)\right)\le S\left(\Tr_{E}{\tilde{\Phi}}_{\lambda,\sigma}(\rho) \right) +S\left(\Tr_{E'}{\tilde{\Phi}}_{\lambda,\sigma}(\rho) \right)= S(\sigma)+S\left(\tilde{\Phi}_{\lambda,\sigma}^{\text{wc}}(\rho)\right)\,.
	\ee
	
\end{proof}

\begin{lemma_app}\label{lemmatrasf}
	The annihilation operators in Heisenberg representation are given by:
	\begin{equation}\label{trasfheisa}
	\left(U_{\lambda}^{(SE)}\right)^\dagger a\,U_{\lambda}^{(SE)}=\sqrt{\lambda}a+\sqrt{1-\lambda}b\,;
	\end{equation}
	\begin{equation}\label{trasfheisb}
	\left(U_{\lambda}^{(SE)}\right)^\dagger b\,U_{\lambda}^{(SE)}=-\sqrt{1-\lambda}a+\sqrt{\lambda}b\,;
	\end{equation}
	\begin{equation}\label{trasfheisa2}
	{U_{\lambda}^{(SE)}} a\left(U_{\lambda}^{(SE)}\right)^\dagger=\sqrt{\lambda}a-\sqrt{1-\lambda}b\,;
	\end{equation}
	\begin{equation}\label{trasfheisb2}
	{U_{\lambda}^{(SE)}} b\left(U_{\lambda}^{(SE)}\right)^\dagger=\sqrt{1-\lambda}a+\sqrt{\lambda}b\,.
	\end{equation}
\end{lemma_app}
\begin{proof}
	By setting $\eta\coloneqq \arccos\sqrt{\lambda}$ i.e.~$\lambda(\eta)=\cos^2(\eta
	)$, and $$\hat{f}(\eta)\coloneqq	\left(U_{\lambda(\eta)}^{(SE)}\right)^\dagger a\,U_{\lambda(\eta)}^{(SE)}=e^{-\eta\left(a^\dagger b-ab^\dagger\right)}a\,e^{\eta\left(a^\dagger b-ab^\dagger\right)}\,,$$
	one obtains:
	\bb\label{f'}
	\hat{f}'(\eta)&=-\left(U_{\lambda(\eta)}^{(SE)}\right)^\dagger [a^\dagger b-ab^\dagger,a]\,U_{\lambda(\eta)}^{(SE)}=	\left(U_{\lambda(\eta)}^{(SE)}\right)^\dagger b\,U_{\lambda(\eta)}^{(SE)}\,,
	\ee
	\bb\label{f'}
	\hat{f}'(\eta)&=-\left(U_{\lambda(\eta)}^{(SE)}\right)^\dagger [a^\dagger b-ab^\dagger,b]\,U_{\lambda(\eta)}^{(SE)}=	-\left(U_{\lambda(\eta)}^{(SE)}\right)^\dagger a\,U_{\lambda(\eta)}^{(SE)}=-\hat{f}(\eta)\,.
	\ee
	Therefore, there exist two operators $\hat{c_0}$ and $\hat{c_1}$ such that
	\begin{equation}\label{formfeta}
	\hat{f}(\eta)=\cos\eta \,\hat{c_0}+\sin\eta \,\hat{c_1}\,.
	\end{equation}
	By imposing $\hat{f}(0)=a$ and $\hat{f}'(0)=b$, we arrive at
	$$\hat{f}(\eta)=\cos\eta \,a+\sin\eta\, b$$
	which implies~\eqref{trasfheisa} to be valid. Furthermore,~\eqref{trasfheisb} follows from~\eqref{f'}. To conclude the proof,~\eqref{trasfheisa2} follows from~\eqref{trasfheisa} by substituting $b\rightarrow-b$, while~\eqref{trasfheisb2} follows from~\eqref{trasfheisb} by substituting $a\rightarrow-a$.
\end{proof}

\begin{lemma_app}\label{LemmaEner}
The mean photon number of the output system and environment are given by
\begin{equation}\label{enphi}
	\langle a^\dagger a\rangle_{\Phi_{\lambda,\sigma}(\rho)} =\lambda\langle a^\dagger a\rangle_\rho +(1-\lambda)\langle b^\dagger b\rangle_\sigma+2\sqrt{\lambda(1-\lambda)}\Re\left(\langle a\rangle_\rho \langle b^\dagger\rangle_\sigma\right)
\end{equation}
and
\begin{equation}\label{enweak}
\langle b^\dagger b\rangle_{\tilde{\Phi}^{\text{wc}}_{\lambda,\sigma}(\rho)} =(1-\lambda)\langle a^\dagger a\rangle_\rho +\lambda\langle b^\dagger b\rangle_\sigma-2\sqrt{\lambda(1-\lambda)}\Re\left(\langle a\rangle_\rho \langle b^\dagger\rangle_\sigma\right)\,,
\end{equation}
respectively.
\end{lemma_app}
\begin{proof}
From Lemma~\ref{lemmatrasf} we have
\bb
\langle a^\dagger a\rangle_{\Phi_{\lambda,\sigma}(\rho)} &=	\Tr_S\left[\Phi_{\lambda,\sigma}(\rho)\,a^\dagger a\right]=\Tr_{SE}\left[U_\lambda^{(SE)} \rho\otimes\sigma \left({U_\lambda^{(SE)}}\right)^\dagger a^\dagger a\right]\\&=\Tr_{SE}\left[\rho\otimes\sigma \left(\left({U_\lambda^{(SE)}}\right)^\dagger a^\dagger U_\lambda^{(SE)}\right)\left(\left({U_\lambda^{(SE)}}\right)^\dagger a U_\lambda^{(SE)}\right)\right]\\&=\Tr_{SE}\left[\rho\otimes\sigma \left(\sqrt{\lambda}a^\dagger+\sqrt{1-\lambda}b^\dagger\right)\left(\sqrt{\lambda}a+\sqrt{1-\lambda}b\right)\right]\\&=\lambda\langle a^\dagger a\rangle_\rho +(1-\lambda)\langle b^\dagger b\rangle_\sigma+\sqrt{\lambda(1-\lambda)}\left(\langle a\rangle_\rho \langle b^\dagger\rangle_\sigma+\langle a^\dagger\rangle_\rho \langle b\rangle_\sigma\right)\\&=\lambda\langle a^\dagger a\rangle_\rho +(1-\lambda)\langle b^\dagger b\rangle_\sigma+2\sqrt{\lambda(1-\lambda)}\Re\left(\langle a\rangle_\rho \langle b^\dagger\rangle_\sigma\right)\,.
\ee
The formula~\eqref{enweak} follows from a similar calculation. 
\end{proof}

\begin{lemma_app}\label{teobeamexp}
Let $i,j\in\N$. It holds that
	\begin{equation}\label{espU1}
U_{\lambda}^{(SE)}\ket{i}_S\ket{j}_E=\sum_{m=0}^{i+j}c_m^{(i,j)}(\lambda)\ket{i+j-m}_S\ket{m}_E \,,
\end{equation}
where for all $m=0\text{, }1\text{, \ldots , }i+j$ we have
\bb\label{cijm}
c_m^{(i,j)}(\lambda)\coloneqq\frac{1}{\sqrt{i!j!}}\sum_{k=\max(0,m-j)}^{\min(i,m)}(-1)^k\binom{i}{k}\binom{j}{m-k}\lambda^{\frac{i+m-2k}{2}}(1-\lambda)^{\frac{j+2k-m}{2}}\sqrt{m!(i+j-m)!}\,.
\ee
As a consequence, for all $i,i',j,j'\in\N$ it holds that:
\bb
\Phi_{\lambda,\ketbraasub{j}{j'}_E}\left(\ketbraasub{i}{i'}_S\right)=\sum_{m=0}^{\min(i+j,i'+j')}c_{m}^{(i,j)}(\lambda)c_{m}^{(i',j')}(\lambda)\ketbraa{i+j-m}{i'+j'-m}_S\,
\ee
\bb\label{formula_wc}
\tilde{\Phi}^{\text{wc}}_{\lambda,\ketbraasub{j}{j'}_E}\left(\ketbraa{i}{i'}_S\right)=\sum_{m=0}^{\min(i+j,i'+j')}c_{i+j-m}^{(i,j)}(\lambda)c_{i'+j'-m}^{(i',j')}(\lambda)\ketbraa{i+j-m}{i'+j'-m}_E\,
\ee
where 
\bb
\Phi_{\lambda,\ketbraasub{j}{j'}_E}\left(\ketbraa{i}{i'}_S\right)&\coloneqq\Tr_E\left[U_{\lambda}^{(SE)}\ketbraa{i}{i'}_S\otimes\ketbraa{j}{j'}_E\left(U_{\lambda}^{(SE)}\right)^\dagger\right]\,,\\ 
\tilde{\Phi}^{\text{wc}}_{\lambda,\ketbraasub{j}{j'}_E}\left(\ketbraa{i}{i'}_S\right)&\coloneqq\Tr_S\left[U_{\lambda}^{(SE)}\ketbraa{i}{i'}_S\otimes\ketbraa{j}{j'}_E\left(U_{\lambda}^{(SE)}\right)^\dagger\right]\,,
\ee
\end{lemma_app}

\begin{proof}
By exploiting Lemma~\ref{lemmatrasf}, it holds that
\bb
&U_{\lambda}^{(SE)}\ket{i}_S\ket{j}_E=\frac{1}{\sqrt{i!j!}}U_{\lambda}^{(SE)} (a^\dagger)^i(b^\dagger)^j\ket{0}_S\ket{0}_E\\&=\frac{1}{\sqrt{i!j!}}\left(U_{\lambda}^{(SE)} a^\dagger \left(U_{\lambda}^{(SE)}\right)^\dagger\right)^i\left(U_{\lambda}^{(SE)} b^\dagger\left(U_{\lambda}^{(SE)}\right)^\dagger\right)^j\ket{0}_S\ket{0}_E\\&=\frac{1}{\sqrt{i!j!}}(\sqrt{\lambda}a^\dagger -\sqrt{1-\lambda}b^\dagger)^i(\sqrt{1-\lambda}a^\dagger +\sqrt{\lambda}b^\dagger)^j\ket{0}_S\ket{0}_E\\&=\frac{1}{\sqrt{i!j!}}\sum_{k=0}^{i}\sum_{w=0}^{j} (-1)^k \binom{i}{k} \binom{j}{w} \lambda^{\frac{i+w-k}{2}}(1-\lambda)^{\frac{j+k-w}{2}} \sqrt{(k+w)!(i+j-k-w)!} \ket{i+j-k-w}_S\ket{k+w}_E\\&=\sum_{m=0}^{i+j}\frac{1}{\sqrt{i!j!}}\sum_{k=\max(0,m-j)}^{\min(i,m)}(-1)^k\binom{i}{k}\binom{j}{m-k}\lambda^{\frac{i+m-2k}{2}}(1-\lambda)^{\frac{j+2k-m}{2}}\sqrt{m!(i+j-m)!}\ket{i+j-m}_S\ket{m}_E\\&=\sum_{m=0}^{i+j}c_m^{(i,j)}(\lambda)\ket{i+j-m}_S\ket{m}_E\,.
\ee
\end{proof}

\begin{definition_app}
For every trace class operator $T$ on $L^2(\mathbb{R}^m)$, its \emph{characteristic function} $\chi_{T}:\mathbb{C^m}\rightarrow\mathbb{C}$ is defined by
\begin{equation}
	\chi_{T}(z)\coloneqq \Tr\left[T\,D(z)\right]\qquad\forall\, z=(z_1,z_2,\ldots,z_m)\in\mathbb{C}^m \,,
\end{equation}
	where $$D(z)\coloneqq\exp{\left[\sum_{i=0}^m \left(z_i a_i^\dagger-z_i^\ast a_i\right)\right] }$$ is the displacement operator, with $a_i$ being the annihilation operator corresponding to the i-th mode.
\end{definition_app}
Conversely, it turns out that every trace class operator $T$ can be reconstructed from its characteristic functions $\chi_{T}$ via the following identity~\cite{HOLEVO-CHANNELS-2,BUCCO}:
\begin{equation}
T=\int_{\mathbb{C}} \frac{d^m z}{\pi^m}\text{ }D(-z)\chi_{T}(z)\,.
\end{equation}

\begin{lemma_app}
The action of a {BS} of transmissivity $\lambda$ on the a two-mode state of the form $\rho\otimes\sigma$ can be cast in the language of characteristic functions as
\begin{equation}\label{caract}
\chi_{U_\lambda^{(SE)}  \rho\otimes\sigma  \left(U_\lambda^{(SE)}\right)^\dagger} \left(z,w\right)=\chi_{\rho}\left(\sqrt{\lambda}z-\sqrt{1-\lambda}w\right)\chi_{\sigma}\left(\sqrt{1-\lambda}z+\sqrt{\lambda}w\right)\quad\forall\, z,w\in\mathbb{C}\,.
\end{equation}
Consequently, it holds that
\begin{equation}\label{caract3}
\chi_{\Phi_{\lambda,\sigma}(\rho)} \left(z\right)=\chi_{\rho}\left(\sqrt{\lambda}z\right)\chi_{\sigma}\left(\sqrt{1-\lambda}z\right)\quad\forall\, z\in\mathbb{C}\,,
\end{equation}
\begin{equation}\label{caract_weak}
\chi_{\tilde{\Phi}^{\text{wc}}_{\lambda,\sigma}(\rho)} \left(z\right)=\chi_{\rho}\left(-\sqrt{1-\lambda}w\right)\chi_{\sigma}\left(\sqrt{\lambda}w\right)\quad\forall\, w\in\mathbb{C}\,.
\end{equation}
\end{lemma_app}
\begin{proof}
Denote the displacement operators on $\HH_S$ and $\HH_E$ respectively as $D_S(z)$ and $D_E(z)$.
\eqref{trasfheisa} and~\eqref{trasfheisb} imply that
\bb
 \left(U_\lambda^{(SE)}\right)^\dagger D_S(z)\, U_\lambda^{(SE)}&=D_S(\sqrt{\lambda}z)\, D_E(\sqrt{1-\lambda}z)\\
 \left(U_\lambda^{(SE)}\right)^\dagger D_E(w)\, U_\lambda^{(SE)}&=D_S(-\sqrt{1-\lambda}w)\, D_E(\sqrt{\lambda}w)
\ee
Consequently, we obtain
\begin{equation}
 \left(U_\lambda^{(SE)}\right)^\dagger D_S(z)\, D_E(w)\, U_\lambda^{(SE)}=D_S(\sqrt{\lambda}z-\sqrt{1-\lambda}w)\, D_E(\sqrt{1-\lambda}z+\sqrt{\lambda}w)\,.
\end{equation}
Hence, it holds that
\bb
\chi_{U_\lambda^{(SE)}  \rho\otimes\sigma  \left(U_\lambda^{(SE)}\right)^\dagger} \left(z,w\right)&=\Tr_{SE}\left[U_\lambda^{(SE)}  \rho\otimes\sigma  \left(U_\lambda^{(SE)}\right)^\dagger \,D_S(z)\, D_E(z) \right] \\&=\Tr_{SE}\left[\rho\otimes\sigma\,D_S(\sqrt{\lambda}z-\sqrt{1-\lambda}w)\, D_E(\sqrt{1-\lambda}z+\sqrt{\lambda}w) \right] \\&=\chi_{\rho}\left(\sqrt{\lambda}z-\sqrt{1-\lambda}w\right)\chi_{\sigma}\left(\sqrt{1-\lambda}z+\sqrt{\lambda}w\right)\,.
\ee
\eqref{caract3} and~\eqref{caract_weak} follow from the fact that for all $\rho_{SE}\in\mathfrak{S}(\HH_S\otimes\HH_E)$ it holds that
\bb
    \chi_{\Tr_S\rho_{SE}}(w)=\chi_{\rho_{SE}}(z=0,w)\,,
    \chi_{\Tr_E\rho_{SE}}(z)=\chi_{\rho_{SE}}(z,w=0)\,.
\ee
\end{proof}

\begin{lemma_app}\label{lemma_invert}
For all $\lambda\in[0,1]$ and all $\rho$, $\sigma$ single-mode states, it holds that
\begin{equation}\label{invert2}
\Phi_{\lambda,\sigma}(\rho)=\Phi_{1-\lambda,\,\rho}(\sigma)\,.
\end{equation}
\end{lemma_app}
\begin{proof}
The characteristic function associated with an output of a general attenuator can be expressed as in~\eqref{caract3} and hence it holds that
\begin{equation}
	\chi_{\Phi_{\lambda,\sigma}(\rho)} \left(z\right)=\chi_{\Phi_{1-\lambda,\rho}(\sigma)} \left(z\right)\quad\forall\, z\in\mathbb{C}\,.
\end{equation}
Since quantum states are in one-to-one correspondence with characteristic functions,~\eqref{invert2} is proved.
\end{proof}

\begin{lemma_app}\label{thermal_composition}
It holds that
\begin{equation}\label{compos}
\Phi_{\lambda_1\lambda_2,\tau_N}=\Phi_{\lambda_1,\tau_N}\circ\Phi_{\lambda_2,\tau_N}\quad\quad\forall\,\lambda_1\text{, }\lambda_2\in[0,1]\text{, }N>0\,.
\end{equation}
\end{lemma_app}
\begin{proof}
Since the characteristic function of $\tau_\nu$ is $\chi_{\tau_\nu}(z)=e^{-\left(\nu+\frac{1}{2}\right)|z|^2}$~\cite[Proof of formula (12.33)]{HOLEVO-CHANNELS-2}, the mapping~\eqref{caract3} becomes $\chi_{\Phi_{\lambda,\tau_N}(\rho)} \left(z\right)=\chi_{\rho}\left(\sqrt{\lambda}z\right)e^{-(1-\lambda)\left(N+\frac{1}{2}\right)|z|^2}$.
Hence, the characteristic functions of $\Phi_{\lambda_1,\tau_N}\circ\Phi_{\lambda_2,\tau_N}(\rho)$ is
\bb
\chi_{\Phi_{\lambda_1,\tau_N}\circ\Phi_{\lambda_2,\tau_N}(\rho)} \left(\mu\right)&=\chi_{\Phi_{\lambda_2,\tau_N}(\rho)}\left(\sqrt{\lambda_1}\mu\right)\exp\left[-(1-\lambda_1)\left(N+\frac{1}{2}\right)|\mu|^2\right]\\&=\chi_{\rho}\left(\sqrt{\lambda_1\lambda_2}\mu\right)\exp\left[-(\lambda_1-\lambda_1\lambda_2)\left(N+\frac{1}{2}\right)|\mu|^2\right] \exp\left[-(1-\lambda_1)\left(N+\frac{1}{2}\right)|\mu|^2\right]\\&=\chi_{\rho}\left(\sqrt{\lambda_1\lambda_2}\mu\right)\chi_{\tau_N}\left(\sqrt{1-\lambda_1\lambda_2}\mu\right)=\chi_{\Phi_{\lambda_1\lambda_2,\tau_N}(\rho)} \left(\mu\right)\quad\forall\, \rho\in\mathfrak{S}(\HH_S)\text{ , }\mu\in\mathbb{C}\,.
\ee
Since quantum states are in one-to-one correspondence with characteristic functions, this concludes the proof.
\end{proof}
\begin{lemma_app}\label{LemmaMaster}
By defining $\rho_N(t)\coloneqq\Phi_{\exp(-t),\tau_N}(\rho)$ with $t\in[0,\infty)$, it holds that
\begin{equation*}
    \frac{d}{dt}\rho_N(t)=(N+1)\left[a\rho_N(t) a^\dagger-\frac{1}{2}\{\rho_N(t),a^\dagger a\}\right]+N\left[a^\dagger\rho_N(t) a-\frac{1}{2}\{\rho_N(t),a a^\dagger \}\right]\,.
\end{equation*}
\end{lemma_app}
\begin{proof}
The identity~\eqref{compos} becomes
\begin{equation}\label{dynsemi}
\Phi_{\exp(-t_1-t_2),\tau_N}=\Phi_{\exp(-t_1),\tau_N}\circ\Phi_{\exp(-t_2),\tau_N}\quad\quad\forall\, t_1,t_2\ge0\,.
\end{equation}
In addition, for $t=0$ one obtains the identity superoperator $\Phi_{1,\tau_N}=I$. As a consequence, one arrives at the \emph{master equation}
\begin{equation}\label{Master}
	\frac{d}{dt}\rho_N(t)=\mathcal{L}_N[\rho_N(t)]\,,
\end{equation}
where $\mathcal{L}_N$, called \emph{Lindbladian superoperator}, is defined by
\begin{equation}\label{Limb}
\mathcal{L}_N\coloneqq \lim\limits_{\varepsilon\rightarrow0^+}\frac{\Phi_{\exp(-\varepsilon),\tau_N}-I}{\varepsilon}\,.
\end{equation}
Now, let us prove that the Lindbladian superoperator is given by
	\begin{equation}\label{Liouv}
	\mathcal{L}_N[{\Theta}]=(N+1)\left[a\Theta a^\dagger-\frac{1}{2}\{\Theta,a^\dagger a\}\right]+N\left[a^\dagger\Theta a-\frac{1}{2}\{\Theta,a a^\dagger \}\right] \,.
	\end{equation}
	By definition, we have that
	\begin{equation}\label{PhiEps}
	\Phi_{\exp(-\varepsilon),\tau_N}[{\Theta}]=\Tr_E\left(U_{\exp(-\varepsilon)}^{(SE)}{\Theta}\otimes\tau_N\left(U_{\exp(-\varepsilon)}^{(SE)}\right)^\dagger\right)\,.
	\end{equation}
	By using that
	\bb
	\arccos( e^{-\varepsilon/2})&=\sqrt{\varepsilon}+O(\varepsilon^{3/2})\,,
	\ee
	one can verify that
	\bb
	  U_{\exp(-\varepsilon)}^{(SE)}&=\exp\left(\arccos( e^{-\varepsilon/2})(a^\dagger b-a b^\dagger)\right)=\mathbb{1}+\sqrt{\varepsilon}(a^\dagger b-ab^\dagger)+\frac{\varepsilon}{2}(a^\dagger b-ab^\dagger)^2+O(\varepsilon^{3/2})\,.\label{noth}
	\ee
	Inserting~\eqref{noth} into~\eqref{PhiEps} yields
	\begin{equation}\label{espress}
	\Phi_{\exp(-\varepsilon),\tau_N}[{\Theta}]=\Theta+\varepsilon\left\{ (N+1)\left[a\Theta a^\dagger-\frac{1}{2}\{\Theta,a^\dagger a\}\right]+N\left[a^\dagger\Theta a-\frac{1}{2}\{\Theta,a a^\dagger \}\right] \right\}+O(\varepsilon^{3/2})\,,
	\end{equation}
	where we used $\Tr_E\left[\tau_N b^\dagger b\right]=N\,,$ and $\Tr_E\left[\tau_N b b^\dagger\right]=N+1\,.$
	By substituting the expression~\eqref{espress} in $\mathcal{L}_N[\Theta]=\lim\limits_{\varepsilon\rightarrow0^+}\frac{\Phi_{\exp(-\varepsilon),\tau_N}[{\Theta}]-\Theta}{\varepsilon}$, we finally arrive at~\eqref{Liouv}.
\end{proof}

\begin{lemma_app}\label{lemmaweakp}
Let $\lambda\in[0,1]$ and let $\sigma\in\mathfrak{S}(\HH_E)$. The following superoperator identity holds
\begin{equation}
\tilde{\Phi}^{\text{wc}}_{\lambda,\sigma}=\mathcal{V}\circ\Phi_{1-\lambda,\mathcal{V}(\sigma)}\text{ ,}
\end{equation}
where $\mathcal{V}$ is the phase space inversion superoperator defined in~\eqref{tV}. 
\end{lemma_app}
\begin{proof}
To distinguish the phase space inversion operators which act on the environment and on the system, let us denote them respectively as $V_E\coloneqq (-1)^{b^\dagger b}$ and $V_S\coloneqq (-1)^{a^\dagger a}$.
From the identity $(-1)^{a^\dagger a} a (-1)^{a^\dagger a} =-a$, one obtains
\begin{equation}
    V_SD_S(z){V_S}^\dagger=D_S(-z)\,.
\end{equation}
As a consequence, in the language of characteristic functions, we obtain
\bb
\chi_{V_S\Phi_{1-\lambda,V_E \sigma {V_E}^\dagger}(\rho){V_S}^\dagger} \left(z\right)&=\Tr_S\left[V_S\Phi_{1-\lambda,V_E \sigma {V_E}^\dagger}(\rho){V_S}^\dagger D_S(z)\right]=\Tr_{SE}\left[\left({U^{(SE)}_{1-\lambda}}\right)^\dagger {V_S}^\dagger D_S(z)V_SU_{1-\lambda}^{(SE)}\, \rho\otimes V_E \sigma {V_E}^\dagger\right]\\&=\Tr_{SE}\left[\left(U_{1-\lambda}^{(SE)}\right)^\dagger D_S(-z)U_{1-\lambda}^{(SE)}\, \rho\otimes V_E \sigma {V_E}^\dagger\right]\\&=\Tr_{SE}\left[D_S(-\sqrt{1-\lambda} z)D_E(-\sqrt{\lambda}z)\, \rho\otimes V_E \sigma {V_E}^\dagger\right]\\&=\Tr_{S}\left[D_S(-\sqrt{1-\lambda} z)\rho\right] \Tr_{E}\left[D_E(\sqrt{\lambda}z)\sigma \right] =\chi_{\rho}\left(-\sqrt{1-\lambda}z\right)\chi_{\sigma}\left(\sqrt{\lambda}z\right)=\chi_{\tilde{\Phi}^{\text{wc}}_{\lambda,\sigma}(\rho)} \left(z\right)\,,
\ee
where we used~\eqref{caract_weak}. This concludes the proof.
\end{proof}

\begin{lemma_app}\label{lowcap}
	If $\sigma$ is diagonal in Fock basis i.e.\ $$\sigma\coloneqq\sum_{i=0}^{\infty}q_i\ketbra{i}\,,$$then for all $N>0$ and $\lambda\in[0,1]$ the following inequalities hold:
	\bb
	    Q\left(\Phi_{\lambda,\sigma},N\right)&\ge f(N,\lambda)\,,\\
	    C_{\text{ea}}\left(\Phi_{\lambda,\sigma},N\right)&\ge g(N)+f(N,\lambda)\,,
	\ee
	where 
	\begin{equation}\label{def_f}
	    f(N,\lambda)\coloneqq H\left(\left\{\sum_{i=0}^{\infty}q_iP_l(N,i,1-\lambda)\right\}_{l\in\mathbb{N}}\right)-H\left(\left\{\sum_{i=0}^{\infty}q_iP_l(N,i,\lambda)\right\}_{l\in\mathbb{N}}\right)-H\left(\{q_i\}_{i\in\mathbb{N}}\right)\,,
	\end{equation}
	 $\left\{P_l(N,i,\lambda)\right\}_{l\in\mathbb{N}}$ is expressed in~\eqref{simpler}, $g(\cdot)$ is expressed in~\eqref{bosonicent}, and $H(\cdot)$ denotes the Shannon entropy.
\end{lemma_app}
\begin{proof}
Lemma~\ref{subadd} implies that
\bb\label{coh_diagonal}
I_{\text{coh}}\left(\Phi_{\lambda,\sigma},\tau_N\right)&=S\left(\Phi_{\lambda,\sigma}(\tau_N)\right)-S\left(\tilde{\Phi}_{\lambda,\sigma}(\tau_N)\right) \ge S\left(\Phi_{\lambda,\sigma}(\tau_N)\right)-S\left(\tilde{\Phi}^{\text{wc}}_{\lambda,\sigma}(\tau_N)\right)-S(\sigma)\,.
\ee
Furthermore, Lemma~\ref{lemmaweakp} and the fact that von Neumann entropy is invariant under unitary transformations imply that
\begin{equation}\label{entropy_wc}
S\left(\tilde{\Phi}^{\text{wc}}_{\lambda,\sigma}(\tau_N)\right)=S\left(\Phi_{1-\lambda,\mathcal{V}(\sigma)}(\tau_N)\right)\text{ .}
\end{equation}
Since $\sigma$ is diagonal in Fock basis, it holds $\mathcal{V}(\sigma)=\sigma$. By substituting~\eqref{entropy_wc} in~\eqref{coh_diagonal}, we obtain
$$I_{\text{coh}}\left(\Phi_{\lambda,\sigma},\tau_N\right)\ge S\left(\Phi_{\lambda,\sigma}(\tau_N)\right)-S\left(\Phi_{1-\lambda,\sigma}(\tau_N)\right)-S(\sigma)\,.$$
Moreover, by using~\eqref{invert2}, we arrive at
\begin{equation}\label{key}
I_{\text{coh}}\left(\Phi_{\lambda,\sigma},\tau_N\right)\ge S\left(\Phi_{1-\lambda,\tau_N}(\sigma)\right)-S\left(\Phi_{\lambda,\tau_N}(\sigma)\right)- S(\sigma)\,.
\end{equation}
Moreover,~\eqref{defP} implies that
\bb
S\left(\Phi_{\lambda,\tau_N}(\sigma)\right)&=S\left(\sum_{i=0}^{\infty}q_i\Phi_{\lambda,\tau_N}(\ketbra{i})\right)=S\left(\sum_{l,i=0}^{\infty}q_iP_l(N,i,\lambda)\ketbra{l}\right)=H\left(\left\{\sum_{i=0}^{\infty}q_iP_l(N,i,\lambda)\right\}_{l\in\mathbb{N}}\right)
\ee
and we have \begin{equation}\label{twoeq}
S(\sigma)=H\left(\{q_i\}_{i\in\mathbb{N}}\right)\,.
\end{equation}
By substituting in~\eqref{key}, we obtain
\begin{equation}
I_{\text{coh}}\left(\Phi_{\lambda,\sigma},\tau_N\right)\ge f(N,\lambda)\,.
\end{equation}
As a consequence, it holds that
\bb
Q\left(\Phi_{\lambda,\sigma},N\right)&\ge I_{\text{coh}}\left(\Phi_{\lambda,\sigma},\tau_N\right)\ge f(N,\lambda)\,,\\
C_{\text{ea}}\left(\Phi_{\lambda,\sigma},N\right)&\ge g(N)+ I_{\text{coh}}\left(\Phi_{\lambda,\sigma},\tau_N\right)\ge g(N)+f(N,\lambda)\,.
\ee
\end{proof}

{\begin{lemma_app}\label{lemmakraus}
\emph{(Kraus representation of the thermal attenuator via explicit formula of {BS})}\\ For all $\lambda\in[0,1],\nu\ge 0$, the thermal attenuator $\Phi_{\lambda,\tau_\nu}$ admits the following Kraus representation:
\begin{equation}
    \Phi_{\lambda,\tau_\nu}(\rho)=\sum_{k,m=0}^\infty \tilde{M}_{k,m}\rho \tilde{M}_{k,m}^\dagger\,,
\end{equation}
where 
\bb
    \tilde{M}_{k,m} =\sqrt{\frac{\nu^k}{(\nu+1)^{k+1}}} \sum_{l=max(m-k,0)}^\infty c_m^{(l,k)}(\lambda) \ketbraa{l+k-m}{l}\,,
\ee
with the coefficients $c_m^{l,k}(\lambda)$ being defined in~\eqref{cijm}.
\end{lemma_app}
\begin{proof}
It holds that
\bb
\Phi_{\lambda,\tau_\nu}=\Tr_E\left[U_{\lambda}^{(SE)}\rho\otimes\tau_\nu\left(U_{\lambda}^{(SE)}\right)^\dagger\right]=\sum_{k,m=0}^\infty \frac{\nu^k}{(\nu+1)^{k+1}}\bra{m}_E U_{\lambda}^{(SE)}\ket{k}_E\rho \left(\bra{m}_E U_{\lambda}^{(SE)}\ket{k}_E\right)^\dagger\,.
\ee
Hence, by defining
\begin{equation}
    \tilde{M}_{k,m}\coloneqq \sqrt{\frac{\nu^k}{(\nu+1)^{k+1}}} \bra{m}_E U_{\lambda}^{(SE)}\ket{k}_E\,,
\end{equation}
we have found a Kraus representation of $\Phi_{\lambda,\tau_\nu}$. By using the explicit formula of {BS} in~\eqref{espU1}, it holds that
\begin{equation}
    \tilde{M}_{k,m}=\sqrt{\frac{\nu^k}{(\nu+1)^{k+1}}} \sum_{l=0}^\infty\bra{m}_E U_{\lambda}^{(SE)}\ket{l}_S\ket{k}_E\bra{l}_S=\sqrt{\frac{\nu^k}{(\nu+1)^{k+1}}} \sum_{l=max(m-k,0)}^\infty c_m^{(l,k)}(\lambda) \ketbraa{l+k-m}{l}_S\,,
\end{equation}
\end{proof}

\begin{lemma_app}\label{lemmakraus2}
\emph{(Kraus representation of the thermal attenuator via master equation trick)}\\ For all $\lambda\in[0,1],\nu\ge 0$, the thermal attenuator $\Phi_{\lambda,\tau_\nu}$ admits the following Kraus representation:
\begin{equation}
    \Phi_{\lambda,\tau_\nu}(\rho)=\sum_{k,m=0}^\infty M_{k,m}\rho M_{k,m}^\dagger\,,
\end{equation}
where 
\bb
    M_{k,m}=\sqrt{\frac{\nu^k(\nu+1)^m(1-\lambda)^{m+k}}{k!m![(1-\lambda)\nu+1]^{m+k+1}}}(a^\dagger)^k\left(\frac{\sqrt{\lambda}}{(1-\lambda)\nu+1}\right)^{a^\dagger a}a^m\,.
\ee
\end{lemma_app}
\begin{proof}
It immediately follows by substituting $\lambda=\exp{(-t)}$ and~\eqref{fF_ambda} into~\eqref{master_sol}.
\end{proof}

\begin{remark}\label{lemmakraus3}
The Kraus representation derived via master equation trick in Lemma~\ref{lemmakraus2} can allow one to obtain significantly simpler expressions for the action of the thermal attenuator on a generic operator, than that obtainable by exploiting the Kraus representation derived via explicit formula of {BS} in Lemma~\ref{lemmakraus}. For example, let us calculate the action of the thermal attenuator on a operator of the form $\ketbraa{n}{i}$, where $\ket{n}$ and $\ket{i}$ are Fock states.
The Kraus representation derived via explicit formula of {BS} in Lemma~\ref{lemmakraus} easily yields:
\begin{equation}\label{expansion_ni}
        \Phi_{\lambda,\tau_\nu}(\ketbraa{n}{i})=\sum_{l=\max(i-n,0)}^\infty     f_{n,i,l}(\lambda,\nu)\ketbraa{l+n-i}{l}\,,
\end{equation}
where
\begin{equation}\label{compl_series}
    f_{n,i,l}(\lambda,\nu)\coloneqq \sum_{k=\max(l-i,0)}^\infty \frac{\nu^k}{(\nu+1)^{k+1}}c_{k+i-l}^{(n,k)}(\lambda) c_{k+i-l}^{(i,k)}(\lambda)\,,
\end{equation}
with the coefficients $c_m^{l,k}(\lambda)$ being defined in~\eqref{cijm}.
While, the Kraus representation derived via master equation trick in Lemma~\ref{lemmakraus2} yields the same expansion in~\eqref{expansion_ni} (as expected), but with a much simpler expression for the coefficients $f_{n,i,l}(\lambda,\nu)$ (see proof below):
\begin{equation}\label{simplified_compl_series}
    f_{n,i,l}(\lambda,\nu)=\sum_{m=\max(i-l,0)}^{\min(n,i)}\frac{\sqrt{n!i!l!(l+n-i)!}\nu^{l+m-i}(\nu+1)^m(1-\lambda)^{2m+l-i}\lambda^{\frac{n+i-2m}{2}}}{(n-m)!(i-m)!m!(l+m-i)!\left((1-\lambda)\nu+1\right)^{l+n+1}}\,.
\end{equation}
Hence, our trick allows one to simplify the expression of $\Phi_{\lambda,\tau_\nu}(\ketbraa{n}{i})$, by finding a closed formula for the sum of the complicated series in~\eqref{compl_series}.
\end{remark}
\begin{proof}[Proof of~\eqref{simplified_compl_series}]
Lemma~\ref{lemmakraus2} implies that $\Phi_{\lambda,\tau_\nu}(\ketbraa{n}{i})=\sum_{k,m=0}^\infty M_{k,m} \ketbraa{n}{i} M_{k,m}^\dagger\,,$. For $m>n$ it holds that $M_{k,m}\ket{n}=0$, otherwise for $m\le n$ it holds that
\begin{equation}
    M_{k,m}\ket{n}=\frac{1}{(n-m)!}\sqrt{\frac{n!(n-m+k)!}{k!m!}}\sqrt{\frac{\nu^k (\nu+1)^m (1-\lambda)^{k+m}\lambda^{n-m}}{\left((1-\lambda)\nu+1\right)^{k-m+1+2n}}}\ket{n-m+k}\,.
\end{equation}
Consequently, we conclude that
\bb
    \Phi_{\lambda,\tau_\nu}(\ketbraa{n}{i})&=\sum_{k=0}^{\infty}\sum_{m=0}^{\min(n,i)}\frac{\sqrt{n!(n-m+k)!i!(i-m+k)!}\nu^k(\nu+1)^m (1-\lambda)^{k+m}\lambda^{\frac{n+i-2m}{2}}}{(n-m)!(i-m)!k!m!\left((1-\lambda)\nu+1\right)^{k-m+1+n+i}}\ketbraa{n-m+k}{i-m+k}\\&=\sum_{l=\max(i-n,0)}^\infty \sum_{m=\max(i-l,0)}^{\min(n,i)}\frac{\sqrt{n!i!l!(l+n-i)!}\nu^{l+m-i}(\nu+1)^m(1-\lambda)^{2m+l-i}\lambda^{\frac{n+i-2m}{2}}}{(n-m)!(i-m)!m!(l+m-i)!\left((1-\lambda)\nu+1\right)^{l+n+1}}\ketbraa{l+n-i}{l}\,.
\ee
\end{proof}
}

\section{Example of noise attenuation protocol with one trigger signal}\label{example_noiseatt}
In this section we give a simple example where the noise attenuation protocol is advantageous. Let us suppose that the stationary environment state is the vacuum $\sigma_0=\ketbra{0}$. Hence, the original channel is the pure loss channel $\Phi_{\lambda,\ketbrasub{0}}$. The energy-constrained entanglement-assisted capacity of $\Phi_{\lambda,\ketbrasub{0}}$ is given by~\eqref{pureloss}.
In addition, let us suppose that Alice sends only one trigger signal initialised in the Fock state $\ketbra{n}_{S_1}$ during step~2 of the noise attenuation protocol. Consequently, at the beginning of step~3 the environment state is 
\begin{equation}\label{sigmaraggln}
\sigma_{\lambda,n}\coloneqq\tilde{\Phi}^{\text{wc}}_{\lambda,\ketbrasub{0}}\left(\ketbra{n}\right)=\Tr_E\left[U_\lambda^{(S_1E)}  \ketbra{n}\otimes\ketbra{0}  \left(U_\lambda^{(S_1E)}\right)^\dagger\right]\,.
\end{equation}
Moreover, suppose that the trigger signal and the information-carrier signal are separated by a temporal interval $\delta t\ll t_E$ (we recall that $t_E$ is the time after which the thermalisation process resets the environment state into the stationary environment state). Consequently, we can state that the information-carrier signal is affected by the channel $\Phi_{\lambda,\sigma_{\lambda,n}}$.
We want to see whether this simple application of the noise attenuation protocol implies the energy-constrained entanglement-assisted capacity of the resulting channel $\Phi_{\lambda,\sigma_{\lambda,n}}$ to be larger than that of the original channel $\Phi_{\lambda,\ketbrasub{0}}$.
By using~\eqref{formula_wc}, one obtains
\bb
\sigma_{\lambda,n}& =\sum_{l=0}^{n} \mathcal{B}_l(n,1-\lambda)\ketbra{l}\,,
\ee
where we have introduced the binomial distribution $\mathcal{B}_l(n,p)\coloneqq\binom{n}{l}p^{l}(1-p)^{n-l}$.
Since $\sigma_{\lambda,n}$ is diagonal in Fock basis, we can apply Theorem~\ref{lowcap}. Consequently, it holds that
\begin{equation}\label{forz}
C_{\text{ea}}\left(\Phi_{\lambda,\sigma_{\lambda,n}},N\right)\ge z(\lambda,N,n)\,,
\end{equation}
where
\bb\label{def_z}
z(\lambda,N,n)&\coloneqq g(N)+ H\left(\left\{\sum_{i=0}^{n}\mathcal{B}_i(n,1-\lambda)P_l(N,i,1-\lambda)\right\}_{l}\right)- \\&-H\left(\left\{\sum_{i=0}^{n}\mathcal{B}_i(n,1-\lambda)P_l(N,i,\lambda)\right\}_{l}\right)-H\left(\left\{\mathcal{B}_i(n,1-\lambda)\right\}_{i}\right)\,.
\ee
From Fig.~\ref{figure_example} one can note that for $N=0.5$ and for all the values of $n$ depicted it holds that:
$$C_{\text{ea}}\left(\Phi_{\lambda,\sigma_{\lambda,n}},N\right)\ge z(\lambda,N,n)\ge C_{\text{ea}}\left(\Phi_{\lambda,\ketbrasub{0}},N\right)\,$$
for $\lambda$ sufficiently small. As a result, for these parameter choices this basic version of the noise attenuation protocol improves the performance of the entanglement-assisted communication. 
\begin{figure}[t]
	\centering
	 {\label{figur:1}\includegraphics[width=0.5\linewidth]{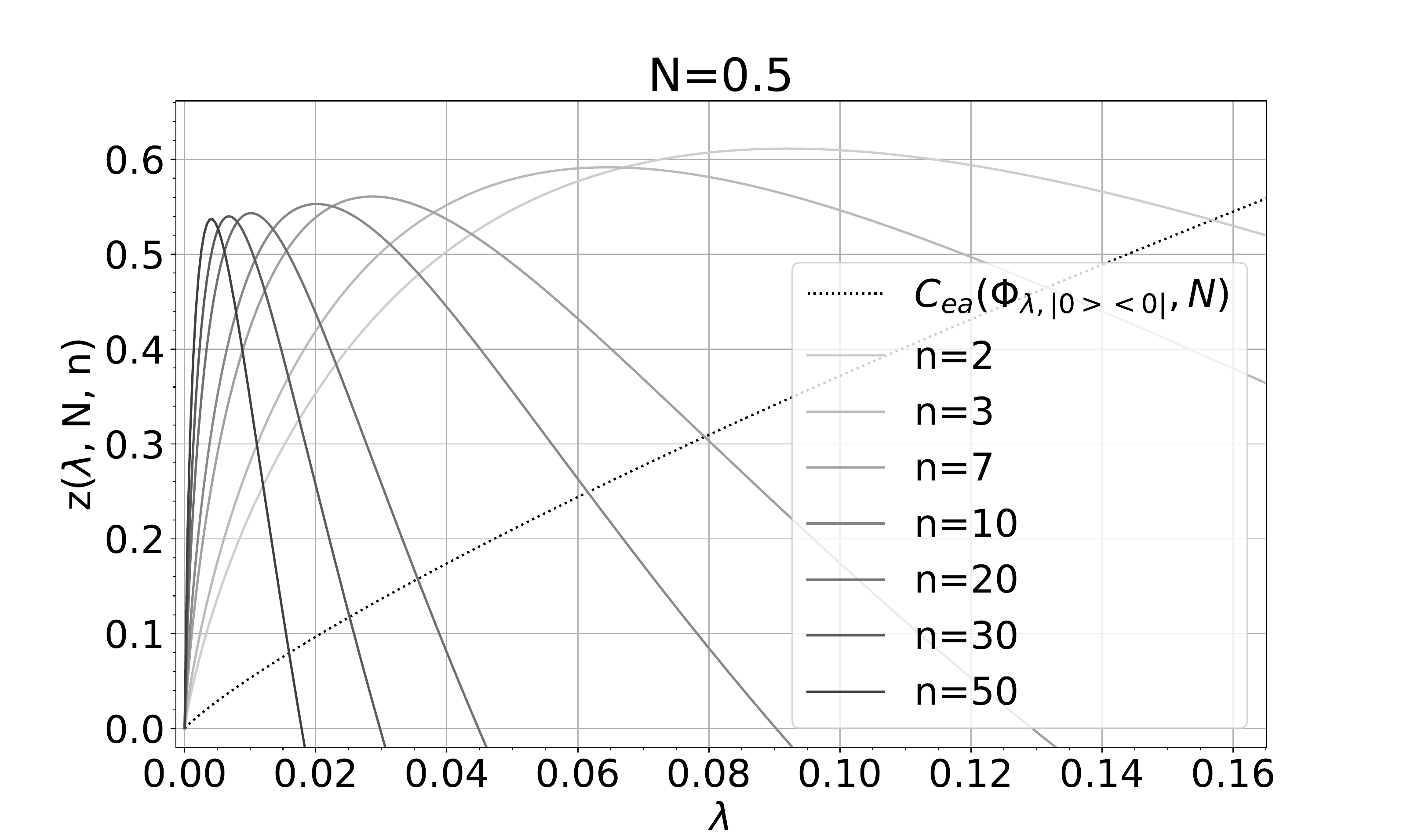}}
	 {\label{figur:2}\includegraphics[width=0.5\linewidth]{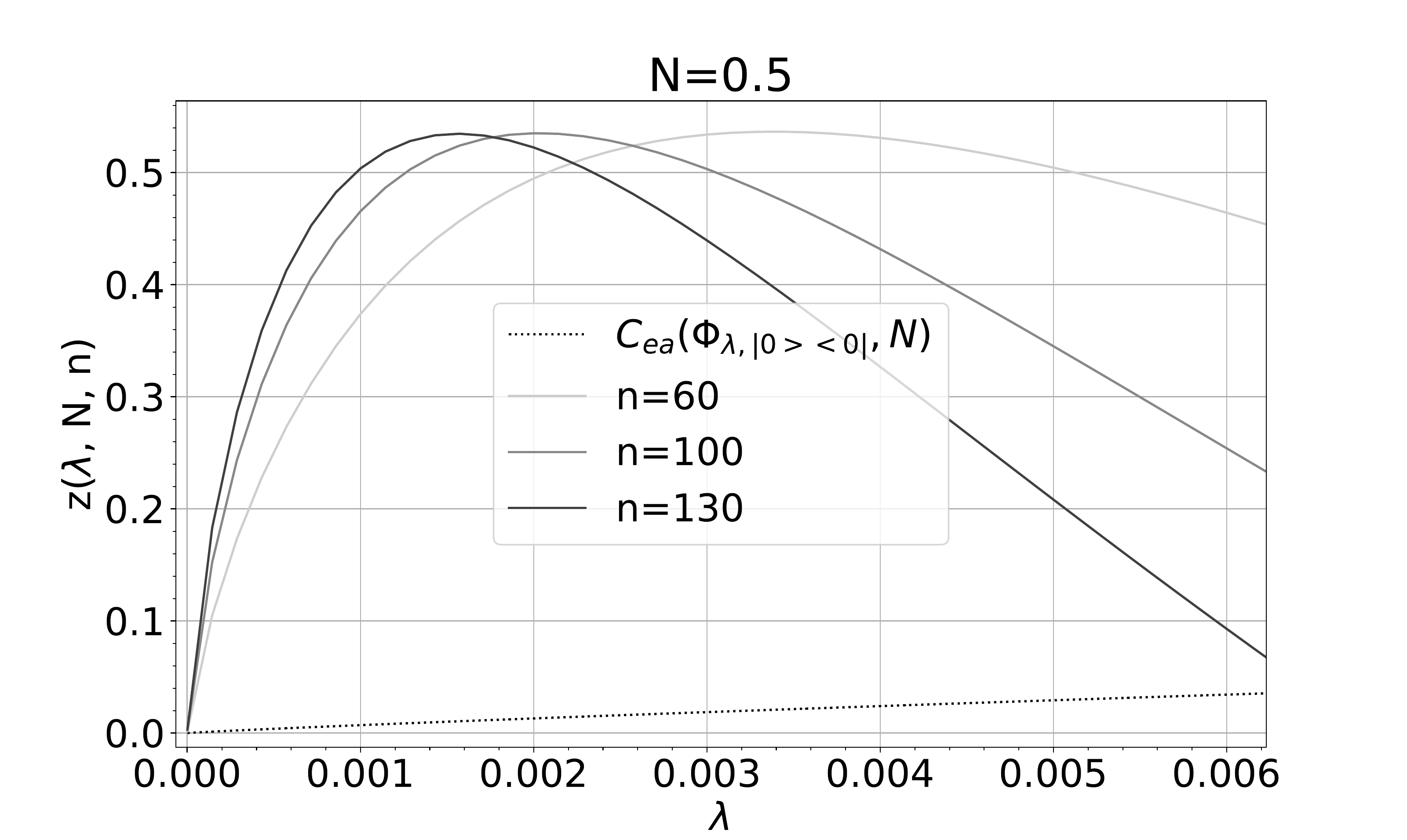}}
	\caption{The function $z(\lambda,N,n)$ plotted with respect to the variable $\lambda$ for $N=0.5$ and several values of $n$, where $z(\lambda,N,n)$ is defined in~\eqref{def_z}. The dotted line is $C_{\text{ea}}\left(\Phi_{\lambda,\ketbrasub{0}},0.5\right)$ as a function of $\lambda$ (see~\eqref{pureloss}).}
	\label{figure_example}
\end{figure}
Moreover, Fig.~\ref{figure_example} suggests also that for all $\lambda>0$ sufficiently small if $n$ is sufficiently large the quantity $C_{\text{ea}}\left(\Phi_{\lambda,\sigma_{\lambda,n}},N\right)$ is larger than a positive constant (independent of $\lambda$). This is quite interesting since the energy-constrained entanglement-assisted capacity of the original channel tends to $0$ when $\lambda$ approaches $0$. Therefore, this application of noise attenuation protocol improves significantly the performance of entanglement-assisted communication for $\lambda>0$ small. For instance, for $\lambda=0.002$ and $N=0.5$ one can numerically verify (by using~\eqref{forz} and~\eqref{pureloss}) that
$$\frac{C_{\text{ea}}\left(\Phi_{\lambda,\sigma_{\lambda,n=100}},N\right)}{C_{\text{ea}}\left(\Phi_{\lambda,\ketbrasub{0}},N\right)}\gtrsim 41\text{ .} $$

\section{Only one trigger signal}\label{Appendix_one_trigger}
In Theorem~\ref{theorem_lambda_small} we have discussed how to achieve an environment state close to $\ketbra{n_\lambda}$ for $\lambda>0$ sufficiently small, where $n_\lambda\in\N$ with $1/\lambda-1\le n_\lambda\le 1/\lambda$. This is done by sending just two trigger signals initialised in the state $\ket{n_\lambda,\lambda}_{S_1S_2}=U_{\frac{1}{1+\lambda}}^{(S_1S_2)}\ket{0}_{S_1}\ket{n_\lambda}_{S_2}$. We recall that achieving the environment state $\ketbra{n_\lambda}$ is important since the channel $\Phi_{\lambda,\ketbrasub{n_\lambda}}$ has strictly positive quantum capacity for $\lambda\in(0,1/2)$.

Note that for $\lambda=0$ and for all stationary environment states $\sigma_0$, after the interaction with a trigger signal initialised in the $n$-th Fock state, the environment state becomes
\begin{equation}
\tilde{\Phi}_{\lambda=0,\sigma_0}^{\text{wc}}\left(\ketbra{n} \right)=\Tr_{S_1}\left[U_{\lambda=0}^{(S_1E)}\ketbra{n}_{S_1}\otimes\sigma_0\left( {U_{\lambda=0}^{(S_1E)}}\right)^\dagger \right]=\ketbra{n}_{E}\,.
\end{equation}
Furthermore, for $\lambda>0$ sufficiently small we still have
\begin{equation}
\tilde{\Phi}_{\lambda,\sigma_0}^{\text{wc}}\left(\ketbra{n} \right)=\Tr_{S_1}\left[U_{\lambda}^{(S_1E)}\ketbra{n}_{S_1}\otimes\sigma_0\left( {U_{\lambda}^{(S_1E)}}\right)^\dagger \right]\sim\ketbra{n}_{E}\,.
\end{equation}
Therefore, a natural question may arise: \emph{instead of sending two signals in $\ket{n_\lambda,\lambda}_{S_1S_2} $, why do not we send just one trigger signal in $\ket{n_\lambda}_{S_1}$?} Below, we answer this question. 

After the interaction with the trigger signal $\ket{n_\lambda}_{S_1}$, the environment achieves the state $\tilde{\Phi}_{\lambda,\sigma_0}^{\text{wc}}\left(\ketbra{n_\lambda}\right)$. Since the state $\ketbra{n_\lambda}_{S_1}$ depends on $\lambda$, it may be that $\tilde{\Phi}_{\lambda,\sigma_0}^{\text{wc}}\left(\ketbra{n_\lambda} \right)$ and $\ketbra{n_\lambda}$ are not close in the limit $\lambda\rightarrow0^+$. We can show that this is indeed the case if the stationary environment state $\sigma_0$ is a thermal state $\tau_\nu$, as in the usual scheme of an optical fibre.

First, notice that~\eqref{weak_formula},~\eqref{scambio}, and the fact that $\Phi_{\lambda,\ketbrasub{n_\lambda}}\left(\tau_\nu \right)$ is diagonal in Fock basis (as guaranteed by~\eqref{defP}) imply that
\begin{equation}
    \tilde{\Phi}_{\lambda,\tau_\nu}^{\text{wc}}\left(\ketbra{n_\lambda} \right)=\mathcal{V}\circ\Phi_{1-\lambda,\tau_\nu}\left(\ketbra{n_\lambda} \right)=\mathcal{V}\circ\Phi_{\lambda,\ketbrasub{n_\lambda}}\left(\tau_\nu \right)=\Phi_{\lambda,\ketbrasub{n_\lambda}}\left(\tau_\nu \right)\,.
\end{equation}
Second, by exploiting the results and the notations used in the proof of Theorem~\ref{congl0}, it holds that:
\bb
\lim\limits_{\lambda\rightarrow0^+}\left\|\tilde{\Phi}_{\lambda,\tau_\nu}^{\text{wc}}\left(\ketbra{n_\lambda} \right)-\ketbra{n_\lambda}\right\|_1&=\lim\limits_{\lambda\rightarrow0^+}\left\|\Phi_{\lambda,\ketbrasub{n_\lambda}}(\tau_\nu)-\ketbra{n_\lambda}\right\|_1 \\&=\lim\limits_{\lambda\rightarrow0^+}\left\|T_{-n_\lambda}\Phi_{\lambda,\ketbrasub{n_\lambda}}(\tau_\nu)T_{-n_\lambda}^\dagger-\ketbra{0}\right\|_1 \\&=\lim\limits_{n\rightarrow +\infty}\left\|T_{-n}\Phi_{\frac{1}{n},\ketbra{n}}(\tau_\nu)T_{-n}^\dagger-\ketbra{0}\right\|_1 \\&=\left\|\sum_{k=-\infty}^\infty q_k(\nu,1)\ketbra{k}-\ketbra{0}\right\|_1=2\left[1-q_0(\nu,1)\right] \\&=2\left[1-e^{-(2\nu+1)}I_{0}\left(2\sqrt{\nu(\nu+1)}\right)\right]\ne 0\,,\label{dist_weak_fock}
\ee
where the probability distribution $\{q_k(N,c)\}_{k\in\mathbb{Z}}$ is expressed in~\eqref{prinf}.

\eqref{dist_weak_fock} implies that the output environment state $\sigma_{\lambda,\nu}\coloneqq\tilde{\Phi}_{\lambda,\tau_\nu}^{\text{wc}}\left(\ketbra{n_\lambda} \right)$ is not close to $\ketbra{n_\lambda}$ even in the limit $\lambda\rightarrow0^+$. Hence, $\Phi_{\lambda,\sigma_{\lambda,\nu}}$ is not close to $\Phi_{\lambda,\ketbrasub{n_\lambda}}$ in energy-constrained diamond norm for $\lambda\rightarrow0^+$. Consequently, we can not apply Lemma~\ref{continuityBound} to conclude that the channel $\Phi_{\lambda,\sigma_{\lambda,\nu}}$ has strictly positive quantum capacity for $\lambda>0$ sufficiently small. This is the reason why the idea of sending just one trigger signal does not work. 

\end{document}